\documentclass[pra,longbibliography,aps,superscriptaddress,nofootinbib,11pt,tightenlines,floatfix]{revtex4-2}

\UseRawInputEncoding
\usepackage[letterpaper,margin=.6in,tmargin=.8in,bmargin=.9in]{geometry}
\usepackage{graphicx}
\usepackage{amsmath,amsfonts,amsthm,amssymb}

\usepackage{graphicx}
\usepackage{multirow}
\usepackage{bm}
\usepackage{braket}
\usepackage{float}
\usepackage{algpseudocode}
\usepackage{enumerate}
\usepackage{enumitem}
\usepackage{caption}
\usepackage{babel}
\usepackage{placeins}
\usepackage{bbm}
\usepackage{times}
\usepackage{xcolor}
\usepackage{mathtools}
\usepackage[
  colorlinks=true,
  urlcolor=teal,
  linkcolor=teal,
  citecolor=teal
]{hyperref}
\usepackage{svg}
\algnewcommand\Input{\item[\textbf{Input:}]}
\algnewcommand\Output{\item[\textbf{Output:}]}

\newtheorem{theorem}{Theorem}

\newtheorem{definition}{Definition}

\newtheorem{lemma}{Lemma}

\newtheorem{problem}{Problem}
\newtheorem{proposition}{Proposition}

\newcounter{algorithm}
\renewcommand{\thealgorithm}{\arabic{algorithm}}
\newcounter{subroutine}
\renewcommand{\thesubroutine}{\arabic{subroutine}}

\newcommand{\subroutinecaption}[2]{%
  \refstepcounter{subroutine}%
  \noindent\textbf{Subroutine~\thesubroutine\quad #1}\label{#2}\par
  \vspace{0.4ex}
  \hrule
  \vspace{0.6ex}
}

\DeclareMathOperator{\BQP}{\mathsf{BQP}}

\renewcommand{\exp}{\ensuremath{\mathrm{exp}}}

\providecommand{\calC}{\ensuremath{\mathcal{C}}}
\providecommand{\calD}{\ensuremath{\mathcal{D}}}

\providecommand{\calL}{\ensuremath{\mathcal{L}}}

\providecommand{\calT}{\ensuremath{\mathcal{T}}}

\allowdisplaybreaks
\begin{document}
\preprint{APS/123-QED}
\bibliographystyle{apsrev4-2}
\title{Provable learning separation for predicting time-evolution of quantum many-body systems}
\author{Rahul Bandyopadhyay}
\affiliation{$\langle aQa^{L} \rangle$, Applied Quantum Algorithms, Leiden University, The Netherlands}
\affiliation{Leiden Institute of Advanced Computer Science (LIACS), Leiden University, Einsteinweg 55, 2333 CC, Leiden, The Netherlands}

\author{Riccardo Molteni}
\affiliation{$\langle aQa^{L} \rangle$, Applied Quantum Algorithms, Leiden University, The Netherlands}
\affiliation{Leiden Institute of Advanced Computer Science (LIACS), Leiden University, Einsteinweg 55, 2333 CC, Leiden, The Netherlands}

\author{Jens Eisert}
\affiliation{Dahlem Center for Complex Quantum Systems, Freie Universit{\"a}t Berlin, 14195 Berlin, Germany}
\affiliation{Fraunhofer Heinrich Hertz Institute, 10587 Berlin, Germany}
\affiliation{Helmholtz-Zentrum Berlin f{\"u}r Materialien und Energie, 14109 Berlin, Germany}

\author{Vedran Dunjko}
\affiliation{$\langle aQa^{L} \rangle$, Applied Quantum Algorithms, Leiden University, The Netherlands}
\affiliation{Leiden Institute of Advanced Computer Science (LIACS), Leiden University, Einsteinweg 55, 2333 CC, Leiden, The Netherlands}

\author{Sofiene Jerbi}
\affiliation{Dahlem Center for Complex Quantum Systems, Freie Universit{\"a}t Berlin, 14195 Berlin, Germany}
\affiliation{Helmholtz-Zentrum Berlin f{\"u}r Materialien und Energie, 14109 Berlin, Germany}

\begin{abstract}
     Given that quantum computers are naturally suited to simulate the behavior of quantum many-body systems, an immediate question arises: can one formulate \emph{physically motivated} quantum machine learning (QML) tasks that exhibit learning separations? We address this problem by studying the learnability of quantum many-body dynamics from the perspective of probably approximately correct (PAC)-learning. Concretely, we devise a supervised learning problem where the training set consists of specifications of randomized stabilizer probe states, evolution times sampled uniformly at random from a polynomially large time interval $\left[ 0, T \right]$, coupled with expectation values of certain observables evaluated on the resulting time-evolved state under an unknown Hamiltonian. For this learning task, we provide an efficient quantum procedure whose training phase learns the underlying Hamiltonian from short-time training samples, and whose deployment phase combines Hamiltonian simulation with the classical shadows protocol to perform inference on a newly given data point. By contrast, the existence of $\mathcal{O} \left( \mathsf{poly} \left( n \right) \right)$-time instances ensures classical hardness: by embedding a $\mathsf{BQP}$-complete computation into the polynomially long time-dynamics of a low-intersection variant of the Feynman-Kitaev clock Hamiltonian construction, we show that, for a certain family of input distributions, no randomized classical polynomial-time algorithm can fulfill our learning condition, unless $\mathsf{BQP} \subseteq \mathsf{P/poly}$. Furthermore, we show that the classically hard instance maintains quantum learnability. We also give an interpretation of our results in learning-assisted certified quantum simulation. Taken together, our results demonstrate a rigorous learning separation for a natural ML task based on Hamiltonian evolution, while building connections between quantum learning theory, quantum simulation, and QML.  
\end{abstract}
\maketitle
\tableofcontents
\section{Introduction}
\label{sec:introduction}
Notions of ML are rapidly reshaping 
the world. Large language models and other applications of ML are now ubiquitous, substantially altering how we communicate, manage organizational tasks, and generate products in industrial settings. At the same time, although the concept of quantum computing is not new~\cite{Shor-1994}, recent years have seen the construction of quantum devices that suggest we may soon achieve quantum computers and simulators 
with practical utility~\cite{MindTheGaps,VastWorld}. Indeed, a number of quantum algorithms are known to exhibit super-polynomial advantages over their classical counterparts, some of which are directly relevant to practical problems in optimization~\cite{OptimizationAdvantages,DecodedQuantumInterferometry,OptimizationReview} and ML~\cite{PACLearning,TemmeML,VedranExponential,PhysRevLett.126.190505,JunyuPruned}, as well as those in quantum simulation~\cite{Lloyd, Nori, Childs_2018, daley2022practical, flannigan2022propagation, trivedi2024quantum, kashyap2025accuracy,RaulStability}. 
In light of these developments, a central challenge remains the identification of real-world ML tasks exhibiting meaningful learning separations. Initial results in this direction have already emerged, including rigorous, in-principle QML separations~\cite{PACLearning,TemmeML,VedranExponential}, as well as early evidence of robust quantum advantages for natural classical data~\cite{JunyuPruned,ProbabilityDiffusion, gyurik2024exponentialseparationsclassicalquantum, Molteni, molteni2026identification, barthe2025quantumadvantagelearningquantum, bokov2026machinelearningminimaluse}. 

Nevertheless, examples of provable learning separations remain relatively scarce, with many of such results relying on cryptographic assumptions. Thus, it is fair to argue that one of the central challenges in the field remains the identification of further well-motivated separations, particularly those grounded in physical considerations. Indeed, it is widely believed that quantum computers are capable of simulating certain many-body systems more efficiently than classical computers~\cite{feynman1986quantum, Lloyd}. Guided by the intuition that similar advantages may hold for \emph{learning} tasks as well, an immediate question arises:
\begin{align*}
    \text{\emph{Can one identify QML separations for physically motivated learning tasks?}}
\end{align*}
An example of such a learning problem is that of predicting ground state properties of quantum systems, a quintessentially quantum learning task for which one may expect quantum learners to yield an advantage over classical ML. Within the natural regime of gapped phases of matter, however, the aforementioned learning problem has been found to be solvable using classical ML algorithms~\cite{ProvablyEfficientQML, Lewis_2024, wanner2024predictinggroundstateproperties, Rouz__2024}, and the question of whether quantum advantage is attainable in learning properties of ground states thus remains dauntingly open. Therefore, in this work, we instead consider a closely related problem that is equally well motivated from a physical perspective, namely, predicting Hamiltonian time-evolution. Any physical system that is sufficiently isolated from its environment undergoes quantum evolution generated by a Hamiltonian. In analog quantum simulation~\cite{CiracZollerSimulation,BlochSimulation,GAN14,Trotzky}, one 
studies the time-dependent evolution generated by some precisely controlled Hamiltonian. Contrasting this, digital quantum simulation involves approximating the dynamics of some target quantum system by a discrete sequence of elementary quantum gates~\cite{Lloyd, Nori}. 

More generally, quantum simulation, whether that concerns real-time dynamics or preparation of ground or thermal states, may be regarded as a \textit{forward computational} problem. In turn, if we consider the problem of estimating the parameters of a many-body quantum system from measurement data thereof: in Ref.~\cite{Bakshi_2024}, this task has been conceptualized as an \textit{inverse, Hamiltonian learning} problem. Within an ML context, this dichotomy bears a natural analogy to \textit{training} and \textit{inference}, or, in the language of PAC-learning, to \textit{identification} and \textit{evaluation}. At a high level, the present work builds on this conceptual viewpoint, using tools from Hamiltonian learning and notions from complexity theory to formulate ML tasks based on Hamiltonian evolution that are classically hard yet quantumly tractable. In other words, we focus on showing a learning separation of the type considered in Refs.~\cite{gyurik2024exponentialseparationsclassicalquantum, Molteni, molteni2026identification, barthe2025quantumadvantagelearningquantum, bokov2026machinelearningminimaluse}, while leveraging ideas from quantum learning theory to inform the design of our learning procedure. 

A closely related work in this regard is Ref.~\cite{barthe2025quantumadvantagelearningquantum}, which has considered the problem of learning time-dynamics generated by an unknown Hamiltonian from classical measurement data. For this task, the authors of that work have shown a learning separation, premised on the widely believed complexity-theoretic conjecture that $\mathsf{BQP} \not\subseteq \mathsf{P/poly}$. However, their learning problem is formulated in terms of a relatively restrictive family of Hamiltonians, namely, ones parametrized by $\mathcal{O} \left( \log n \right)$-dimensional coefficient vectors. Consequently, this leads to an ML task that, while being physically motivated, does not capture the learning of physically plausible local Hamiltonians with a number of terms scaling, say, linearly or polynomially in the system size. By contrast, in this work, by assuming access to a stronger form of training data, we address the problem of predicting the dynamics of a more general family of Hamiltonians. Concretely, our contributions are as follows:
\begin{enumerate}
    \item We formulate a supervised learning task where one is required to predict expectation values of quantum states that have been time-evolved by some unknown, \textit{low-intersection} Hamiltonian, containing a number of constituent terms scaling as $\mathcal{O} \left( n \right)$. More precisely, our setting involves input data specifying probe states given by randomized stabilizer product states, as well as evolution times sampled uniformly at random, labeled by a list of expectation values evaluated on certain observables, which may be thought of as being morally equivalent to a classical description of the probe states after time-evolution by the unknown Hamiltonian. 
    \item We establish an exponential learning separation for the aforementioned learning problem, premised on the well-founded complexity-theoretic assumption $\mathsf{BQP} \not\subseteq \mathsf{P/poly}$. In particular, we provide an efficient quantum learning algorithm whose training phase is based on a simplified variant of the Hamiltonian learning protocol of Haah, Kothari, and Tang~\cite{Haah_2024}, and the deployment stage of which is a combination of Hamiltonian simulation and the classical shadow formalism. To complete the argument and prove classical hardness, we rely on the construction of Oliveira and Terhal~\cite{oliveira2005complexity}, who 
    have provided a low-intersection version of the Feynman-Kitaev clock Hamiltonian whose time-evolution realizes a BQP-complete computation for polynomially large evolution times.
\end{enumerate}
Intuitively, the set-up is as follows: our training dataset corresponds to the expectation values of time-evolved states for evolution times that are uniformly sampled over a time interval $\left[ 0, T \right]$, where $T \in \mathcal{O} \left( \mathsf{poly} \left( n \right) \right)$, so that the samples may be regarded as being divided into ``short"- and ``long"-time regimes. The short-time window constitutes an inverse-polynomial fraction of the time interval, and its corresponding samples 
are used during training to learn the data-generating Hamiltonian in a classical fashion. On the other hand, we require the presence of samples corresponding to polynomially long time-evolutions to ensure classical hardness. That is, 
given a time $t \sim \mathsf{Unif} \left( \left[ 0, T \right] \right)$ during the inference stage, the learner may be required to predict expectation values even for polynomially long evolution times, which is precisely the regime where the underlying dynamics can encode $\mathsf{BQP}$-complete computations.

We emphasize two aspects of our work. In our setting, we assume access to training data that is sufficiently rich to enable Hamiltonian learning, unlike that of Ref.~\cite{barthe2025quantumadvantagelearningquantum}, where the labels are scalars given by \emph{single} expectation values. In other words, while the formulation of their learning problem is substantially more elegant than ours, from a methodological perspective, it precludes the use of approaches based on Hamiltonian learning which requires a set of sufficient statistics to fully reconstruct the Hamiltonian. More fundamentally, as we explain further in Subsection \ref{subsec:distinction-bw-present-and-alice}, it also imposes complexity-theoretic limitations on the kinds of Hamiltonian families to which their approach can be applied \cite{arunachalam2021quantum}. 

We circumvent these difficulties by assuming stronger and more informative training data, naturally motivating the use of a simple Hamiltonian learning algorithm that is tailored to our randomized-time setting. At the same time, we draw attention to the fact that this leads to a training procedure that is somewhat trivial and algorithmically elementary, involving what amounts to classical processing of the training samples and a \emph{proper PAC} identification of the concept generating the training data. The source of our quantum learning advantage thus lies in the evaluation of the concept and not the identification thereof, leaving open the possibility of further identification-based 
separations in physically motivated learning tasks.
\begin{figure*}
    \label{fig:visualization-learning-problem}
    \includegraphics[width=0.9\linewidth]{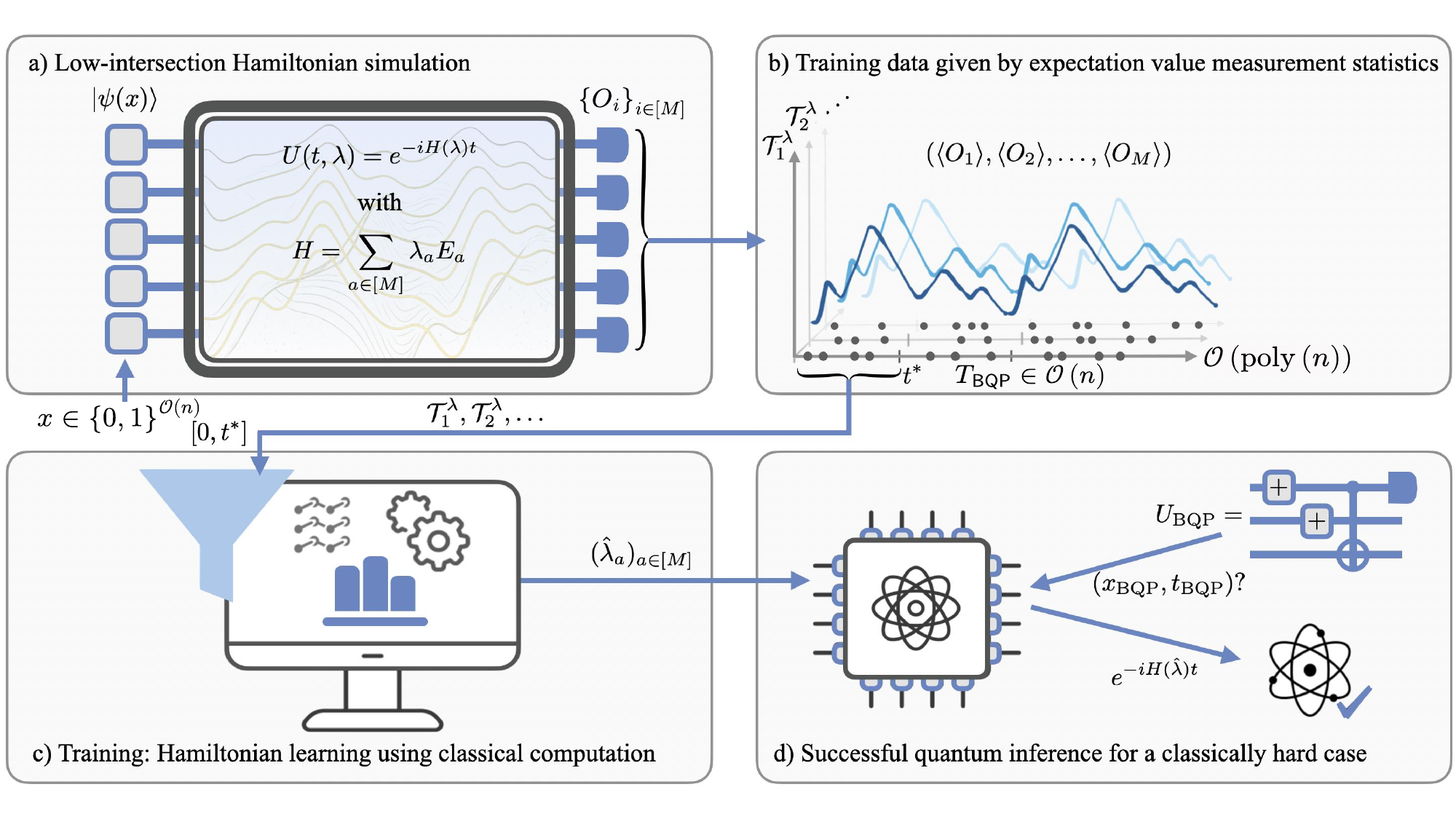}
    \caption{A visualization of the key ideas of the present work. (a) Our training samples are generated by the time-evolution of a low-intersection Hamiltonian $H \left( \lambda \right)$ for randomly chosen times $t \sim \mathsf{Unif} \left( \left[ 0, T \right] \right)$ acting on some initial probe state vector $\lvert \psi \left( x \right) \rangle$. (b) We consider as our labels a vector of expectation values evaluated on a set of certain pre-determined observables. This has been schematically depicted in the 
    above figure as a collection of measurement samples. (c) We filter out the short-time samples from our training data and use them to identify the underlying data-generating Hamiltonian via a simple Hamiltonian learning algorithm that is classical. (d) Finally, we use a quantum computer to successfully perform inference on a newly given data point. By universality of Hamiltonian evolution, this may involve predicting a $\mathsf{BQP}$-complete computation, which is believed to be hard for classical learners.} 
\end{figure*}

\vspace{2.5mm}

\noindent \textbf{Organization of this work.} This work is organized as follows. In Section~\ref{sec:related-works}, we provide an overview of the related literature on supervised learning of quantum many-body systems, learning separations, both information- and complexity-theoretic, and Hamiltonian learning. This is followed by Section~\ref{sec:prelims}, where we provide the requisite preliminaries on computational learning theory, complexity theory, and Hamiltonian learning protocols. We then rigorously formulate our learning problem in the language of the PAC-learning formalism in Section~\ref{sec:learning-problem-statement}. 
Then, in Subsection~\ref{sec:quantum-learnability}, we provide a detailed description of our quantum learning algorithm, following which we give an overview of our proof of classical hardness in Section~\ref{sec:quantum-advantage}. We conclude the main text with a discussion section where, in addition to enumerating a list of open questions and future directions, we outline in detail the main differences between our work and Ref.~\cite{barthe2025quantumadvantagelearningquantum}.

\vspace{2.5mm}

\section{Related works}
\label{sec:related-works}

\subsection{Supervised learning of quantum many-body physics via classical methods}
\label{subsec:supervised-learning-quantum-many-body-classical-methods}

In Ref.~\cite{ProvablyEfficientQML}, the authors employ classical ML techniques to learn classical representations of ground states of gapped Hamiltonians, as well as expectation values on local observables, with rigorous guarantees on sample and computational complexities. The success of classical ML in solving the aforementioned learning tasks is essentially due to the fact that, within the gapped regime, the physical properties of quantum systems and, in particular, their ground states are well-behaved, in the sense that they are smooth functions of the Hamiltonian specifications. Thus, the authors are able to use the $l_{2}$-Dirichlet kernel to learn such functions, using a number of training samples and computational runtime that scaling as $\mathcal{O} \bigl( n^{1/\varepsilon} \bigr)$, so that, for constant precision $\varepsilon \in \mathcal{O} \left( 1 \right)$, their algorithm is efficient. For the same learning task, follow-up papers considered the setting where the learner is equipped with additional knowledge of certain aspects of the quantum system being learned, such as its geometry and the knowledge of the observable being predicted. These include works by Lewis et al.~\cite{Lewis_2024} and Wanner et al.~\cite{wanner2024predictinggroundstateproperties}, which have improved the scaling of Ref.~\cite{ProvablyEfficientQML} to $\mathcal{O} \bigl( \log n \cdot 2^{\mathsf{poly} \log \left( 1/\varepsilon \right)} \bigr)$ and $\mathcal{O} \bigl( 2^{\mathsf{poly} \log \left( 1/\varepsilon \right)} \bigr)$, respectively, while retaining the quasi-polynomial dependence on system size for non-constant $\varepsilon$. As extensions of the gapped scenario, subsequent works have considered important classes of states beyond  ground states of gapped Hamiltonians. In Ref.~\cite{Rouz__2024}, the authors
have employed classical methods to learn properties of thermal states with exponentially decaying correlations, in addition to ground and thermal states satisfying the \emph{generalized approximate local indistinguishability} 
property. The setting of Ref.~\cite{Rouz__2024} has further been generalized in Ref.~\cite{onorati2023provablyefficientlearningphases} to Lindbladian phases of matter~\cite{Coser_2019}, which provably encompass the more conventional notions of phases. Finally, within the complementary setting of predicting dynamical properties of quantum systems, 
a noteworthy work is that of Huang, Chen, and Preskill~\cite{huang2023learningpredictarbitraryquantum}, where they address a supervised learning problem involving quantum states evolved by quantum channels of arbitrary complexity. Remarkably, by generalizing the notion of a low-degree approximation from classical PAC-learning theory to quantum observables, they have been able to devise a quasi-polynomial time classical ML algorithm.

\vspace{3mm}

\noindent\textbf{Relation of the present work to the above.} While our work is also situated within the context of the above literature, it should be noted that the vast majority of the aforementioned papers have to do with predicting \textit{static} properties of quantum systems, in that the learning problems they consider involve ground 
and thermal states. Moreover, the works mentioned above are able to successfully tackle such problems via methods that are either entirely or predominantly classical. Given the success of classical methods in this regime, in contrast, we consider the challenge of showing learning separations for the task 
of predicting learning quantum \textit{dynamics}. The advantage of this setting lies in the $\mathsf{BQP}$-completeness of time-evolution, a fact that allows us to prove complexity-theoretic learning separations for our task.   

\subsection{Learning separations}
\label{subsec:learning-separations}

Within the related field of quantum learning theory, information-theoretic learning separations have been shown for tasks involving genuinely quantum data, such as predicting properties of quantum states, learning quantum processes, quantum principal component analysis, as well as certain variants of quantum property testing~\cite{Huang_2021, chen2021exponentialseparationslearningquantum, Huang_2022_quantum_adv_exp, Aharonov_2022}. Importantly, such separations are predicated on being able to transduce quantum data onto an external quantum memory in a stable manner, followed by coherent processing of quantum information. A different strand of research has focused on the supervised learning setting, similar to that 
of Ref.~\cite{ProvablyEfficientQML}, involving primarily classical data obtained from quantum experiments~\cite{gyurik2024exponentialseparationsclassicalquantum, VedranExponential, barthe2025quantumadvantagelearningquantum, bokov2026machinelearningminimaluse}. In contrast to the previously discussed works on quantum learning separations, which involve unconditional information-theoretic lower bounds, the latter papers prove separations that rely on the classical hardness of simulating $\BQP$-complete processes. In particular, the authors of Ref.~\cite{gyurik2024exponentialseparationsclassicalquantum} have laid out a systematic framework for proving separations for quantum many-body learning tasks, by leveraging the $\mathsf{BQP}$-completeness of the data-generating function to prove classical heuristic hardness-of-learning in physically motivated settings under a fixed input distribution. In Ref.~\cite{VedranExponential}, premised on the assumption $\BQP\not\subseteq \mathsf{P\slash poly}$, learning separations were shown within distribution-free settings as well, marking an improvement over previously shown hardness statements based on $\mathsf{BQP} \not\subseteq \mathsf{HeurP/poly}$. The learning problem considered there is that of predicting expectation values of time-evolved quantum states on unknown observables, given as either linear combination of $k$-local Pauli operators as $O \left( \lambda \right) = \sum_{i} \lambda_{i} P_{i}$, or unitarily parametrized observables, given as $O \left( \lambda \right) = W \left( \lambda \right) O W\left( \lambda \right)^{\dagger}$, for some Hermitian operator $O$ and parametrized unitary $\lambda\mapsto W ( \lambda )$.  

In Ref.~\cite{barthe2025quantumadvantagelearningquantum}, the learning task is that of predicting expectation values of quantum states on a \emph{known} observable, which the authors tackle via a combination of Fourier coefficient extraction of parametrized Trotter circuits, linear regression, and kernel methods. However, in contrast to Ref.~\cite{VedranExponential}, the quantum states are time-evolved according to an \emph{unknown} parameterized Hamiltonian. The work of Ref.~\cite{bokov2026machinelearningminimaluse} introduces the \emph{learning under quantum privileged information} (LUQPI) framework, a quantum generalization of the classical formalism of Vapnik et al.~\cite{vapnik2009new, vapnik2015learning}. There, the authors envision a supervised learning scenario involving minimal quantum resources, whereby a quantum computer is used only in the extraction of certain features. Within this framework, the authors then construct examples of concept classes demonstrating exponential learning separations and provide numerical evidence of examples where the use of quantum privileged information is shown to enhance the performance of classical ML algorithms for learning quantum systems. To conclude, we also mention the work of Ref.~\cite{molteni2026identification}, wherein the authors asked: can quantum advantage in QML arise from the hardness of identifying the concept generating the data, as opposed to the hardness of evaluating it? There, the authors answered this question in the affirmative, by using the notion of random generatability to provide the first proofs of identification-based learning separations.

\vspace{3mm}

\noindent \textbf{Relation of the present work to the above.} Our work is in a similar vein to the works of Refs.~\cite{gyurik2024exponentialseparationsclassicalquantum,VedranExponential,barthe2025quantumadvantagelearningquantum,bokov2026machinelearningminimaluse}, in that we show a complexity-theoretic learning separation for the task of predicting expectation values of certain time-evolved quantum states. In terms of research direction, we reiterate that the current work is perhaps most closely aligned with that of Ref.~\cite{barthe2025quantumadvantagelearningquantum}, where the task of learning quantum real-time dynamics has also been considered. However, as we document in Subsection \ref{subsec:distinction-bw-present-and-alice}, our work differs from theirs in several key aspects. Namely, we consider Hamiltonian families satisfying properties such as low-intersection and geometric locality, leading to a learning task that is arguably more natural from a physical standpoint. Other distinctions have to do with details of how the learning problem is defined, in addition to the learning algorithm itself, which involves a combination of Hamiltonian learning and simulation during training and inference respectively. In connection with this methodological difference, we make note of the fact that the works of Refs.~\cite{gyurik2024exponentialseparationsclassicalquantum, molteni2024exponential, VedranExponential} were the first to envision the Hamiltonian learning problem as an identification problem. Our use of Hamiltonian learning in the training procedure may be interpreted as leveraging this insight.   

\subsection{Hamiltonian learning}
\label{subsec:rel-works-ham-learning}

Hamiltonian learning is the task of inferring the Hamiltonian corresponding to some quantum many-body system of interest, given measurement data obtained from the system. Historically, parameter estimation of specific families of Hamiltonians has been a well-studied problem in quantum control, quantum metrology, and quantum sensing~\cite{PhysRev.78.695, caves1981quantum, holland1993interferometric, Lee_2002, giovannetti2004quantum, de_Burgh_2005, degen2017quantum}. 
More recently, within the field of quantum learning theory, the literature on Hamiltonian learning has expanded substantially, with works on learning many-body systems from singular eigenstates~\cite{Garrison_2018, Bairey_2019, Qi_2019, Li_2020}, thermal states~\cite{SampleEfficientHamiltonianLearning, arunachalam2025testinglearningstructuredquantum, bakshi2026learningquantumhamiltonianstemperature, Haah_2024, Fawzi_2024, Garc_a_Pintos_2024, narayanan2024improvedalgorithmslearningquantum, Rouz__2024, chen2025quantumprobetomography}, and time-evolution~\cite{HamiltonianLearning,arunachalam2025testinglearningstructuredquantum, WildeLearning,Bairey_2019, Bakshi_2024, Caro_2024, da_Silva_2011, de_Burgh_2005, Flynn_2022, franca2025learningcertificationlocaltimedependent, Gentile_2021, giovannetti2004quantum, Haah_2024, Hu_2025, Huang_2023_hl, Lee_2002, li2023heisenberglimitedhamiltonianlearninginteracting, ZollerLearning,ma2024learningkbodyhamiltonianscompressed, Mirani_2024, ni2024quantumhamiltonianlearningfermihubbard, Odake_2024, PhysRev.78.695, PhysRevA.84.012107, sinha2025improvedhamiltonianlearningsparsity, Wiebe_2014, Wiebe_2014_prl, Zhao_2025, zubida2021optimalshorttimemeasurementshamiltonian, ZollerLearning,shin2026heisenberglimitedhamiltonianlearningshorttime, Dutkiewicz_2024, abbas2025nearlyoptimalalgorithmslearn, chen2025quantumprobetomography}. Here, we focus on the third setting in particular, since our learning problem has to do with predicting expectation values on time-evolved states. This setting is particularly relevant for experimentally learning large-scale Hamiltonians
from realistic data in a robust fashion \cite{HamiltonianLearning}.
For present purposes, however, rigorous arguments about precise
sample complexities in idealized situations are particularly 
important. A landmark result within that literature has been  the work of Haah, Kothari, and Tang, wherein rigorous guarantees were provided for the first time on Hamiltonian learning from high-temperature Gibbs states, as well as query-access to time evolution~\cite{Haah_2024}. Subsequent works also considered a version of Hamiltonian learning where the interaction structure of the underlying many-body system is unknown to the learner~\cite{Bakshi_2024, Hu_2025, abbas2025nearlyoptimalalgorithmslearn}. By this, we mean that the interaction terms constituting the Hamiltonian are not provided as inputs to the problem, so that one is required to perform \textit{structure learning} in addition to learning the interaction strengths. A key requirement is Heisenberg-limited scaling, where one is interested in achieving a total query time of $t_{\text{time}} \in \mathcal{O} \left( 1 / \varepsilon  \right)$, where $\varepsilon$ is the desired precision on the Hamiltonian coefficients. Numerous works exist in this direction, such as those of Refs.~\cite{Huang_2023_hl, Hu_2025, li2023heisenberglimitedhamiltonianlearninginteracting, Mirani_2024, ni2024quantumhamiltonianlearningfermihubbard, shin2026heisenberglimitedhamiltonianlearningshorttime}, as well as the related work of Dutkiewicz et al., where they show the fundamental necessity of a certain amount of quantum control in achieving Heisenberg-limited scaling~\cite{Dutkiewicz_2024}.

\vspace{3mm}

\noindent \textbf{Relation of present work to the above.} As we will see later, our learning algorithm will involve a simplification of the Haah-Kothari-Tang protocol, an adjustment that will prove appropriate for our setting, since it involves training data corresponding to randomized time instances. So, while the learning problem considered in this work is a supervised one, we borrow heavily from the set-up of Ref.~\cite{Haah_2024}, and our analysis of quantum learnability will involve numerous mathematical objects considered in that work.

While the core aim of our work is the physically motivated learning separation as such, there are a few other implications that are of independent interest in the context of Hamiltonian learning. One interpretation of our results is in \emph{learning-assisted certified quantum simulation}. Instead of attempting to certify a difficult long-time quantum evolution directly, we characterize an experimentally accessible short-time regime, infer the governing Hamiltonian, and use this learned model as the basis for long-time quantum prediction. In this way, Hamiltonian learning serves as a calibration stage for subsequent quantum simulation, yielding rigorous guarantees on the resulting predictions. This perspective is somewhat reminiscent of the approach taken in Ref.~\cite{ZollerLearning}, but here comes with explicit complexity-theoretic and learning-theoretic guarantees.

\section{Preliminaries}
\label{sec:prelims}

\subsection{PAC-learning}
\label{subsec:pac-learning}

The PAC-learning framework provides a formal and mathematically precise language for describing supervised learning problems. Within this framework, a learning problem is specified by a set of functions, called the \textit{concept class}, usually denoted by $\mathcal{C}$. The individual functions contained in the concept class are called \textit{concepts}. Each concept $c \colon \mathcal{X} \rightarrow \mathcal{Y}$ maps an input/instance space $\mathcal{X}$ to an output/label space $\mathcal{Y}$. In the supervised learning setting, a learner has access to a \textit{training dataset} of the form $\mathcal{T} = \left\{ \left( x_{i}, c \left( x_{i} \right) \right) \right\}_{i = 1}^{N}$. This is a set of input-output pairs, where the inputs $x_{i}$ are drawn independently from an unknown but fixed \textit{target distribution} $\mathcal{D}$ over the input space $\mathcal{X}$, and the outputs are examples that are labeled by a concept $c \in \mathcal{C}$. Given this dataset, the learner must then output a \textit{hypothesis} $h \colon \mathcal{X} \rightarrow \mathcal{Y}$ that approximately ``agrees" with the labeling concept $c$ with high probability over the target distribution $\mathcal{D}$. To formalize the aforementioned picture, we will now give a definition of \textit{efficient PAC-learnability}, which requires that the learner outputs a hypothesis in polynomial time.

    \begin{definition}[Efficient PAC-learnability]
        \label{def:eff-pac-learnability}
        Consider a concept class $\mathcal{C} = \left\{ c_{j} \right\}_{j}$, comprised of concepts $c_{j} \colon \mathcal{X} \rightarrow \mathcal{Y}$, where $\mathcal{X}$ is the instance space and $\mathcal{Y}$ is the output space. Let $\mathcal{L} \colon \mathcal{Y} \times \mathcal{Y} \rightarrow \mathbb{R}_{\geqslant 0}$ be a loss function. Then, $\mathcal{C}$ is efficiently PAC-learnable with respect to a fixed distribution $\mathcal{D}$ over $\mathcal{X}$ if there exists an algorithm $\mathcal{A}$ that takes in as input a dataset $\mathcal{T} = \left\{ \left( x_{i}, c(x_{i}) \right) \right\}_{i = 1}^{N}$, accuracy and confidence parameters $\varepsilon, \delta > 0$, such that, for all concepts $c \in \mathcal{C}$, the size of the dataset $N \in \mathcal{O} \left( \mathsf{poly} \left( n, 1/\varepsilon, 1/\delta \right) \right)$, where $n$ is the dimension of the instance space, the algorithm $\mathcal{A}$ outputs a hypothesis function $h$ satisfying
        \begin{align}
            \label{eq:efficient-pac-learnability}
            \underset{x \sim \mathcal{D}}{\mathbb{E}} \left[ \mathcal{L} \left( h \left( x \right), c \left( x \right) \right) \right] \leq \varepsilon, 
        \end{align}
        with high probability $1 - \delta$ over $\mathcal{D}$ over $\mathcal{X}$, and runs in time $\mathcal{O} \left( \mathsf{poly} \left( n, 1/\varepsilon, 1/\delta \right) \right)$. 
    \end{definition}

\subsection{Complexity theory}
\label{subsec:complexity-theory}

The central result of our work is a learning separation premised on the widely-held complexity-theoretic assumption that $\mathsf{BQP} \not\subseteq \mathsf{P/poly}$. We therefore begin with the definition of the class $\mathsf{BQP}$.
\begin{definition}[$\mathsf{BQP}$]
    \label{def:BQP}
    Let $\mathcal{L} \coloneqq \left( \mathcal{L}_{\textnormal{yes}}, \mathcal{L}_{\textnormal{no}}  \right)$ be a promise problem, where $\mathcal{L}_{\textnormal{yes}}, \mathcal{L}_{\textnormal{no}}  \subseteq \left\{ 0, 1 \right\}^{*}$ and $\mathcal{L}_{\textnormal{yes}} \cap \mathcal{L}_{\textnormal{no}} = \emptyset$. Then, $\mathcal{L} \in \mathsf{BQP}$ if and only if there exists a polynomial-time, uniform family of quantum circuits $\left\{ U_{n} \colon n \in \mathbb{N} \right\}$, where, for all $n \in \mathbb{N}$, $U_{n}$ takes in as input an $n$-qubit computational basis state and outputs a single bit obtained by measuring the first output qubit in the computational basis, such that the following properties are satisfied:
    \begin{enumerate}
        \item For all $x \in \mathcal{L}_{\textnormal{yes}}$, the probability that $U_{\left\vert x \right\vert}$ accepts $x$ is greater than or equal to $2/3$; that is, $\Pr \left[ U_{\left\vert x \right\vert} \left( x \right) = 1 \right] \geqslant 2/3$.
        \item For all $x \in \mathcal{L}_{\textnormal{no}}$, the probability that $U_{\vert x \rvert}$ does not accept $x$ is greater than or equal to $2/3$; that is, $\Pr \left[ U_{\left\vert x \right\vert} \left( x \right) = 0 \right] \geqslant 2/3$.
    \end{enumerate}
\end{definition}
An important category of complexity classes we will consider is that of advice classes, wherein a Turing machine is provided, in addition to the input, a polynomially sized \textit{advice} string. Intuitively, the advice string should be thought of as enhancing the ``power" of the computation being performed, compared to the setting where it is only supplied with the problem input. Advice classes are used to model the complexity of computational problems which involve a resource intensive pre-processing stage, followed by a post-processing stage that is computationally less expensive. This makes them suitable for analyzing the complexity of ML problems. We provide a definition of the class $\mathsf{P/poly}$, as follows.

\begin{definition}[$\mathsf{P/poly}$, Def. 6.5,~\cite{arora2009computational}]
    \label{def:p-slash-poly}
    A decision problem $\mathcal{L} \colon \left\{ 0, 1 \right\}^{*} \rightarrow \left\{ 0, 1 \right\}$ is contained in the complexity class $\mathsf{P}/\mathsf{poly}$ if there exists a polynomial-time, classical algorithm, $\mathcal{A}$, satisfying the following property: for every $n \in \mathbb{N}$, there exists a polynomially sized bit-string, $\lambda_{n} \in \left\{ 0, 1  \right\}^{\mathsf{poly} \left( n \right)}$, such that, for all $x \in \left\{ 0, 1 \right\}^{n}$, $\mathcal{A} \left( x, \lambda_{n} \right) = \mathcal{L}(x)$.
\end{definition}

\subsection{Hamiltonian learning}
\label{subsec:ham-learning}
Consider a quantum system defined on $n$ qubits. The interactions characterizing the system are governed by a $2^{n} \times 2^{n}$ Hermitian operator called the \textit{Hamiltonian}, which we will denote by $H$. Such an operator is expressible as a linear combination of Pauli operators, with their coefficients representing bounded interaction strengths. The task of \textit{learning} a Hamiltonian, then, may be informally described as recovering $\varepsilon$-accurate estimates of the Hamiltonian coefficients with high probability, 
given as input either of the following:
\begin{enumerate}
    \item Query access to the real-time evolution of the Hamiltonian, $e^{-iHt}$, for some $t > 0$.
    \item Several, identically prepared copies of the thermal state of the Hamiltonian 
\begin{equation}\rho_{\beta} \left( H \right) = \frac{e^{-\beta H}}{\operatorname{Tr}\left[ e^{-\beta H} \right]}, 
    \end{equation}
    at a known inverse temperature $\beta > 0$.
\end{enumerate}
The setting of the present work will involve the first access model, in that the input to our supervised learning problem will be a training dataset containing expectation values of randomized probe states that have been evolved under an unknown, low-intersection Hamiltonian, which we define below. 
\begin{definition}[Hamiltonian]
    \label{def:hamiltonian}
    A Hamiltonian $H$ is given by a set of tuples of the form $\left\{ \left( \lambda_{a}, E_{a} \right) \right\}_{a \in \left[ M \right]}$, where $M \in \mathbb{N}$, so that $a \in \left[ M \right]$ is an index ranging over $\left[ M \right] = \left\{ 1, 2, \ldots, M \right\}$, $\lambda_{a} \in \left[ -1, 1 \right]$ for all $a \in \left[ M \right]$, and $\left( E_{a} \right)_{a \in \left[ M \right]}$ are unique, non-identity Pauli operators. The associated Hamiltonian is then given as 
    \begin{align}
        H \coloneqq \sum_{a \in \left[ M \right]} \lambda_{a} E_{a}.
    \end{align}
    On occasion, we will also use $\lambda \coloneqq \left( \lambda_{1}, \lambda_{2}, \ldots, \lambda_{M} \right)$ to denote the entire vector of coefficients, and $H \left( \lambda \right)$ to express the Hamiltonian as a function of $\lambda$. We also say that $H$ is $\mathfrak{K}$-local if $\left\vert \operatorname{supp} \left( E_{a} \right) \right\vert \leqslant \mathfrak{K}$ for all $a \in \left[ M \right]$.
\end{definition}

\begin{definition}[Low-intersection Hamiltonian~\cite{Haah_2024}]
    \label{def:low-intersection-hamiltonians}
    Let $H = \sum_{a \in \left[ M \right]} \lambda_{a} E_{a}$ be an $n$-qubit, $\mathfrak{K}$-local Hamiltonian as in Def.~\ref{def:hamiltonian}. The dual-interaction graph $\mathfrak{G}$ of $H$ is an undirected graph whose vertices are identified with elements of $\left[ M \right]$, with an edge connecting vertices $a, b \in \left[ M \right]$ if
    \begin{align}
        \operatorname{supp} \left( E_{a} \right) \cap \operatorname{supp} \left( E_{b} \right) \neq \emptyset. 
    \end{align}
    Let $\mathfrak{d}$ denote the maximum degree of $\mathfrak{G}$ over all vertices. Then, we say that $H$ is low-intersection if both $\mathfrak{K}, \mathfrak{d} \in \mathcal{O} \left( 1 \right)$.
\end{definition}
Low-intersection Hamiltonians have also been referred to in the literature as \textit{sparsely-interacting Hamiltonians}~\cite{Gu_2024} and \textit{low-interaction Hamiltonians}~\cite{Huang_2023_hl}. We note here that low-intersection $\mathfrak{K}$-local Hamiltonians subsume the more familiar class of geometrically local $\mathfrak{K}$-local Hamiltonians, which constitute, by far, the most physically plausible family of Hamiltonians, since the latter must satisfy the additional locality constraint that each interaction term be supported on a spatially contiguous region. We now formally state the problem of Hamiltonian learning from query access to time-evolution. 
\begin{problem}[Hamiltonian learning from time-evolution,~\texorpdfstring{\cite{Haah_2024}}{Haah_2024}]
    Let $H \coloneqq \left\{ \left( \lambda_{a}, E_{a} \right) \colon a \in \left[ M \right] \right\}$ be an $n$-qubit Hamiltonian as defined in Def.~\ref{def:hamiltonian}. Let $t \leqslant t_{c}$ be a known time, where $t_{c} \in \mathcal{O} \left( 1 \right)$ is a constant and $U \coloneqq e^{-iHt}$. Then, assuming knowledge of $\left( E_{a} \right)_{a \in \left[ M \right]}$, the learning problem involves taking in as inputs precision and confidence parameters $\varepsilon, \delta > 0$, and $\mathcal{O} \left( \mathsf{poly} \left( n, 1 / \varepsilon, 1 / \delta \right) \right)$ queries to the time-evolution unitary $U$, and producing as outputs $\hat{\lambda} \coloneqq \bigl( \hat{\lambda}_{a} \bigr)_{a \in \left[ M \right]}$ such that, with probability at least $1 - \delta$, $\bigl\vert \lambda_{a} - \hat{\lambda}_{a} \bigr\vert \leqslant \varepsilon$ holds for all $a \in \left[ M \right]$.
\end{problem}

\subsubsection{Background on the methods of Haah et al.}
\label{subsubsec:background-haah-kothari-tang}
We now provide some background on the results of Ref.~\cite{Haah_2024}. Let $H = \left\{ \left( \lambda_{a}, E_{a} \right)~\vert~a \in \left[ M \right] \right\}$ be a low-intersection Hamiltonian as in Defs.~\ref{def:hamiltonian} and \ref{def:low-intersection-hamiltonians}. The starting point of their Hamiltonian learning algorithm is an expectation value we denote by $\mathcal{F}_{a, t}$, which, for $a \in \left[ M \right]$, they define as
\begin{align}
    \label{eq:definition-of-mathcal-f-a-t}
    \mathcal{F}_{a, t} \coloneqq \operatorname{Tr} \left[ Q_{a} e^{-iHt} \left( \frac{I + P_{a}}{D} \right) e^{iHt} \right] = \frac{1}{D} \operatorname{Tr} \left[ Q_{a} e^{-iHt} P_{a} e^{iHt} \right].
\end{align}
In the above, $P_{a}$ is chosen to be any single-qubit Pauli operator that anti-commutes with the Hamiltonian term $E_{a}$, $Q_{a} \coloneqq i \left[ P_{a}, E_{a} \right] = 2 i P_{a} E_{a}$, $D \coloneqq 2^{n}$, where $n$ is the number of qubits constituting the many-body system, and $t \leqslant t_{c} \in \mathcal{O} \left( 1/\mathsf{poly} \left( \mathfrak{d} \right) \right)$. It can be shown, as in the proof of Lemma \ref{lem:tail-bound-higher-order} or Eq.\ (8) of Ref.~\cite{Haah_2024}, that the quantity $\mathcal{F}_{a, t}$ admits the
expansion
\begin{align}
    \label{eq:mathcal-f-a-t-taylor-expansion-simple}
    \mathcal{F}_{a, t} = \frac{1}{D} \operatorname{Tr} \left[ Q_{a} e^{-iHt} P_{a} e^{iHt} \right] = 4\lambda_{a} t + \mathcal{O} \left( t^{2} \right),
\end{align}
so that the coefficient $\lambda_{a}$ is exposed in the first-order term. This is a consequence of the Hadamard formula (see Lemma \ref{lem:hadamard-formula}), orthonormality of Pauli operators with respect to the Hilbert-Schmidt inner product, as well as the specific choices of $P_{a}$ and $Q_{a}$. The authors of Ref.~\cite{Haah_2024} envision the Hamiltonian learning problem as one of inverting the non-linear map
\begin{align}
    \label{eq:ham-learning-hkt-non-linear-map}
    \mathcal{F}_{t} \colon \left[ -1, 1 \right]^{M} \ni \left( \lambda_{a} \right)_{a \in \left[ M \right]} \mapsto \left( \mathcal{F}_{b, t} \coloneqq \operatorname{Tr} \left[ Q_{b} e^{-iHt} P_{b} e^{iHt} \right] \right)_{b \in \left[ M \right]} \in \mathbb{R}^{M},
\end{align}
for all $a \in \left[ M \right]$. In particular, their protocol is carried out via the following two steps:
\begin{enumerate}
    \item Computing estimates $\tilde{\mathcal{F}}_{a, t}$ of $\mathcal{F}_{a, t}$ up to some precision $\xi$, such that $\bigl\vert \tilde{\mathcal{F}}_{a, t} - \mathcal{F}_{a, t} \bigr\vert \leqslant \xi$, for all $a \in \left[ M \right]$. This is straightforwardly performed by measuring time-evolved randomized stabilizer states on the observable $Q_{a}$.
    \item Using estimates $\tilde{\mathcal{F}}_{a, t}$, inverting the map $\mathcal{F}_{t}$ via the Newton-Raphson method.
\end{enumerate}
Their use of the Newton-Raphson method is necessitated by the fact that the map $\mathcal{F}_{t}$ is non-linear, in that the forward map from the coefficients $\left( \lambda_{a} \right)_{a \in \left[ M \right]}$ to $\left( \mathcal{F}_{a, t} \right)_{a \in \left[ M \right]}$ involves exponentials in the Hamiltonian. Furthermore, due to their lack of control over $t$, they are unable to utilize a first-order approximation of $\mathcal{F}_{a, t}$ to arbitrarily small error. However, as we explain later, our setting \textit{will} allow us to take advantage of a linear approximation of $\mathcal{F}_{a, t}$, leading to a drastically simpler Hamiltonian learning algorithm.

\section{Learning problem}
\label{sec:learning-problem-statement}

In this section, we formalize the central learning problem behind our quantum advantage results. At a high level, the goal of this task is to predict, for various given times and input states, a classical description of these states after they have been time evolved under an unknown Hamiltonian. The training data for this task is generated as follows: for different inputs $x\in\{0,1\}^{\mathsf{poly}(n)}$ and evolution times $t \in[0,T]$, an $n$-qubit input state vector $\ket{\psi (x)}$ is prepared and time-evolved by a low-intersection Hamiltonian $H \left( \lambda \right)$ for time $t$, i.e., it evolves by $U( t, \lambda) = e^{- i H \left( \lambda \right) t}$. The classical description of this time-evolved state then takes the form of expectation values of a predetermined set of certain observables that are gathered via repeated preparations and measurements. The Hamiltonian parameters $\lambda \in[-1,1]^M$ are fixed but unknown, with $M \in \mathcal{O} \left( n \right)$, which follows from the low-intersection property. In particular, we consider Hamiltonians 
\begin{equation}\label{eq: hamiltonian}
    H \left( \lambda \right) = \sum_{a \in \left[ M \right]} \lambda_a E_a
\end{equation}
of the kind described in Def.~\ref{def:hamiltonian}. 
The task of a learner will then be to use these classical snapshots of time-evolved states from polynomially many different $(x,t)$ pairs to generalize to new inputs and evolution times. This learning problem can be rigorously formalized by means of the following concept class.

\begin{definition}[Concept class]
    \label{def:learning_problem_concept_class}
    Let $\mathsf{H} \coloneqq \bigl\{ H \left( \lambda \right) \vert \lambda \in \left[ -1, 1 \right]^{M} \bigr\}$ be a family of $n$-qubit low-intersection Hamiltonians. Let $x \in \left\{ 0, 1 \right\}^{p \left( n \right)}$, where $p \colon \mathbb{N} \rightarrow \mathbb{N}$ is a function mapping $n$ to some integer and $p  \left( n \right) \in \mathcal{O} \left( \mathsf{poly} \left( n \right) \right)$. Also, consider $\left( Q_a \right)_{a \in \left[ M \right]}$, the operators defined in Subsection \ref{subsubsec:background-haah-kothari-tang}. Then, we define the concept class
    \begin{align}
        \label{eq:main-concept-class-time-dynamics}
        \mathcal{C}_{U_{\textnormal{enc}}, T}^{\mathsf{H}} \coloneqq \left\{ c_{\lambda} \colon \left\{ 0, 1 \right\}^{p \left( n \right)} \times \left[ 0, T \right] \ni \left( x, t \right) \mapsto \left( \langle \psi \left( x \right) \rvert U^{\dagger} \left( t, \lambda \right) Q_{a} U \left( t, \lambda \right) \lvert \psi \left( x \right) \rangle \right)_{a \in \left[ M \right]} \in \mathbb{R}^{M} \bigl\vert \lambda \in \left[ -1, 1 \right]^{M} \right\},
    \end{align}
    where $U \left( t, \lambda \right) \coloneqq e^{-iH \left( \lambda \right) t}$, such that $\lvert \psi \left( x \right) \rangle = U_{\textnormal{enc}} \left( x \right) \lvert 0 \rangle^{\otimes n}$, where $U_{\textnormal{enc}}$ is a data-encoding unitary and $T \in \mathcal{O} \left( \mathsf{poly} \left( n \right) \right)$. 
\end{definition}

Note that the concept class is specified by a fixed choice of encoding unitary $U_{\text{enc}}$, the Hamiltonian family $\mathsf{H} = \{H(\lambda)\}_\lambda$ under which we perform the time evolution, and a maximal evolution time $T$. In this work, we are interested in an encoding unitary of a particular kind, whose action is determined by the first bit of the input bit-string $x$. For some input pair $\left( x^{(j)}, t^{(j)} \right)$, we define it as
\begin{equation}
    \label{eq:data-encoding-unitary}
	U_{\text{enc}}\bigl( x^{(j)} \bigr)\ket{0^n} \coloneqq
	\begin{cases}
		\ket{\psi \left( x^{(j)} \right)}=\ket{s^{(j)}} \in \{\ket{0},\ket{1},\ket{+},\ket{-},\ket{i},\ket{-i}\}^{\otimes n},  &\text{if}~x^{(j)}_1=0,\\
	       \ket{\psi \left( x^{(j)} \right)}=\ket{x^{(j)}_2, \dots, x^{(j)}_{n + 1}}, &\text{if}~x^{(j)}_1=1,
	\end{cases}
\end{equation}
where, in the case when $x^{(j)}_1=0$, the choice of the state vector $\ket{s^{(j)}}$ is fully determined by the bit-string $x^{(j)}_2, \dots, x^{(j)}_{p(n)}$. In particular, the state vectors $\ket{s^{(j)}}$ are tensor products of single-qubit stabilizer states and can therefore be described by bitstrings of length $p(n) = 1 + \lceil n \log_{2} 6 \rceil$, where $n$ is the number of qubits on which the state vector $\ket{\psi \left( x^{(j)} \right)}$ is defined. In the following, we denote by $s^{(j)}$ the bit-strings $x^{(j)}$ with $x^{(j)}_1=0$, where the remaining bits $x^{(j)}_2,\ldots,x^{(j)}_{p(n)}$ specify the particular stabilizer state vector $\ket{s^{(j)}}$. The goal of the ML algorithm, then, is to learn a model $h(x,t)$ which approximates the unknown concept
\begin{align}
    c_\lambda(x,t) = \left( \langle \psi(x) \rvert U^{\dagger} \left( t, \lambda \right) Q_{a} U \left( t, \lambda \right) \lvert \psi(x) \rangle\right )_{a \in \left[ M \right]},    
\end{align}
given training data of the form $\mathcal{T}^{\lambda} \coloneqq \left\{ \left( x^{(j)},t^{(j)},y^{(j)} \right) \right\}_{j=1}^N$, where $y^{(j)} \coloneqq c_\lambda(x^{(j)},t^{(j)})$. On occasion, in the context of quantum learnability, i.e., when $x_{1} = 0$, we will also use 
\begin{align}
    B_{a}^{(j)} \coloneqq \operatorname{Tr} \left[ Q_{a} e^{-iH\left( \lambda \right)t} \bigotimes_{i = 1}^{n} \bigl\lvert s_{i}^{(j)} \bigr\rangle \bigl\langle s_{i}^{(j)} \bigr\rvert e^{i H\left( \lambda \right)t} \right]
\end{align}
to denote the $a^{\text{th}}$ component of the $j^{\text{th}}$ training sample, for notational convenience.
To fully formalize the learning task, we also need to specify the family of distributions $\{ \mathcal{D}_i \}_i$ from which the training data points $(x^{(j)},t^{(j)})$ are sampled. These distributions will also govern the measure of approximation (or generalization performance) of the hypothesis $h(x,t)$ over the entire space $(x,t)$. We will consider in particular input distributions of the following form:

\begin{definition}[Allowed input distributions]
    \label{def: input distribution}
	Let $p > 0$ and $\mathcal{D}_i\in\{\mathcal{D}_i\}_i$ be a family of distributions over $(x,t) \in  \{0,1\}^{p}\times [0,T]$ defined as:
	\begin{itemize}
        \item The time $t$ is uniformly sampled over $[0,T]$.
		\item The first bit $x_1$ of $x$ is randomly selected with equal probability between 0 and 1.
		\item If $x_1=0$ then the other $p-1$ bits $x_2, x_3, \ldots, x_p$ are sampled from the uniform distribution over $p-1$ bit-strings $\mathsf{Unif}(\{0,1\}^{p-1})$.
		\item If $x_1=1$ then the $p-1$ bits in $x_2, \dots, x_p$ follow any possible distribution over $\{0,1\}^{p - 1}$. For each $i$, the distribution over these bit-strings defines the specific overall distribution $\calD_i$.
	\end{itemize}
	In our setting, we will consider $p = p(n)$.
\end{definition}
Formally, our learning condition is defined as follows.
\begin{definition}[Efficient PAC-learning condition]\label{def: learning condition}
The concept class $\mathcal{C}^{\mathsf{H}}_{U_{\textnormal{enc}},T}$ is (efficiently) PAC-learnable with respect to the set of distributions  $\{ \mathcal{D}_i \}_i$ over $(x,t)\in\{0,1\}^{p ( n )}\times[0,T]$ if there exists an algorithm $\mathcal{A}$ such that, for all $ \varepsilon, \delta > 0$, for any $c_{\lambda}$ in $\calC^{\mathsf{H}}_{U_{\textnormal{enc}},T}$, and any input distribution $\calD_i \in \left\{ \mathcal{D}_{i} \right\}_{i}$, $\mathcal{A}$ receives as input a training set $\calT^\lambda$ of size $\mathcal{O} \left( \mathsf{poly} \left( n, 1/\varepsilon, 1/\delta \right) \right)$ and outputs a hypothesis $h( \cdot )= \mathcal{A} (\calT^\lambda, \varepsilon,\delta, \cdot)$ which satisfies 
    \begin{equation}\label{eq:learning}
    \underset{\substack{(x,t) \sim \mathcal{D}_i}}{\mathbb{E}} \left[ \left\Vert c_{\lambda} \left( x,t \right) - h \left( x,t \right) \right\Vert_{\infty} \right] \leqslant \varepsilon
    \end{equation}
    with probability greater than $1-\delta$. The concept class is said to be efficiently PAC-learnable if the time complexity of $\mathcal{A}$ is $\mathcal{O} \left( \mathsf{poly} \left( n, 1/\varepsilon, 1/\delta \right) \right)$.
\end{definition}

\section{Quantum learnability}
\label{sec:quantum-learnability}

Having properly formalized the learning task considered in this work, we now provide a quantum learning algorithm that is able to efficiently PAC-learn the concept class described in Def.~\ref{def:learning_problem_concept_class} with respect to a certain family of input distributions. This entails: i) learning the Hamiltonian parameters $\lambda$ using short-time-evolved states sampled from the $x_1=0$ subspace, and ii) guaranteeing that the learned parameters $\hat{\lambda}$ allow time-evolution of new input states up to good approximation error, even for long evolution times $t \sim T$. Recall from the previous section that our training data is of the form
\begin{align}
    \label{eq:training-dataset}
    \mathcal{T}^{\lambda} &= \left\{ \left( x^{(j)}, t^{(j)}, c_{\lambda} \bigl( x^{(j)}, t^{(j)} \bigr) \right) \right\}_{j = 1}^{N}, \\
    \text{where}\quad c_{\lambda} \bigl( x^{(j)}, t^{(j)} \bigr) &= \left( \langle \psi(x) \rvert U^{\dagger} \left( t, \lambda \right) Q_{a} U \left( t, \lambda \right) \lvert \psi(x) \rangle\right )_{a \in \left[ M \right]}. 
\end{align}
Given such a training dataset as input, our learning algorithm is divided into the following two stages:
\begin{enumerate}
    \item A training phase, where, using the expectation values provided in $\mathcal{T}^{\lambda}$, we identify the concept generating the data via Hamiltonian learning. This phase is specified by Algorithm \ref{alg:training-via-ham-learning-short-times}.
    \item An inference stage, involving evaluation of the learned concept, which is performed by means of Hamiltonian simulation and the classical shadows protocol. This phase is specified by Algorithm \ref{alg:inference-via-hamiltonian-simulation}.
\end{enumerate}
We will proceed by considering the expansion of $\mathcal{F}_{a, t}$ shown in Eq.\  (\ref{eq:mathcal-f-a-t-taylor-expansion-simple}). Following an alternative idea sketched in Ref.~\cite{Haah_2024}, we note that one could choose to determine an $\varepsilon$-dependent $t$ such that higher-order terms satisfy an appropriate tail bound, in which case the coefficient becomes recoverable through simply division of an estimate of $\mathcal{F}_{a, t}$ by $4t$. For the scenario considered in this work, where the times contained in the training dataset are sampled from a uniform distribution over $\left[ 0, T \right]$, where $T \in \mathcal{O} \left( \mathsf{poly} \left( n \right) \right)$, we consider the map 
\begin{align}
    \tilde{\mathcal{F}}_{t^{*}} \colon \left[ -1, 1 \right]^{M} \ni \left( \lambda_{a} \right)_{a \in \left[ M \right]} \mapsto \left( \tilde{\mathcal{F}}_{b, t^{*}} \coloneqq \frac{1}{t^{*}} \int_{0}^{t^{*}} \mathcal{F}_{b, t} dt \right)_{b \in \left[ M \right]} \in \mathbb{R}^{M}
\end{align}
(instead of the one defined in Eq.\ (\ref{eq:ham-learning-hkt-non-linear-map})),
where $t^{*}$ is chosen such that 
\begin{align}
    \label{eq:ideal-estimator-series-expansion}
    \frac{1}{t^{*}} \int_{0}^{t^{*}} \mathcal{F}_{b, t}~dt = 2 \lambda_{a} t^{*} + \underset{\leqslant \varepsilon_{\text{tail}}}{\underbrace{\ldots}}
\end{align}
can be shown to hold,
where $\varepsilon_{\text{tail}}$ will be some allocated error budget. Algorithmically, we will compute the time-averaged integral via Monte-Carlo estimation, which we perform simply by taking the empirical average over training samples contained in the short-time window $\left[ 0, t^{*} \right] \subset \left[ 0, T \right]$. Dividing the empirical average by $2t^{*}$ gives us an estimate $\hat{\lambda}_{a}$ of $\lambda_{a}$, up to some precision controlled by $\varepsilon_{\text{tail}}$ and the accuracy of Monte-Carlo estimation. We then calculate the sample complexity of short-time samples, which we denote by $N_{t^{*}}$. We do this by bounding the difference between the Monte-Carlo estimate and the time-averaged integral, which we show to be the expectation of the random variable from which the training samples are drawn, conditioned on $t \in \left[ 0, t^{*} \right]$, thereby allowing us to use Hoeffding's inequality. Finally, we guarantee that enough training samples are contained in the short-time window, as well as in the $x_{1}^{(j)} = 0$ subspace, so as to ensure efficient execution of the training algorithm. To address these requirements, we use the Chernoff bound to compute the total sample complexity $N$, and find that we incur only a polynomial blow-up with respect to $N_{t^{*}}$.

The sample complexity of the inference stage is determined by the cost of implementing the classical shadows protocol, which is logarithmic in the number of observables being predicted and the inverse of the failure probability, and inverse-polynomial in the desired precision. 
We also find that the overall runtime of our learning procedure, a detailed analysis of which we defer to Appendix \ref{app:quantum-learnability}, is polynomial in the relevant parameters. We now give an informal statement of our main quantum learnability result.

\begin{figure}[t]
\refstepcounter{algorithm}
\label{alg:training-via-ham-learning-short-times}
\hrule
\vspace{0.5ex}
\noindent\textbf{Algorithm~\thealgorithm. Training via Hamiltonian learning from short-time randomized probes}
\vspace{0.5ex}
\hrule
\vspace{0.75ex}
\begin{algorithmic}[1]
    \Input 
    \Statex \hspace{\algorithmicindent} 1. Final precision and confidence parameters $\varepsilon, \delta > 0$. 
    \Statex \hspace{\algorithmicindent} 2. A training dataset of the form $\mathcal{T}^{\lambda}$ of size $N$ as defined in Eq.\ (\ref{eq:training-dataset}).
    \Statex \hspace{\algorithmicindent} 3. Knowledge of the maximum degree $\mathfrak{d}$ of the dual-interaction graph $\mathfrak{G}$ of $H$, as well as the Hamiltonian terms $\left( E_{a} \right)_{a \in \left[ M \right]}$.
    \Output Estimates $\hat{\lambda} \coloneqq \bigl( \hat{\lambda}_{a} \bigr)_{a \in \left[ M \right]}$ of $\lambda \coloneqq \left( \lambda_{a} \right)_{a 
    \in \left[ M \right]}$ such that, for all $a \in \left[ M \right]$, $\bigl\vert \lambda_{a} - \hat{\lambda}_{a} \bigl\vert \leqslant \varepsilon/2MT$ with probability at least $1 - \delta/2$.
    \vspace{2mm}
    \State $t^{*} \leftarrow \mathcal{O} \left(  \varepsilon/M T \left\Vert O \right\Vert_{\infty} \left( \mathfrak{d} + 1 \right)^{2} \right)$, where $O$ is the observable being predicted
    \For{$a = 1$ to $M$}
        \State $S_{a, t^{*}} \leftarrow 0$  \Comment{Initialize sum variable.}
        \State $\mathrm{count}_{a, t^{*}} \leftarrow 0$ \Comment{Initialize counter variable.}
        \For{$j = 1$ to $N$}    \Comment{Loop over all data points in training set.}
        \If{$x^{(j)}_{1} = 0$} \Comment{Check whether sample can be used for learning.}
            \If{$t^{(j)} \in \left[ 0, t^{*} \right]$} \Comment{Check if data point is contained in short-time window}
                \If{$P_{a} \bigl\lvert s_{i}^{(j)} \bigr\rangle  = + \bigl\lvert s_{i}^{(j)} \bigr\rangle$} \Comment{Check if $\bigl\vert s_{i}^{(j)} \bigr\rangle$ is a $+1$ eigenstate of $P_{a}$.}
                \State $S_{a, t^{*}} \leftarrow S_{a, t^{*}} + 6B_{a}^{(j)}$
            \EndIf
            \State $\mathrm{count}_{a, t^{*}} \leftarrow \mathrm{count}_{a, t^{*}} + 1$ \Comment{Count number of points in short-time window.}
            \EndIf
        \EndIf
        \EndFor
        \State $\hat{\mathcal{F}}_{a, t^{*}} \leftarrow S_{a, t^{*}}/\mathrm{count}_{a, t^{*}}$ \Comment{Compute empirical average.}
        \State $\hat{\lambda}_{a} \leftarrow \hat{\mathcal{F}}_{a, t^{*}}/2t^{*}$   \Comment{Invert the map $\hat{\lambda}_{a} \mapsto \hat{\mathcal{F}}_{a, t^{*}}$ via division by $2t^{*}$.}
    \EndFor
    \State \Return $\bigl( \hat{\lambda}_{a} \bigr)_{a \in \left[ M \right]}$ \Comment{Return estimates of Hamiltonian coefficients.}
\end{algorithmic}
\vspace{0.75ex}
\hrule
\end{figure}

\begin{figure}[t]
\refstepcounter{algorithm}
\label{alg:inference-via-hamiltonian-simulation}
\hrule
\vspace{0.5ex}
\noindent\textbf{Algorithm~\thealgorithm. Inference via Hamiltonian simulation and classical shadows}
\vspace{0.5ex}

\hrule
\vspace{0.75ex}
    \begin{algorithmic}[1]
        \Input 
        \Statex \hspace{\algorithmicindent} 1. Final precision and confidence parameters $\varepsilon, \delta > 0$.
        \Statex \hspace{\algorithmicindent} 2. Learned coefficients $\bigl( \hat{\lambda}_{a} \bigr)_{a \in \left[ M \right]}$ via Algorithm \ref{alg:training-via-ham-learning-short-times}. 
        \Statex \hspace{\algorithmicindent} 3. Inference point $\bigl( x, t \bigr) \sim \mathsf{Unif} \left( \left\{ 0, 1 \right\}^{p \left( n \right)} \times \left[ 0, T \right] \right)$.
        \vspace{1mm}
        \Statex \hspace{\algorithmicindent} 4. The Hamiltonian terms $\left( E_{a} \right)_{a \in \left[ M \right]}$.
        \Statex \hspace{\algorithmicindent} 5. The set of observables $\bigl( Q_{a} \bigr)_{a \in \left[ M \right]}$.
        \Statex \hspace{\algorithmicindent} 6. Sample-access to the uniform distribution over the ensemble of unitaries $\mathcal{U} \coloneqq \left\{ \mathbbm{1}, H, HS^{\dagger} \right\}^{\otimes n}$.
        \vspace{1mm}
        \Output Estimates $\left( \langle \hat{Q}_{a} \rangle_{\rho \left( x, t \right)} \right)_{a \in \left[ M \right]}$ such that, for all $a \in \left[ M \right]$, $\left\vert \langle \hat{Q}_{a} \rangle_{\rho \left( s, t \right)} - \langle Q_{a} \rangle_{\rho \left( s, t \right)} \right\vert \leqslant \varepsilon$ with probability at least $1 - \delta/2$.
        \vspace{1mm}
        \State $\hat{H} \leftarrow \sum_{a \in \left[ M \right]} \hat{\lambda}_{a} E_{a}$
        \vspace{1mm}
        \State Set $G \leftarrow \lceil 2 \log \left( 4M/\delta \right) \rceil$ \Comment{Number of classical shadow segments.}
        \vspace{1mm}
        \State Set $\tilde{N}$ to be the following: \Comment{Number of copies of time-evolved states required for classical shadows protocol.}
        \begin{align}
            \tilde{N} \leftarrow \frac{136G}{\varepsilon^{2}} \underset{a \in \left[ M \right]}{\max} \hspace{-0.6cm} \underset{\hspace{1cm}\leqslant 4\cdot 3^{\left\vert\text{supp}\left( E_{a} \right)  \right\vert} ~\text{(for Pauli measurements)}}{\underbrace{\biggl\Vert Q_{a} - \overset{=0}{\overbrace{\frac{\operatorname{Tr} \left[ Q_{a} \right]}{2^{n}}}} \mathbbm{1} \biggr\Vert_{\text{shadow}}^{2}}}\hspace{-1cm}
        \end{align}
        \vspace{2mm}
        \For{$\ell = 1$ to $\tilde{N}$}  \Comment{Loop over all copies for data acquisition phase.}
            \vspace{1mm}
            \State $\rho_{\mathrm{init}}^{(\ell)}(x)
            \leftarrow
            U_{\mathrm{enc}}(x)\ket{0^n}\!\bra{0^n}U_{\mathrm{enc}}(x)^{\dagger}$
            \Comment{Encoded input state preparation.}

            \State $\hat{\rho}^{(\ell)}(x,t) \leftarrow e^{-i\hat{H}t} \rho_{\mathrm{init}}^{(\ell)}(x)e^{i\hat{H}t}$
            \Comment{Hamiltonian simulation.}
        
        \EndFor \Comment{Obtain $\tilde{N}$ copies of the time-evolved probe state.}
        \vspace{-1mm}
        \State Run Subroutine \ref{alg:subroutine-data-acquisition-classical-shadows} on the states $\left( \hat{\rho}^{(\ell)}(x,t) \right)_{\ell = 1}^{\tilde{N}}$ \Comment{Data acquisition phase of classical shadow protocol, outputs $\bigl( \hat{\sigma}^{(\ell)} \bigr)_{\ell = 1}^{\tilde{N}}$.}
        \State Run Subroutine \ref{alg:subroutine-median-of-means-classical-shadows} using $\bigl( \hat{\sigma}^{(\ell)} \bigr)_{\ell = 1}^{\tilde{N}}$, $\left( Q_{a} \right)_{a \in \left[ M \right]}$, and $G$ as inputs   \Comment{Median-of-means estimation of expectation values on $\bigl( Q_{a} \bigr)_{a \in \left[ M \right]}$.}
        \State \Return $\left( \langle \hat{Q}_{a} \rangle_{\rho \left( s, t \right)} \right)_{a \in \left[ M \right]}$ \Comment{Predicted expectation values.}
    \end{algorithmic}
\vspace{0.75ex}
\hrule
\end{figure} 
\begin{theorem}[Quantum learnability, informal]
\label{thm:quantum-learnability-informal}
    There exists a family of distributions, namely, the one defined in Def.~\ref{def: input distribution}, such that, with respect to all distributions in that family, for precision and confidence parameters $\varepsilon, \delta > 0$, the concept class in Def.~\ref{def:learning_problem_concept_class} is efficiently PAC-learnable using a quantum learning algorithm whose sample- and time-complexity scales as $\mathcal{O} \left( \mathsf{poly} \left( n, 1/\varepsilon, 1/\delta \right) \right)$.
\end{theorem}
\begin{proof}[Proof-sketch]
    The above theorem is formally restated as Theorem \ref{thm:quantum-learnability-formal}, which may be found in Subsection \ref{subsec:sample-time-complexity-analysis} of Appendix \ref{app:quantum-learnability}. In order to show that we can efficiently PAC-learn the concept class using Algorithms \ref{alg:training-via-ham-learning-short-times} and \ref{alg:inference-via-hamiltonian-simulation}, we need to prove that the overall procedure for solving the learning task is both sample- and time-efficient. Our sample complexity analysis proceeds by analyzing the error decomposition 
    \begin{align}
        \left\vert \hat{\lambda}_{a} - \lambda_{a} \right\vert &= \left\vert \frac{\hat{\mathcal{F}}_{a, t^{*}}}{2t^{*}} - \lambda_{a} \right\vert = \left\vert \frac{\hat{\mathcal{F}}_{a, t^{*}}}{2t^{*}} - \frac{\tilde{\mathcal{F}}_{a, t^{*}}}{2t^{*}} + \frac{\tilde{\mathcal{F}}_{a, t^{*}}}{2t^{*}} - \lambda_{a}  \right\vert \leqslant \left\vert \frac{\hat{\mathcal{F}}_{a, t^{*}}}{2t^{*}} - \frac{\tilde{\mathcal{F}}_{a, t^{*}}}{2t^{*}} \right\vert + \left\vert \frac{\tilde{\mathcal{F}}_{a, t^{*}}}{2t^{*}} - \lambda_{a} \right\vert,
    \end{align}
    where $\bigl( \hat{\lambda}_{a} \bigr)_{a \in \left[ M \right]}$ are the coefficients retrieved in the training phase, $\bigl( \lambda_{a} \bigr)_{a \in \left[ M \right]}$ are the actual coefficients corresponding to the ground-truth Hamiltonian, and $\hat{\mathcal{F}}_{a, t^{*}}$ is the Monte-Carlo estimator used to approximate the time-averaged integral $\tilde{\mathcal{F}}_{a, t^{*}}$. As we have already discussed, we bound the difference between $\hat{\mathcal{F}}_{a, t^{*}}/2t^{*}$ and $\tilde{\mathcal{F}}_{a, t^{*}}/2t^{*}$ using Hoeffding's inequality, and that between $\tilde{\mathcal{F}}_{a, t^{*}}/2t^{*}$ and $\lambda_{a}$ via tail bounds on $\mathcal{F}_{a, t^{*}}$. This is followed by an application of the Chernoff bound to compute the total sample complexity, which, for precision and failure probability $\varepsilon, \delta > 0$ respectively, we find to be polynomial in $n$ and $1/\varepsilon$, and logarithmic in $1/\delta$. The sample complexity of the inference stage is dominated by the cost of running the classical shadows formalism, which is inverse-polynomial in $\varepsilon$ and logarithmic in $1/\delta$ and the number of observables being predicted. The time complexity we show to be polynomial in the relevant parameters as well. The training stage involves elementary constant-time arithmetic operations for each data sample and observable $Q_{a}$, whereas the runtime of the deployment phase is dictated by the cost of Hamiltonian simulation, which, for polynomial evolution time, is also polynomial, and that of the classical shadows protocol, which involves preparation of logarithmically many copies of time-evolved states, followed by classical polynomial-time post-processing.
\end{proof}

\section{Quantum advantage}
\label{sec:quantum-advantage}
In this section, we instantiate the learning problem described by the concept class in Def.~\ref{def:learning_problem_concept_class} in a way that shows a quantum learning advantage. Specifically, we define a family of input distributions $\mathcal{D}_i \in \{\mathcal{D}_i\}_i$ and Hamiltonians $H(\lambda)$ for which no classical algorithm can satisfy the learning condition in Def.~\ref{def: learning condition}. Together with the fact that the quantum algorithm described in Section \ref{sec:quantum-learnability} can still learn the particular instantiation described here, this establishes a quantum learning advantage. Our construction relies on mainly two tools: i) a Hamiltonian family $\{H_{\BQP}(\lambda)\}_\lambda$ due to Oliveira and Terhal~\cite{oliveira2005complexity} that satisfies the low intersection property and whose time-evolution for time $T_{\mathsf{BQP}} \in \mathcal{O} \left( n \right)$ is $\BQP$-complete with high probability, ii) and the seminal boosting result of Schapire~\cite{schapire1990strength} that allows us to promote any efficient classical PAC-learner under our choice of input distributions to polynomial-size classical circuits for deciding $\BQP$ languages (i.e., that would imply $\BQP\subseteq\mathsf{P\slash poly}$).

\subsection{Classical hardness}
\label{sec: circuit}
We start by stating our classical hardness result, which is based on the widely-believed conjecture that $\BQP\not\subseteq\mathsf{P\slash poly}$, and provide a proof-sketch. A detailed proof is given in Appendix~\ref{app:classical-hardness}.
\begin{theorem}[Classical hardness, informal]
    \label{thm: classical hardness}
	For any $(\mathsf{Promise})\BQP$-complete language, there exists a Hamiltonian family $\mathsf{H}_{\BQP}$ specifying a concept class $\mathcal{C}_{U_{\text{enc}}, T}^{\mathsf{H}_{\BQP}}$ as per Def.~\ref{def:learning_problem_concept_class} and a family of input distributions following the specification of Def.~\ref{def: input distribution} such that no randomized polynomial-time classical algorithm $\mathcal{A}_C$ satisfies the PAC-learning condition of Def.~\ref{def: learning condition} for this concept class, even when restricted to a constant learning error $\varepsilon < 1/48$, unless $\BQP \subseteq \mathsf{P/poly}$.
\end{theorem}
\begin{proof}[Proof-sketch]
	Here, we summarize the main argument behind the proof. We begin by defining a Hamiltonian family $\mathsf{H}_{\BQP} = \{ H_{\BQP}(\lambda)\}_\lambda$ such that $\BQP$ computations can be embedded into the Hamiltonian evolution of $H_{\BQP}(\lambda)$ by properly choosing $\lambda$. Let $\mathcal{L}$ be a $(\mathsf{Promise})\BQP$-complete language over input bit-strings $z\in\{0,1\}^m$. Since $\mathcal{L}\in\BQP$, there exists a polynomial-depth quantum circuit $U_{\BQP}$ deciding $\mathcal{L}$, in the sense that
	\begin{itemize}
		\item if $z\in\mathcal{L}$, then applying $U_{\BQP}$ to $\ket{z}$ outputs $1$ with probability at least $2/3$;
		\item if $z\notin\mathcal{L}$, then applying $U_{\BQP}$ to $\ket{z}$ outputs $0$ with probability at least $2/3$.
	\end{itemize}
	Using the construction of Ref.~\cite{oliveira2005complexity}, one obtains a low-intersection local Hamiltonian $H_{\BQP}$ acting on $n \in \mathcal{O} \left( \mathsf{poly} \left( m \right) \right)$ qubits whose real-time evolution for polynomially or linearly long times $T_{\BQP} \in \mathcal{O} \left( n \right)$ implements the circuit $U_{\BQP}$ with high probability (see Appendices~\ref{subsec:pauli_decomposition} and~\ref{subsec:time-evolution-feynman-kitaev-hamiltonian}). The key to making this Hamiltonian low-intersection compared to the standard Feynman/Kitaev construction is to expand the simulated circuit $U_{\BQP}$ over $\mathcal{O} \left( \mathsf{poly} \left( m \right) \right)$ many qubits (or, more precisely, $mK$ qubits for a depth $K$ circuit) such that any qubit involved in this circuit gets only acted upon by a constant number of gates. In the Feynman-Kitaev construction, this translates into each qubit being involved in only a constant number of local Hamiltonian terms, i.e., the Hamiltonian satisfying the low-intersection property. Moreover, the construction can be easily modified so that $H_{\BQP}$ simulates $U_{\BQP}$ with arbitrarily high probability for a randomly sampled evolution time in $[0,T_{\BQP}]$. Expanding each local term of $H_{\BQP}$ in the Pauli basis gives a Hamiltonian of the form in Eq.~\eqref{eq: hamiltonian}. The Hamiltonian family $\mathsf{H}_{\BQP} = \{ H_{\BQP}(\lambda)\}_\lambda$ is then defined by considering all Pauli expansions resulting from this construction, for arbitrary circuits $U_{\BQP}$. In this way, the Pauli coefficients $\lambda\in[-1,1]^M$ specify the particular Hamiltonian $H_{\BQP}$ that simulates a given circuit $U_{\BQP}$.

	We now relate this Hamiltonian evolution to the decision problem for $\mathcal{L}$. For inputs $x$ with $x_1=1$, we identify the bits $z~=~x_2\cdots x_{m+1}$
	with an input to the $\BQP$ computation. Measuring the observable $O = Z\otimes I\otimes\cdots\otimes I$ on the state obtained after the evolution generated by $H_{\BQP}$ distinguishes whether $z\in\mathcal{L}$ or $z\notin\mathcal{L}$. More precisely, the concept associated with $\lambda$ contains a component of the form
	\begin{equation}
		\left( c_{\lambda}(x,t) \right)_{i^{\star}}
		=
		\operatorname{Tr}\!\left[
		O e^{-iH_{\BQP}(\lambda)t}\rho(x)e^{iH_{\BQP}(\lambda)t}
		\right],
	\end{equation}
	where $i^{\star}$ is the index associated with component corresponding to $O$. Averaging this component over $t\in[0,T_{\BQP}]$, and restricting our attention to the subset of inputs with $x_1=1$, we find that the resulting function
	\begin{equation}
		\overline{q}(z) = \frac{1}{T}\int_{t=0}^{T} \left( c_{\lambda}((\bm1, z, \bm0),t) \right)_{i^{\star}} 
	\end{equation}
	approximates the $\BQP$ computation deciding $\mathcal{L}$ up to good margin for all $z\in\{0,1\}^m$. Armed with this observation, the only difficulty from here onward is that a classical learner achieving non-trivial \emph{average} performance with respect to some input distributions $\{\mathcal{D}_i\}_i$ does not immediately guarantee that it can solve the $\BQP$ problem. We need additional machinery to leverage this hypothetical efficient classical learner to construct an algorithm in $\mathsf{P/poly}$ that can decide the $\mathsf{BQP}$ language.

	Suppose now, toward a contradiction, that there were a randomized polynomial-time classical algorithm satisfying the learning condition for the concept class $\mathcal{C}_{U_{\text{enc}}, T}^{\mathsf{H}_{\BQP}}$ with small overall error $\varepsilon \in \mathcal{O} \left( 1 \right)$. Applying this learner to the concept specified by $\lambda$ would produce, with high probability, a hypothesis that approximates the component $\left( c_{\lambda}(x,t) \right)_{i^{\star}}$ with average error $\mathcal{O} \left( \varepsilon \right)$ over the subspace $x_1 = 1$ and $t\in[0,T_{\BQP}]$. By further averaging the hypothesis over the time window $t\in[0,T_{\BQP}]$ and suitably setting the threshold error $\varepsilon$ to $1/48$, one obtains a classical hypothesis that correctly decides $\mathcal{L}$ with probability strictly more than $1/2$ with respect to any distribution $\mu$ over $z\in\{0,1\}^m$.

	The final step uses the result of Ref.~\cite{schapire1990strength}, which states that if a function is PAC-learnable under arbitrary distributions, then it can be computed by a polynomial-size classical circuit on every input. We recall this result and explain its consequences in Appendix~\ref{app:classical-hardness}. Applying it here would imply the existence of a polynomial-size classical circuit deciding the $(\mathsf{Promise})\BQP$-complete language $\mathcal{L}$. Therefore $\BQP\subseteq\mathsf{P/poly}$, contradicting the assumed hardness. Hence no randomized polynomial-time classical learner can satisfy the learning condition for $\mathcal{C}_{U_{\text{enc}}, T}^{\mathsf{H}_{\BQP}}$, unless $\BQP\subseteq\mathsf{P/poly}$.
\end{proof}

\subsection{Quantum learnability for a classically hard case}
\label{subsec:quantum-learnability-for-hard-case}

Recall that the objective of our quantum learner is to output an estimate $\hat{\lambda}$ such that the hypothesis function $h$, defined by $h(x,t)=c_{\hat{\lambda}}(x,t)$, satisfies the learning condition in Eq.~(\ref{eq:learning}). For this, the learning algorithm devised in the last section (Algorithm \ref{alg:training-via-ham-learning-short-times}) is directly applicable, as per Theorem \ref{thm:quantum-learnability-informal}. An interesting observation however is that, in the case of our Hamiltonian family $\mathsf{H}_{\BQP}$, it is sufficient to recover an $\hat{\lambda}$ such that $\bigl\Vert \hat{\lambda} - \lambda \bigr\Vert_{\infty} \leqslant \varepsilon \in \mathcal{O}(1)$ in order to recover $H_{\BQP}(\lambda)$ exactly. In turn, this means that we can learn the concept class perfectly, i.e., with \emph{zero} generalization error, and that too using less samples than stated in Theorem \ref{thm:quantum-learnability-informal}.

\begin{lemma}[Quantum learnability of $\mathcal{C}^{\mathsf{H}_{\mathsf{BQP}}}_{U_{\text{enc}}, T}$]
    \label{lemma: quantum learnability const error}
	Consider the learning problem corresponding to the concept class $\mathcal{C}_{U_{\text{enc}}, T}^{\mathsf{H}_{\BQP}}$ that is defined by the Hamiltonian family $\mathsf{H}_{\BQP}$ presented in Theorem~\ref{thm: classical hardness}. Then, Algorithm~\ref{alg:training-via-ham-learning-short-times} is able to meet the learning condition in Def.~\ref{def: learning condition} with zero error, i.e., $h( \cdot ) = \mathcal{A} (\mathcal{T}^\lambda,\varepsilon, \delta, \cdot~)$ is such that with probability $1-\delta$,
    \begin{equation}
        h(x,t)=c_\lambda(x,t) \quad \forall(x,t)\in\{0,1\}^{p(n)}\times[0,T],
    \end{equation}
    for any concept $c_{\lambda}$ in $\mathcal{C}_{U_{\text{enc}}, T}^{\mathsf{H}_{\BQP}}$. Moreover, the number of samples required by $\mathcal{A}$ is $N_{\mathsf{BQP}} \in \mathcal{O} \left( T \log \left( M / \delta \right) \right)$, and the time-complexity of $\mathcal{A}$ is given by $\mathcal{O} \left( M N_{\mathsf{BQP}} \right)$. 
\end{lemma}
\begin{proof}[Proof]
    The proof follows from Lemmas~\ref{lemma: quantum learnability const error app} and~\ref{lem:samp-time-complexity-of-quantumly-learning-classically-hard-instance} in the Appendix. In particular, Lemma~\ref{lemma: quantum learnability const error app} guarantees that the Hamiltonian associated with the classically hard instance can be written as a $5$-local, low-intersection Hamiltonian $ H_{\mathsf{BQP}}(\lambda)=\sum_{a \in \left[ M \right]}\lambda_a E_a$, where $\left\{ E_a \right\}_{a \in \left[ M \right]}$ are local Pauli strings and the coefficients $\lambda_a$ belong to a fixed finite set $\Lambda\subset\mathbb{R}$ that is independent of the system size. Hence, to reconstruct $H_{\mathsf{BQP}}(\lambda)$ exactly, it suffices to estimate each coefficient with precision smaller than half the minimum separation between any two distinct admissible values. In the Pauli expansion provided in Lemma~\ref{lemma: quantum learnability const error app}, two identical Pauli terms can appear in $H_{\BQP}$ only in the case we consider the $Z$-dependent part of the gates \begin{equation}
    T = \left( \frac{1 + e^{i \pi / 4}}{2} \right) I + \left( \frac{1 - e^{i \pi / 4}}{2} \right) Z 
    \end{equation}
    and $H = ( X + Z )/\sqrt{2}$. In this case, the difference of the coefficients is
\begin{equation}\label{eq:precision_hl}  
        \varepsilon_P=\left|\frac{1-e^{i\pi/4}}{2}-\frac{1}{\sqrt{2}}\right|\approx 0.66.
    \end{equation}
    After accounting for the normalization factors introduced by the clock-register terms and the two $\mathrm{SWAP}$ operations, the minimum separation between two distinct admissible coefficients is lower-bounded by $\varepsilon'_P=\varepsilon_P/32$. We may therefore run $\mathcal{A}$ with coefficient-estimation precision strictly smaller than $\varepsilon'_P/2$. With probability at least $1-\delta$, each estimated coefficient $\hat{\lambda}_a$ can then be uniquely rounded to the corresponding exact value $\lambda_a$. Consequently, $H_{\BQP}(\hat{\lambda})=H_{\mathsf{BQP}}(\lambda)$, and the hypothesis returned by the learner coincides with the target concept 
    \begin{equation}
        h(x,t)=c_\lambda(x,t)
        \qquad
        \forall (x,t)\in\{0,1\}^{p(n)}\times[0,T]
    \end{equation}
    for every input and every evolution time.
    Moreover, the ability to learn the Hamiltonian exactly also has an impact on the sample complexity of our learning algorithm. Indeed, unlike the general case considered in Theorem~\ref{thm:quantum-learnability-informal}, the error propagation in Hamiltonian simulation (see Lemma~\ref{lem:error-prop-ham-coeff-exp-values}) is trivially accounted for, which significantly reduces the sample complexity of Algorithm~\ref{alg:training-via-ham-learning-short-times}. As detailed in Lemma~\ref{lem:samp-time-complexity-of-quantumly-learning-classically-hard-instance}, we can get away with a sample complexity $N_{\mathsf{BQP}} \in \mathcal{O} \left( T \log \left( M / \delta \right) \right)$, as opposed to  
    \begin{equation}
    N \in \mathcal{O} \left( \frac{M^{5} T^{6} \left\Vert O \right\Vert_{\infty}^{3} \left( \mathfrak{d} + 1 \right)^{6}}{\varepsilon^{5}} \log \left( \frac{M}{\delta} \right) \right) 
    \end{equation}
    in the general case of Theorem~\ref{thm:quantum-learnability-informal}, which gives $N = \mathcal{O} \left(M^{5} T^{6} \log \left( M/\delta \right) \right)$ upon plugging in the scaling $\varepsilon \in \mathcal{O} \left( 1 \right)$.
\end{proof}

\section{Discussion}
\label{sec:discussion}

\subsection{Distinction between the present work and earlier work}
\label{subsec:distinction-bw-present-and-alice}
We recall the PAC-learnable version of the concept class in Ref.~\cite{barthe2025quantumadvantagelearningquantum}:
\begin{definition}[PAC-learnable Hamiltonian dynamics concept class, Def.~4,~\cite{barthe2025quantumadvantagelearningquantum}]
    \label{def:hamiltonian-dynamics-concept-class-log-pac-learnable}
    Consider a sequence of $n$-qubit Hamiltonians $\left\{ H_{n} \left( x, \lambda \right) \right\}_{n \in \mathbb{N}}$, such that each Hamiltonian $H_{n} \left( x, \lambda \right)$ is parametrized by a bit-string $x \in \left\{ 0, 1 \right\}^{p \left( n \right)}$ of length $p \left( n \right)$, where $p \colon \mathbb{N} \rightarrow \mathbb{N}$, and $\lambda \in \left[ 0, 1 \right]^{d \left( n \right)}$, a vector of continuous, unknown parameters, with $d \left( n \right) \in \mathcal{O} \left( \log n \right)$. For $\left\{ O_{n} \right\}_{n \in \mathbb{N}}$ a sequence of observables, consider the concept class
    \begin{align}
        \mathcal{C}^{\mathsf{H}}_{n, d} \coloneqq \left\{ c_{\lambda} \colon \left\{ 0, 1 \right\}^{p \left( n \right)} \ni x \mapsto \langle 0 \rvert U_{n}^{\dagger} \left( x, \lambda \right) O_{n} U_{n} \left( x, \lambda \right) \lvert 0 \rangle \in \mathbb{R} ~\bigl\vert~\lambda \in \left[ 0, 1 \right]^{d \left( n \right)} \right\},
    \end{align}
    where $U_{n} \left( x, \lambda \right) \coloneqq e^{-i H_{n} \left( x, \lambda \right) \tau}$, where $\tau \in \mathbb{R}_{> 0}$ is a fixed, positive real number.
\end{definition}
Let us enumerate the key differences between the work of Barthe et al. and that of our paper. To begin with, in Ref.~\cite{barthe2025quantumadvantagelearningquantum}, for $n \in \mathbb{N}$, the unknown part of the Hamiltonian $H_{n} \left( x, \lambda \right)$, i.e., the part parametrized by $\lambda$, can contain only logarithmically many terms, since the dimensionality of $\lambda$ must be $d \left( n \right) \in \mathcal{O} \left( \log n \right)$. This is in contrast with our setting, since we impose no such restriction. In particular, for our learning task, the unknown Hamiltonian generating the training data satisfies the low-intersection property described in Def.~\ref{def:low-intersection-hamiltonians}, which leads to the number of constituent terms being $\mathcal{O} \left( n \right)$, a scenario that encompasses many-body systems of the kind typically encountered in condensed matter physics. 

Another difference between our work and Ref.~\cite{barthe2025quantumadvantagelearningquantum} lies in the form of the training data considered. Specifically, each of our training examples is labeled by a list of expectation values on the observables $\left( Q_{a} \right)_{a \in \left[ M \right]}$, whereas the authors of Ref.~\cite{barthe2025quantumadvantagelearningquantum} require only a single expectation value on some arbitrary observable $O$ per label. In this sense, our training data is richer and potentially more informative, which leads to a methodological difference between our work and Ref.~\cite{barthe2025quantumadvantagelearningquantum}---the stronger access model enables us to learn the unknown Hamiltonian generating the data, which in turn allows us to predict a wider range of observables than in Ref.~\cite{barthe2025quantumadvantagelearningquantum}. Put differently, our training procedure approximately identifies the concept generating the training data by learning the underlying Hamiltonian itself. Our approach is thus a form of proper PAC-learning, where the learner is required to output a hypothesis that is contained in the concept class underlying the learning problem. By contrast, in Ref.\  \cite{barthe2025quantumadvantagelearningquantum}, the hypothesis class differs from the concept class described in Def.~\ref{def:hamiltonian-dynamics-concept-class-log-pac-learnable}; in particular, it is given by a family of generalized trigonometric polynomials associated with the parametrized Trotter circuit implementing the time-evolution operator $U_{n} \left( x, \lambda \right) = e^{-i H_{n} \left( x, \lambda \right) t}$:
\begin{align}
    \mathcal{H}_{n, d}^{\mathsf{PQC}} \coloneqq \left\{ h_{\lambda} \colon \left\{ 0, 1 \right\}^{p \left( n \right)} \ni x \mapsto \hspace{-0.5mm}\sum_{\ell \in \mathcal{L}}b_{\ell} \hspace{-0.7mm}\left( x \right) e^{i \langle \lambda, \ell \rangle} \in \mathbb{R} ~\biggl\vert \hspace{0.25mm}\lambda \in \left[ 0, 1 \right]^{d \left( n \right)} \right\},
\end{align}
where $\mathcal{L}$ is the frequency spectrum of the parametrized quantum circuit (PQC) and $\left\vert \mathcal{L} \right\vert = d \left( n \right) \in \mathcal{O} \left( \log n \right)$. 

It has been observed in Ref.~\cite{barthe2025quantumadvantagelearningquantum} that the aforementioned logarithmic scaling of $d \left( n \right)$ is a fundamental limitation imposed by complexity theory. Specifically, this follows from the results of Arunachalam et al.~\cite{arunachalam2021quantum}, which showed conditional quantum hardness for learning certain classes of shallow classical circuits. More precisely, they define the following concept class:  
\begin{align}
    \mathcal{C}_{n} \subseteq \left\{ c \colon \left\{ 0, 1 \right\}^{n} \rightarrow \left\{ 0, 1 \right\} \right\},
\end{align}
where $\mathcal{C}_{n}$ corresponds to the circuit classes $\mathsf{TC}^{0}, \mathsf{AC}^{0}$, and $\mathsf{TC}_{2}^{0}$. Then, working within the framework of Ref.~\cite{bshouty1995learning}, they consider a \emph{quantum} PAC-learner $\mathcal{A}$, which is given examples of the form 
\begin{align}
    \lvert \psi_{c, \mathcal{D}} \rangle \coloneqq \sum_{x \in \left\{ 0, 1 \right\}^{n}} \sqrt{\mathcal{D} \left( x \right)} \lvert x, c \left( x \right) \rangle, 
\end{align}
where $\mathcal{D} \colon \left\{ 0, 1 \right\}^{n} \rightarrow \left[ 0, 1 \right]$ is an unknown distribution, taken to be fixed in Ref.~\cite{arunachalam2021quantum} as opposed to arbitrary, and the learning criteria to be satisfied by $\mathcal{A}$ are analogous to those in Def.~\ref{def:eff-pac-learnability}. Additionally provided in their setting is access to a quantum membership oracle $O_{c} \colon \lvert x, b \rangle \mapsto \lvert x, b \oplus c \left( x \right) \rangle$, where $b \in \left\{ 0, 1 \right\}$. Note that the quantum PAC-learning framework assumes stronger access to training data than its classical counterpart, since examples are provided in the form of the coherent superposition $\sum_{x} \sqrt{\mathcal{D} \left( x \right)} \lvert x, c \left( x \right) \rangle$, so that one exactly recovers the ordinary classical example $\left( x, c \left( x \right) \right)$, with $x \sim \mathcal{D}$, via measurement in the computational basis. Consequently, one can simulate classical PAC-learning within this framework, and it follows that hardness results for the quantum PAC-learning setting also apply to classical PAC-learners. 

We briefly recall the hardness results of Ref.~\cite{arunachalam2021quantum}. There, assuming that \emph{ring-learning with errors} ($\mathsf{RLWE}$) cannot be solved in quasi-polynomial-time using quantum computers~\cite{lyubashevsky2010ideal}, the authors showed that efficient (weak) quantum PAC-learners do not exist for polynomial-size $\mathsf{TC}^{0}$ circuits. Furthermore, under stronger sub-exponential $\mathsf{RLWE}$ assumptions, polynomially sized $\mathsf{AC}^{0}$ circuits were also shown to be unlearnable using $n^{\mathcal{O} \left( \log^{\nu}  \hspace{-1mm}n \right)}$-time quantum PAC-learners, with $\nu$ being some constant. Additionally, under $\mathsf{LWE}$ hardness, depth-$2$ circuits contained in $\mathsf{TC}^{0}_{2}$ were found to not be weakly quantum PAC-learnable in polynomial-time. A sufficiently general PQC, i.e., one with polynomially many unknown parameters, or a Hamiltonian containing $\mathcal{O}\left( n \right)$-many terms whose time-evolution can be compiled into such a circuit, is capable of encoding these shallow circuit classes through a single expectation value obtained by measuring $Z_{1}$ on the output qubit. Hence, an efficient quantum learner capable of solving the learning task of Ref.~\cite{barthe2025quantumadvantagelearningquantum} in full generality would violate the hardness results of Ref.~\cite{arunachalam2021quantum}. Our work avoids this obstruction by considering a different form of training data, in which the labels are given by a richer vector of expectation values chosen to enable Hamiltonian learning.

Finally, in Def.~\ref{def:hamiltonian-dynamics-concept-class-log-pac-learnable}, it should be noted that the unitary $U_{n} \left( x, \lambda \right) \coloneqq e^{-iH_{n} \left( x, \lambda \right) \tau}$ corresponds to time-evolution by a fixed time $\tau$, so that the training data considered in Ref.~\cite{barthe2025quantumadvantagelearningquantum} are expectation values of states that have been time-evolved by $\tau$. This is unlike the setting considered in our work, since our training samples correspond to a uniform distribution over the interval $\left[ 0, T \right]$, with $T \in \mathcal{O} \left( \mathsf{poly} \left( n \right) \right)$.

\subsection{Summary and future directions}
\label{subsec:future-directions}

To summarize, in this work, we have formulated a physically motivated learning problem in which the goal is to predict expectation values of quantum states that have been time-evolved by an unknown, low-intersection Hamiltonian. Assuming $\mathsf{BQP} \not\subseteq \mathsf{P/poly}$ to hold, we prove that the aforementioned task demonstrates an exponential quantum-classical learning separation with respect to a certain family of input distributions. We show the existence of an efficient quantum learner, whose training phase involves a simplified version of the Hamiltonian learning protocol of Ref.~\cite{Haah_2024}, while its inference stage relies on Hamiltonian simulation and the classical shadows formalism. On the other hand, to prove classical hardness, we use a low-intersection variant of the Feynman-Kitaev clock Hamiltonian, whose polynomial-time dynamics encode $\mathsf{BQP}$-complete computations.
In this sense, our work shows a proper quantum-classical separation for a meaningful learning task that is naturally motivated by a plausible physical application. We finally conclude by listing a few compelling future directions and questions that have been left open in our work. 

\vspace{2.5mm}
\noindent\textbf{Applicability of long-time Hamiltonian learning.} Currently, our approach in this work relies on short-time samples to perform Hamiltonian learning in the training phase. One could also consider whether the unknown Hamiltonian can be learned within a setting involving long-time samples only. In this regard, we highlight the recent works of Refs.~\cite{shin2026heisenberglimitedhamiltonianlearningshorttime, depradenne2026learninghamiltonianslongtimes}. The former shows that Heisenberg-limited Hamiltonian learning can be performed even without access to short-time control; however, their method interleaves queries to long-time Hamiltonian evolution with certain learned auxiliary operations, which does not correspond exactly to our setting. 

The latter work of Ref.~\cite{depradenne2026learninghamiltonianslongtimes} studies the recoverability of local Hamiltonians from arbitrarily long times, and provides \emph{average-case identifiability} guarantees over certain local, random ensembles of Hamiltonians. However, whether their methods apply to the PAC-learning setting considered here, where the concept class may contain instances corresponding to worst-case Hamiltonian instances such as those in $\mathsf{H}_{\mathsf{BQP}}$, remains unclear. Nonetheless, we view the question of whether long-time Hamiltonian learning algorithms akin to Refs.~\cite{shin2026heisenberglimitedhamiltonianlearningshorttime, depradenne2026learninghamiltonianslongtimes} can be applied to our supervised learning problem as an interesting future direction and worthy of further study.

\vspace{2.5mm}
\noindent\textbf{Extensions to other physical settings of interest.} Obvious extensions of our work would be to show learning separations for predicting ground and thermal state properties. In particular, for the analogous scenario involving Gibbs states, it is foreseeable that a similar approach could hold for some distribution over inverse temperatures $\beta \sim \left[ \beta_{\text{low}}, \beta_{\text{high}} \right]$, with high (or constant) temperatures corresponding to samples enabling Hamiltonian learning~\cite{Haah_2024, bakshi2023learningquantumhamiltonianstemperature, narayanan2024improvedalgorithmslearningquantum} and low temperatures ensuring classical hardness~\cite{Rouz__2026, Bravyi_2022, bravyi2024quantumcomplexitykroneckercoefficients}. Other interesting settings could possibly include states evolved by processes more general than unitary time-evolution, such as Lindbladians, for which learning algorithms also exist~\cite{liu2025robustlindbladianestimationquantum, Onorati_2023, Bairey_2020, franca2022efficientrobustestimationmanyqubit, ivashkov2026ansatzfreelearninglindbladiandynamics, cai2026optimaldetectiondissipationlindbladian}. 

\vspace{2.5mm}
\noindent\textbf{Identification-based learning separations.} We stress here again that a limitation of our construction lies in the fact that our labels are vector-valued lists of expectation values, corresponding to sufficient statistics from which the unknown Hamiltonian can be recovered via elementary classical post-processing. Thus, our learner obtains an \emph{explicit} description of the target concept in a manner that is entirely classical. Therefore, our learning separation does not arise from the hardness of identifying the data-generating concept. Instead, the classical hardness enters at the evaluation stage: even after the unknown Hamiltonian has been identified in a proper PAC sense, evaluating the learned concept on potentially long-time inputs requires simulating time-dynamics that may encode $\mathsf{BQP}$-complete computations. Thus, a pressing open direction is the formulation of physically inspired learning tasks similar to the one considered in the present work, where the learning separation stems from the hardness of identification as opposed to evaluation.

Another future direction of interest could thus be to consider a similar learning task to that considered in this work, but one involving more frugal training data and possibly involving a combination of the setting of the present work with the PQC learning-based approach of Ref.~\cite{barthe2025quantumadvantagelearningquantum}. Alternatively, one could contrive constructions wherein only part of the Hamiltonian is identifiable from short-time dynamics, with additional parameters or auxiliary subsystems ``switching on" only after long evolution times. Such settings could potentially exhibit identification-based learning separations.

\vspace{2.5mm}
\noindent\textbf{Certified quantum simulation.} 
One related interpretation of our results is in learning-assisted certified quantum simulation. Rather than certifying long-time quantum dynamics directly, we learn the governing Hamiltonian from experimentally accessible short-time data and use it for long-time quantum prediction, obtaining rigorous learning-theoretic and complexity-theoretic guarantees.

\vspace{2.5mm}
\noindent\textbf{Connections to cryptographic notions.} 
Another intriguing implication of our work 
arises in the context of quantum cryptography. While our result does not provide a construction of one-way functions, it identifies a physically motivated family of quantum-generated prediction tasks that are efficiently solvable on a quantum computer while remaining inaccessible to efficient classical learners. In this sense, our work lies near the conceptual boundary explored in the emerging area of Microcrypt \cite{Microcrypt}: quantum processes may give rise to useful asymmetric computational tasks without necessarily implying classical one-way functions.

\vspace{2.5mm}
\noindent\textbf{Complexity-theoretic no-gos.} Finally, an exciting direction would be to place fundamental limitations on the learnability of time-dynamics. That is, it would be of interest to explore an intermediate regime between labels that are singular expectation values, as 
in Ref.~\cite{barthe2025quantumadvantagelearningquantum}, and ones that amount to maximally informative classical descriptions of probe states. In other words, one could try to prove hardness results of the type shown in Ref.~\cite{arunachalam2021quantum}, while considering different generalizations of their setting, such as input states that go beyond computational basis states, like stabilizer states, Choi states, and so on, as well as measurement of observables more general than single-qubit observables.

\section{Acknowledgments}
\label{sec:acknowledgements}

This project was funded by the European Union. Views and opinions expressed are however those of the author(s) only and do not necessarily reflect those of the European Union or the European Research Council Executive Agency. Neither the European Union nor the granting authority can be held responsible for them. This work is supported by ERC grant BeMAIQuantum (grant agreement No. 101124342; project DOI: 10.3030/101124342) and the Dutch National Growth Fund (NGF) as part of the QDNL programme. This publication is part of the NWA research program “Research on Routes by Consortia (ORC),” financed by the Dutch Research Council (NWO), through the projects ``Divide \& Quantum" (Project No. 1389.20.241) and ``Quantum Inspire– the Dutch Quantum Computer in the Cloud" (Project No. NWA.1292.19.194). The Berlin team would like to acknowledge support from the BMFTR (MUNIQC-Atoms, QuSol, Hybrid++), BIFOLD, the DFG (CRC 
183 and SPP 2514), the Quantum Flagship (Millenion, PasQuanS2), the Munich Quantum Valley, Berlin Quantum, and the European Research Council (DebuQC) for financial support.

\bibliography{main, tail_bound, BigReferences70}

@Misc{r,
  title                     = {{The Complexity of Translationally-Invariant Spin Chains with Low Local Dimension}},

  Author                   = {Bausch, Johannes and Cubitt, Toby and Ozols, Maris},
  Note                     = {arXiv:1605.01718}
}

@Article{BlochSimulation,
  title                     = {Quantum simulations with ultracold quantum gases},
  Author                   = {I. Bloch and J. Dalibard and S. Nascimbene},
  Journal                  = {Nature Phys.},
doi={10.1038/nphys2259},
  Year                     = {2012},
  Pages                    = {267},
  Volume                   = {8}
}

@article{MindTheGaps,
      title={Mind the gaps: The fraught road to quantum advantage}, 
      author={Jens Eisert and John Preskill},
      year={2026},
      eprint={2510.19928},
      archivePrefix={arXiv},
      url={https://arxiv.org/abs/2510.19928},
      journal={arXiv} 
}

@article{WildeLearning,
author={F. Wilde and A. Kshetrimayum and I. Roth and D. Hangleiter and R. Sweke and J. Eisert}, 
title={{Scalably learning quantum many-body Hamiltonians from dynamical data}},
  year = {2026},
  journal={Mach. Learn. Sc. Tech.},
  DOI={10.1088/2058-9565/ae6fe3},
  volume=11, pages={035002}}

@article{GAN14, author={I. M. Georgescu and S. Ashhab and F. Nori}, 
  title ={Quantum simulation}, journal={Rev. Mod. Phys.},volume= 86, 
pages=153,
  doi = {10.1103/RevModPhys.86.153},
year=2014}

@Article{CiracZollerSimulation,
  title                     = {Goals and opportunities in quantum simulation},
  Author                   = {J. I. Cirac and P. Zoller},
  doi={10.1038/nphys2275},  
  Journal                  = {Nature Phys.},
  Year                     = {2012},
  Pages                    = {264},
  Volume                   = {8}
}

@article{PACLearning,
journal={Quantum},
 year=2021,
 volume=
  5, pages=417,
doi={10.22331/q-2021-03-23-417},
    title={On the quantum versus classical learnability of discrete distributions},
    author={R. Sweke and J.-P. Seifert and D. Hangleiter and J. Eisert}}

@article{OptimizationAdvantages,
	journal={Science Adv.},volume= 10, pages={eadj5170} , year=2024,
doi={10.1126/sciadv.adj517},
title={An in-principle super-polynomial quantum advantage for approximating combinatorial optimization problems via computational learning theory},
author={N. Pirnay and V. Ulitzsch and F. Wilde and J. Eisert and J.-P. Seifert}}

@Article{Single,
  Author                   = {J. Eisert and M. Cramer},
  Journal                  = {Phys. Rev. A},
  Year                     = {2005},
  Pages                    = {042112},
  Volume                   = {72}
}

@article{Molteni,
  title={Exponential quantum advantages in learning quantum observables from classical data},
  author={Molteni, Riccardo and Gyurik, Casper and Dunjko, Vedran},
  volume=12, 
  pages=19,
  DOI={10.1038/s41534-025-01162-2},
  journal={npj Quant. Inf.},
  year={2026},
  publisher={Nature Publishing Group UK London}
}

@article{Nori,
   title={Quantum simulation},
   volume={86},
   ISSN={1539-0756},
   url={http://dx.doi.org/10.1103/RevModPhys.86.153},
   DOI={10.1103/revmodphys.86.153},
   number={1},
   journal={Rev. Mod. Phys.},
   publisher={American Physical Society (APS)},
   author={Georgescu, I. M. and Ashhab, S. and Nori, Franco},
   year={2014},
   month=Mar, pages={153–185} }

@article{SampleEfficientHamiltonianLearning,
title={Sample-efficient learning of interacting quantum systems},
author={Anurag Anshu and Srinivasan Arunachalam and Tomotaka Kuwahara and Mehdi Soleimanifar},
doi={10.1038/s41567-021-01232-0},
journal={Nature Phys.},volume=17, pages={931}, year=2021}

@Article{Lloyd,
  title                     = {Universal quantum simulators},
  Author                   = {S. Lloyd},
  doi={10.1126/science.273.5278.1073},
  Journal                  = {Science},
  Year                     = {1996},
  Pages                    = {1073},
  Volume                   = {273}
}

@Article{TemmeML,
title={A rigorous and robust quantum speed-up in supervised machine learning},
author={Y. Liu and S. Arunachalam  and K. Temme},
doi={10.1038/s41567-021-01287-z},
journal={Nature Phys.},volume=17, pages={1013},  year=2021}

@article{ProvablyEfficientQML,
title={Provably efficient machine learning for quantum many-body problems},
author={Hsin-Yuan Huang and Richard Kueng and Giacomo Torlai and Victor V. Albert and John Preskill},
journal={Science}, Volume={377}, pages=6613, year=2022,
DOI={10.1126/science.abk3333}}

@article{PhysRevLett.126.190505,
  title = {Information-Theoretic Bounds on Quantum Advantage in Machine Learning},
  author = {Huang, Hsin-Yuan and Kueng, Richard and Preskill, John},
  journal = {Phys. Rev. Lett.},
  volume = {126},
  issue = {19},
  pages = {190505},
  numpages = {7},
  year = {2021},
  month = {May},
  publisher = {American Physical Society},
  doi = {10.1103/PhysRevLett.126.190505},
  url = {https://link.aps.org/doi/10.1103/PhysRevLett.126.190505}
}

@article{VedranExponential,
title={Exponential quantum advantages in learning quantum observables from classical data},
author={R. Molteni and C. Gyurik and V. Dunjko}, 
doi={10.1038/s41534-025-01162-2},
journal={npj Quant. Inf.}, volume=12, pages=19, year=2026}

@Article{Shor-1994,
  title                     = {Algorithms for quantum computation: discrete logarithms and factoring},
  Author                   = {Shor, P. W.},
  pages={124-134},
  Journal                  = {Proc. 50th Ann. Symp. Found.  Comp. Sc.},
  Year                     = {1994},
  Doi                      = {10.1109/sfcs.1994.365700}
}

@Article{Time,
  title                     = {Topological characterization of periodically driven quantum systems},
  Author                   = {T. Kitagawa, T. and E. Berg and M. Rudne and E. Demler},
  Journal                  = {Phys. Rev. B},
  Year                     = {2010},
  Pages                    = {235114},
  Volume                   = {82}
}

@Article{Will,
  title                     = {{Time-resolved observation of coherent multi-body interactions in quantum phase revivals}},
  Author                   = {Will, Sebastian and Best, Thorsten and Schneider, Ulrich and HackermÃÂÃÂÃÂÃÂ¼ller, Lucia and LÃÂÃÂÃÂÃÂ¼hmann, Dirk-S\"oren and Bloch, I. },
  Journal                  = {Nature},
  Year                     = {2010},
  Pages                    = {197-201},
  Volume                   = {465}
}

@article{JunyuPruned,
title={Towards provably efficient quantum algorithms for large-scale machine-learning models},
author={Junyu Liu and Minzhao Liu and Jin-Peng Liu and Ziyu Ye and Yunfei Wang and Yuri Alexeev and Jens Eisert and Liang Jiang},
journal={Nature Comm.},volume= 15, pages=434 , year=2024,
DOI={10.1038/s41467-023-43957-x}}

@article{ZollerLearning,
  eprint = {2511.23392},
  year=2025,
  archiveprefix = {arXiv},
  title={{Bounded-error quantum simulation via Hamiltonian and Lindbladian learning}},
author={Tristan Kraft and Manoj K. Joshi and William Lam and Tobias Olsacher and Florian Kranzl and Johannes Franke and Lata Kh Joshi and Rainer Blatt and Augusto Smerzi and Daniel Stilck França and Benoît Vermersch and Barbara Kraus and Christian F. Roos and Peter Zoller},
      journal={arXiv} }

@article{Microcrypt,
  eprint = {2505.14461},
  year=2025,
  title={{MicroCrypt assumptions with quantum input sampling and pseudodeterminism: Constructions and separations}},
author={Mohammed Barhoush and Ryo Nishimaki and Takashi Yamakawa},
      journal={arXiv} }

@article{DecodedQuantumInterferometry,
   title={Optimization by decoded quantum interferometry},
   volume={646},
   ISSN={1476-4687},
   url={http://dx.doi.org/10.1038/s41586-025-09527-5},
   DOI={10.1038/s41586-025-09527-5},
   number={8086},
   journal={Nature},
   publisher={Springer Science and Business Media LLC},
   author={Jordan, Stephen P. and Shutty, Noah and Wootters, Mary and Zalcman, Adam and Schmidhuber, Alexander and King, Robbie and Isakov, Sergei V. and Khattar, Tanuj and Babbush, Ryan},
   year={2025},
   month=Oct, pages={831–836} }

@Article{Trotzky,
  title                     = {Probing the relaxation towards equilibrium in an isolated strongly correlated one-dimensional {B}ose gas},
  Author                   = {S. Trotzky and Y.-A. Chen and A. Flesch and I. P. McCulloch and U. Schollw\"ock and J. Eisert and I. Bloch},
  Journal                  = {Nature Phys.},
  Year                     = {2012},
  Pages                    = {325-330},
  Volume                   = {8},
  eprint = {1101.2659},
  archiveprefix = {arXiv},
  Doi                      = {doi:10.1038/nphys2232}
}

@article{VastWorld,
      title={The vast world of quantum advantage}, 
      author={Hsin-Yuan Huang and Soonwon Choi and Jarrod R. McClean and John Preskill},
      year={2025},
      eprint={2508.05720},
      archivePrefix={arXiv},
      optprimaryClass={quant-ph},
      url={https://arxiv.org/abs/2508.05720},
      journal={arXiv} 
}

@article{OptimizationReview,
title={Challenges and opportunities in quantum optimization},
author={Amira Abbas and Andris Ambainis and Brandon Augustino and Andreas B{\"a}rtschi and Harry Buhrman and Carleton Coffrin and Giorgio Cortiana and Vedran Dunjko and Daniel J. Egger and Bruce G. Elmegreen and Nicola Franco and Filippo Fratini and Bryce Fuller and Julien Gacon and Constantin Gonciulea and Sander Gribling and Swati Gupta and Stuart Hadfield and Raoul Heese and Gerhard Kircher and Thomas Kleinert and Thorsten Koch and Georgios Korpas and Steve Lenk and Jakub Marecek and Vanio Markov and Guglielmo Mazzola and Stefano Mensa and Naeimeh Mohseni and Giacomo Nannicini and Corey O'Meara and Elena Pe{\~n}a Tapia and Sebastian Pokutta and Manuel Proissl and Patrick Rebentrost and Emre Sahin and Benjamin C. B. Symons and Sabine Tornow and Victor Valls and Stefan Woerner and Mira L. Wolf-Bauwens and Jon Yard and Sheir Yarkoni and Dirk Zechiel and Sergiy Zhuk and Christa Zoufal},
doi={10.1038/s42254-024-00770-9},
journal={Nature Rev. Phys.}, volume=6, pages={718},year=2024}

@article{Training,
title={Training Schr{\"o}dinger?s cat: quantum optimal control},
Author={S. J. Glaser and U. Boscain and T. Calarco and C. P. 
Koch and W. K{\"o}ckenberger and R. Kosloff and I. Kuprov and B. Luy 
and S. Schirmer and T. Schulte-Herbr{\"u}ggen and D. 
Sugny and F. K. Wilhelm},
year=2015,journal={Europ. Phys. J. D},
Volume=69,
Pages=
279}

@article{RaulStability,
archiveprefix = {arXiv},
  eprint = {2510.08467},
year=2025,
title={Stability of digital and analog quantum simulations under noise},
author={Jayant Rao and Jens Eisert and Tommaso Guaita},
      journal={arXiv}}

@article{HamiltonianLearning,
journal = {Nature Comm.},volume=  15, pages={9595}, year= 2024,
title = {{Precise Hamiltonian identification of a superconducting quantum processor}},
doi = {10.1038/s41467-024-52629-3},
author = {D. Hangleiter and I. Roth and J. Fuksa and J. Eisert and P. Roushan}
}

@article{ProbabilityDiffusion,
      title={Towards efficient quantum algorithms for diffusion probabilistic models}, 
      author={Yunfei Wang and Ruoxi Jiang and Yingda Fan and Xiaowei Jia and Jens Eisert and Junyu Liu and Jin-Peng Liu},
      year={2025},
      eprint={2502.14252},
      archivePrefix={arXiv},
      url={https://arxiv.org/abs/2502.14252},
      journal={arXiv}  
}

@article{Shadows,
    title={Predicting many properties of a quantum system from very few measurements},
    author={H.-Y. Huang and R. Kueng and J. Preskill},
    volume=16, pages={1050},
    DOI={10.1038/s41567-020-0932-7},
journal={Nature Phys.}, year=2020}

@article{molteni2024exponential,
  title={Exponential quantum advantages in learning quantum observables from classical data},
  author={Molteni, Riccardo and Gyurik, Casper and Dunjko, Vedran},
  archiveprefix = {arXiv},
  eprint = {2405.02027},
year=2024,
      journal={arXiv} 
}

@article{oliveira2005complexity,
  title={The complexity of quantum spin systems on a two-dimensional square lattice},
  author={Oliveira, Roberto and Terhal, Barbara M.},
  archiveprefix = {arXiv},
  eprint = {quant-ph/0504050},
  year={2005},
journal={arXiv}
}

@article{feynman1986quantum,
  title={Quantum mechanical computers},
  author={Feynman, Richard P},
  journal={Found.
  Phys.},
  volume={16},
  number={6},
  pages={507--531},
  year={1986},
  DOI={10.1007/BF01886518},
  publisher={Kluwer Academic Publishers-Plenum Publishers New York}
}

@article{nagaj2010fast,
  title={Fast universal quantum computation with railroad-switch local Hamiltonians},
  author={Nagaj, Daniel},
  journal={J. Math. Phys.},
  DOI={10.1063/1.3384661
},
  volume={51},
     pages=062201,
  year={2010},
  publisher={AIP Publishing}
}

@article{schapire1990strength,
  title={The strength of weak learnability},
  author={Schapire, Robert E},
  DOI={10.1007/BF00116037},
  journal={Mach. Learn.},
  volume={5},
  pages={197--227},
  year={1990},
  publisher={Springer}
}

@article{flannigan2022propagation,
  title={Propagation of errors and quantitative quantum simulation with quantum advantage},
  author={Flannigan, Stuart and Pearson, Natalie and Low, Guang Hao and Buyskikh, Anton and Bloch, Immanuel and Zoller, Peter and Troyer, Matthias and Daley, Andrew J.},
  journal={Quant. Sc. Tech.},
  doi={10.1088/2058-9565/ac88f5},
  volume={7},
  number={4},
  pages={045025},
  year={2022},
  publisher={IOP Publishing}
}

@article{Haah_2024,
   title={Learning quantum Hamiltonians from high-temperature Gibbs states and real-time evolutions},
   volume={20},
   ISSN={1745-2481},
   url={http://dx.doi.org/10.1038/s41567-023-02376-x},
   DOI={10.1038/s41567-023-02376-x},
   number={6},
   journal={Nature Phys.},
   publisher={Springer Science and Business Media LLC},
   author={Haah, Jeongwan and Kothari, Robin and Tang, Ewin},
   year={2024},
   month=mar, pages={1027–1031} }

@article{bakshi2023learningquantumhamiltonianstemperature,
      title={Learning quantum Hamiltonians at any temperature in polynomial time}, 
      author={Ainesh Bakshi and Allen Liu and Ankur Moitra and Ewin Tang},
      year={2023},
      eprint={2310.02243},
      archivePrefix={arXiv},
      optprimaryClass={quant-ph},
      url={https://arxiv.org/abs/2310.02243},
      journal={arXiv}  
}

@misc{feldman2007,
	author = {Joel Feldman},
	title = {{D}uhamel's {F}ormula},
	year = {2007},
	note = {[Accessed 19-03-2026]},
    url = {https://personal.math.ubc.ca/~feldman/m428/duhamel.pdf}
}

@article{Gu_2024,
   title={Practical Hamiltonian learning with unitary dynamics and Gibbs states},
   volume={15},
   DOI={10.1038/s41467-023-44008-1},
   journal={Nature Comm.},
   publisher={Springer Science and Business Media LLC},
   pages=312,
   DOI={10.1038/s41467-023-44008-1},
   author={Gu, Andi and Cincio, Lukasz and Coles, Patrick J.},
   year={2024}}

@book{arora2009computational,
  title={Computational complexity: a modern approach},
  author={Arora, Sanjeev and Barak, Boaz},
  year={2009},
  publisher={Cambridge University Press}
}

@article{Lewis_2024,
   title={Improved machine learning algorithm for predicting ground state properties},
   volume={15},
   DOI={10.1038/s41467-024-45014-7},
   number={1},
   journal={Nature Comm.},
   publisher={Springer Science and Business Media LLC},
   author={Lewis, Laura and Huang, Hsin-Yuan and Tran, Viet T. and Lehner, Sebastian and Kueng, Richard and Preskill, John},
   pages=895,
   year={2024}}

@article{wanner2024predictinggroundstateproperties,
      title={Predicting Ground State Properties: Constant Sample Complexity and Deep Learning Algorithms}, 
      author={Marc Wanner and Laura Lewis and Chiranjib Bhattacharyya and Devdatt Dubhashi and Alexandru Gheorghiu},
      year={2024},
      eprint={2405.18489},
      archivePrefix={arXiv},
      optprimaryClass={quant-ph},
      url={https://arxiv.org/abs/2405.18489},
      journal={arXiv}
}

@article{Rouz__2024,
   title={Efficient learning of ground and thermal states within phases of matter},
   volume={15},
   url={http://dx.doi.org/10.1038/s41467-024-51439-x},
   DOI={10.1038/s41467-024-51439-x},
   pages=7755,
   journal={Nature Comm.},
   author={Rouz\'e, Cambyse and Stilck Fran\c{c}a, Daniel and Onorati, Emilio and Watson, James D.},
   year={2024}}

@article{onorati2023provablyefficientlearningphases,
      title={Provably Efficient Learning of Phases of Matter via Dissipative Evolutions}, 
      author={Emilio Onorati and Cambyse Rouz\'e and Daniel Stilck Fran\c{c}a and James D. Watson},
      year={2023},
      eprint={2311.07506},
      archivePrefix={arXiv},
      optprimaryClass={quant-ph},
      journal={arXiv}
}

@article{Coser_2019,
   title={Classification of phases for mixed states via fast dissipative evolution},
   volume={3},
   ISSN={2521-327X},
   url={http://dx.doi.org/10.22331/q-2019-08-12-174},
   DOI={10.22331/q-2019-08-12-174},
   journal={Quantum},
   publisher={Verein zur Forderung des Open Access Publizierens in den Quantenwissenschaften},
   author={Coser, Andrea and Pérez-García, David},
   year={2019},
   month=Aug, pages={174} }

@article{Huang_2021,
   title={Information-Theoretic Bounds on Quantum Advantage in Machine Learning},
   volume={126},
   ISSN={1079-7114},
   url={http://dx.doi.org/10.1103/PhysRevLett.126.190505},
   DOI={10.1103/physrevlett.126.190505},
   pages=190505,
   journal={Phys. Rev. Lett.},
   publisher={American Physical Society (APS)},
   author={Huang, Hsin-Yuan and Kueng, Richard and Preskill, John},
   year={2021}}

@article{chen2021exponentialseparationslearningquantum,
      title={Exponential separations between learning with and without quantum memory}, 
      author={Sitan Chen and Jordan Cotler and Hsin-Yuan Huang and Jerry Li},
      year={2021},
      eprint={2111.05881},
      archivePrefix={arXiv},
      optprimaryClass={quant-ph},
      url={https://arxiv.org/abs/2111.05881},
      journal={arXiv}  
}

@article{Huang_2022_quantum_adv_exp,
   title={Quantum advantage in learning from experiments},
   volume={376},
   ISSN={1095-9203},
   url={http://dx.doi.org/10.1126/science.abn7293},
   DOI={10.1126/science.abn7293},
   number={6598},
   journal={Science},
   publisher={American Association for the Advancement of Science (AAAS)},
   author={Huang, Hsin-Yuan and Broughton, Michael and Cotler, Jordan and Chen, Sitan and Li, Jerry and Mohseni, Masoud and Neven, Hartmut and Babbush, Ryan and Kueng, Richard and Preskill, John and McClean, Jarrod R.},
   year={2022},
   month={June}, 
   pages={1182–1186} 
}

@article{Aharonov_2022,
   title={Quantum algorithmic measurement},
   volume={13},
   pages=887,
   DOI={10.1038/s41467-021-27922-0},
   journal={Nature Comm.},
   publisher={Springer Science and Business Media LLC},
   author={Aharonov, Dorit and Cotler, Jordan and Qi, Xiao-Liang},
   year={2022},
   month=Feb }

@article{gyurik2024exponentialseparationsclassicalquantum,
      title={Exponential separations between classical and quantum learners}, 
      author={Casper Gyurik and Vedran Dunjko},
      year={2024},
      eprint={2306.16028},
      archivePrefix={arXiv},
      optprimaryClass={quant-ph},
      url={https://arxiv.org/abs/2306.16028},
      journal={arXiv}  
}

@article{barthe2025quantumadvantagelearningquantum,
      title={Quantum Advantage in Learning Quantum Dynamics via Fourier coefficient extraction}, 
      author={Alice Barthe and Mahtab Yaghubi Rad and Michele Grossi and Vedran Dunjko},
      year={2025},
      eprint={2506.17089},
      archivePrefix={arXiv},
      optprimaryClass={quant-ph},
      url={https://arxiv.org/abs/2506.17089},
      journal={arXiv}
}

@article{bokov2026machinelearningminimaluse,
      title={Machine learning with minimal use of quantum computers: Provable advantages in Learning Under Quantum Privileged Information (LUQPI)}, 
      author={Vasily Bokov and Lisa Kohl and Sebastian Schmitt and Vedran Dunjko},
      year={2026},
      eprint={2601.22006},
      archivePrefix={arXiv},
      optprimaryClass={quant-ph},
      url={https://arxiv.org/abs/2601.22006},
      journal={arXiv}
}

@article{caves1981quantum,
  title={Quantum-mechanical noise in an interferometer},
  author={Caves, Carlton M},
  journal={Phys. Rev. D},
  volume={23},
  pages={1693},
  year={1981},
  DOI={10.1103/PhysRevD.23.1693},
  publisher={APS}
}

@article{holland1993interferometric,
  title={Interferometric detection of optical phase shifts at the Heisenberg limit},
  author={Holland, Murray J and Burnett, Keith},
  journal={Phys. Rev. Lett.},
  doi={10.1103/PhysRevLett.71.1355},
  volume={71},
  pages={1355},
  year={1993},
  publisher={APS}
}

@article{Lee_2002,
   title={A quantum Rosetta stone for interferometry},
   volume={49},
   ISSN={1362-3044},
   url={http://dx.doi.org/10.1080/0950034021000011536},
   DOI={10.1080/0950034021000011536},
   number={14-15},
   journal={J. Mod. Opt.},
   publisher={Informa UK Limited},
   author={Lee, Hwang and Kok, Pieter and Dowling, Jonathan P.},
   year={2002},
   month=Nov, pages={2325–2338} }

@article{giovannetti2004quantum,
  title={Quantum-enhanced measurements: beating the standard quantum limit},
  author={Giovannetti, Vittorio and Lloyd, Seth and Maccone, Lorenzo},
  journal={Science},
   DOI={10.1126/science.1104149},
  volume={306},
  number={5700},
  pages={1330--1336},
  year={2004},
  publisher={American Association for the Advancement of Science}
}

@article{de_Burgh_2005,
   title={Quantum methods for clock synchronization: Beating the standard quantum limit without entanglement},
   volume={72},
   pages=042301,
   DOI={10.1103/physreva.72.042301},
   number={4},
   journal={Phys. Rev. A},
   publisher={American Physical Society (APS)},
   author={de Burgh, Mark and Bartlett, Stephen D.},
   year={2005},
   month=Oct }

@article{degen2017quantum,
  title={Quantum sensing},
  author={Degen, Christian L and Reinhard, Friedemann and Cappellaro, Paola},
  journal={Rev. Mod. Phys.},
  volume={89},
  DOI={10.1103/RevModPhys.89.035002},
  pages={035002},
  year={2017},
  publisher={APS}
}

@article{vapnik2009new,
  title={A new learning paradigm: Learning using privileged information},
  author={Vapnik, Vladimir and Vashist, Akshay},
  journal={Neur. Net.},
  volume={22},
  DOI={10.1016/j.neunet.2009.06.042},
  number={5-6},
  pages={544--557},
  year={2009},
  publisher={Elsevier}
}

@article{vapnik2015learning,
  title={Learning using privileged information: similarity control and knowledge transfer},
  author={Vapnik, Vladimir and Izmailov, Rauf},
  DOI={10.1007/s10994-025-06773-6},
  journal={J. Mach. Learn. Res.},
  volume={16},
  number={1},
  pages={2023--2049},
  year={2015},
  publisher={JMLR. org}
}

@article{Bairey_2019,
   title={Learning a Local Hamiltonian from Local Measurements},
   volume={122},
   DOI={10.1103/physrevlett.122.020504},
   pages=020504,
   journal={Phys. Rev. Lett.},
   publisher={American Physical Society (APS)},
   author={Bairey, Eyal and Arad, Itai and Lindner, Netanel H.},
   year={2019},
   month=Jan }

@article{Garrison_2018,
   title={Does a Single Eigenstate Encode the Full Hamiltonian?},
   volume={8},
   DOI={10.1103/physrevx.8.021026},
  pages=021026,
   journal={Phys. Rev. X},
   publisher={American Physical Society (APS)},
   author={Garrison, James R. and Grover, Tarun},
   year={2018},
   month=Apr }

@article{Li_2020,
   title={Hamiltonian Tomography via Quantum Quench},
   volume={124},
   ISSN={1079-7114},
   url={http://dx.doi.org/10.1103/PhysRevLett.124.160502},
   DOI={10.1103/physrevlett.124.160502},
   pages=160502,
   journal={Phys. Rev. Lett.},
   publisher={American Physical Society (APS)},
   author={Li, Zhi and Zou, Liujun and Hsieh, Timothy H.},
   year={2020},
   month=Apr }

@article{Qi_2019,
   title={Determining a local Hamiltonian from a single eigenstate},
   volume={3},
   ISSN={2521-327X},
   url={http://dx.doi.org/10.22331/q-2019-07-08-159},
   DOI={10.22331/q-2019-07-08-159},
   journal={Quantum},
   publisher={Verein zur Forderung des Open Access Publizierens in den Quantenwissenschaften},
   author={Qi, Xiao-Liang and Ranard, Daniel},
   year={2019},
   month={July}, 
   pages={159} 
}

@article{arunachalam2025testinglearningstructuredquantum,
      title={Testing and learning structured quantum Hamiltonians}, 
      author={Srinivasan Arunachalam and Arkopal Dutt and Francisco Escudero Gutiérrez},
      year={2025},
      eprint={2411.00082},
      archivePrefix={arXiv},
      optprimaryClass={quant-ph},
      url={https://arxiv.org/abs/2411.00082},
      journal={arXiv}  
}

@article{bakshi2026learningquantumhamiltonianstemperature,
      title={Learning quantum Hamiltonians at any temperature in polynomial time}, 
      author={Ainesh Bakshi and Allen Liu and Ankur Moitra and Ewin Tang},
      year={2026},
      eprint={2310.02243},
      archivePrefix={arXiv},
      optprimaryClass={quant-ph},
      url={https://arxiv.org/abs/2310.02243},
      journal={arXiv}  
}

@article{Fawzi_2024,
   title={Certified algorithms for equilibrium states of local quantum Hamiltonians},
   volume={15},
   pages={7394},
   DOI={10.1038/s41467-024-51592-3},
   number={1},
   journal={Nature Comm.},
   publisher={Springer Science and Business Media LLC},
   author={Fawzi, Hamza and Fawzi, Omar and Scalet, Samuel O.},
   year={2024},
   month=Aug }

@article{Garc_a_Pintos_2024,
   title={Estimation of Hamiltonian Parameters from Thermal States},
   volume={133},
   DOI={10.1103/physrevlett.133.040802},
   pages=040802,
   journal={Phys. Rev. Lett.},
   publisher={American Physical Society (APS)},
   author={García-Pintos, Luis Pedro and Bharti, Kishor and Bringewatt, Jacob and Dehghani, Hossein and Ehrenberg, Adam and Yunger Halpern, Nicole and Gorshkov, Alexey V.},
   year={2024} 
}

@article{narayanan2024improvedalgorithmslearningquantum,
      title={Improved algorithms for learning quantum Hamiltonians, via flat polynomials}, 
      author={Shyam Narayanan},
      year={2024},
      eprint={2407.04540},
      archivePrefix={arXiv},
      optprimaryClass={quant-ph},
      url={https://arxiv.org/abs/2407.04540},
      journal={arXiv} 
}

@inproceedings{Bakshi_2024,
   title={Structure Learning of Hamiltonians from Real-Time Evolution},
   url={http://dx.doi.org/10.1109/FOCS61266.2024.00069},
   DOI={10.1109/focs61266.2024.00069},
   booktitle={2024 IEEE 65th Annual Symposium on Foundations of Computer Science (FOCS)},
   publisher={IEEE},
   author={Bakshi, Ainesh and Liu, Allen and Moitra, Ankur and Tang, Ewin},
   year={2024},
   month=Oct, pages={1037–1050} }

@article{Caro_2024,
   title={Learning Quantum Processes and Hamiltonians via the Pauli Transfer Matrix},
   volume={5},
   ISSN={2643-6817},
   url={http://dx.doi.org/10.1145/3670418},
   DOI={10.1145/3670418},
   number={2},
   journal={ACM Trans. Quant. Comp.},
   publisher={Association for Computing Machinery (ACM)},
   author={Caro, Matthias C.},
   year={2024},
   month={June}, 
   pages={1–53} }

@article{da_Silva_2011,
   title={Practical Characterization of Quantum Devices without Tomography},
   volume={107},
   DOI={10.1103/physrevlett.107.210404},
   pages=210404,
   journal={Phys. Rev. Lett.},
   publisher={American Physical Society (APS)},
   author={da Silva, Marcus P. and Landon-Cardinal, Olivier and Poulin, David},
   year={2011}}

@article{Flynn_2022,
   title={Quantum model learning agent: characterisation of quantum systems through machine learning},
   volume={24},
   ISSN={1367-2630},
   url={http://dx.doi.org/10.1088/1367-2630/ac68ff},
   DOI={10.1088/1367-2630/ac68ff},
   number={5},
   journal={New J. Phys.},
   publisher={IOP Publishing},
   author={Flynn, Brian and Gentile, Antonio A and Wiebe, Nathan and Santagati, Raffaele and Laing, Anthony},
   year={2022},
   month=May, pages={053034} }

@article{franca2025learningcertificationlocaltimedependent,
      title={Learning and certification of local time-dependent quantum dynamics and noise}, 
      author={Daniel Stilck Fran\c{c}a and Tim M\"obus and Cambyse Rouz\'e and Albert H. Werner},
      year={2025},
      eprint={2510.08500},
      archivePrefix={arXiv},
      optprimaryClass={quant-ph},
      url={https://arxiv.org/abs/2510.08500},
      journal={arXiv}  
}

@article{Gentile_2021,
   title={Learning models of quantum systems from experiments},
   volume={17},
  pages=837,
   DOI={10.1038/s41567-021-01201-7},
   number={7},
   journal={Nature Phys.},
   publisher={Springer Science and Business Media LLC},
   author={Gentile, Antonio A. and Flynn, Brian and Knauer, Sebastian and Wiebe, Nathan and Paesani, Stefano and Granade, Christopher E. and Rarity, John G. and Santagati, Raffaele and Laing, Anthony},
   year={2021},
   month=Apr, pages={837–843} }

@article{Hu_2025,
   title={Ansatz-Free Hamiltonian Learning with Heisenberg-Limited Scaling},
   pages=040315,
   volume={6},
   DOI={10.1103/j7b8-pb77},
   number={4},
   journal={PRX Quantum},
   publisher={American Physical Society (APS)},
   author={Hu, Hong-Ye and Ma, Muzhou and Gong, Weiyuan and Ye, Qi and Tong, Yu and Flammia, Steven T. and Yelin, Susanne F.},
   DOI={10.1103/j7b8-pb77},
   year={2025}}

@article{Huang_2023_hl,
   title={Learning Many-Body Hamiltonians with Heisenberg-Limited Scaling},
   volume={130},
   DOI={10.1103/physrevlett.130.200403},
   pages=200403,
   journal={Phys. Rev. Lett.},
   publisher={American Physical Society (APS)},
   author={Huang, Hsin-Yuan and Tong, Yu and Fang, Di and Su, Yuan},
   year={2023},
   month=May }

@article{li2023heisenberglimitedhamiltonianlearninginteracting,
      title={Heisenberg-limited Hamiltonian learning for interacting bosons}, 
      author={Haoya Li and Yu Tong and Hongkang Ni and Tuvia Gefen and Lexing Ying},
      year={2023},
      eprint={2307.04690},
      archivePrefix={arXiv},
      optprimaryClass={quant-ph},
      url={https://arxiv.org/abs/2307.04690},
      journal={arXiv} 
}

@article{ma2024learningkbodyhamiltonianscompressed,
      title={Learning $k$-body Hamiltonians via compressed sensing}, 
      author={Muzhou Ma and Steven T. Flammia and John Preskill and Yu Tong},
      year={2024},
      eprint={2410.18928},
      archivePrefix={arXiv},
      optprimaryClass={quant-ph},
      url={https://arxiv.org/abs/2410.18928},
      journal={arXiv}  
}

@article{Mirani_2024,
   title={Learning interacting fermionic Hamiltonians at the Heisenberg limit},
   volume={110},
   DOI={10.1103/physreva.110.062421},
   pages=062421,
   journal={Phys. Rev. A},
   publisher={American Physical Society (APS)},
   author={Mirani, Arjun and Hayden, Patrick},
   year={2024},
   month=Dec }

@article{ni2024quantumhamiltonianlearningfermihubbard,
      title={Quantum Hamiltonian Learning for the Fermi-Hubbard Model}, 
      author={Hongkang Ni and Haoya Li and Lexing Ying},
      year={2024},
      eprint={2312.17390},
      archivePrefix={arXiv},
      optprimaryClass={quant-ph},
      url={https://arxiv.org/abs/2312.17390},
      journal={arXiv}  
}

@article{Odake_2024,
   title={Higher-order quantum transformations of Hamiltonian dynamics},
   volume={6},
   DOI={10.1103/physrevresearch.6.l012063},
  pages={l012063},
   journal={Phys. Rev. Res.},
   publisher={American Physical Society (APS)},
   author={Odake, Tatsuki and Kristjánsson, Hlér and Soeda, Akihito and Murao, Mio},
   year={2024},
   month=Mar }

@article{PhysRev.78.695,
  title = {A Molecular Beam Resonance Method with Separated Oscillating Fields},
  author = {Ramsey, Norman F.},
  journal = {Phys. Rev.},
  volume = {78},
  issue = {6},
  pages = {695--699},
  numpages = {0},
  year = {1950},
  month = {Jun},
  publisher = {American Physical Society},
  doi = {10.1103/PhysRev.78.695},
  url = {https://link.aps.org/doi/10.1103/PhysRev.78.695}
}

@article{PhysRevA.84.012107,
  title = {Estimation of many-body quantum Hamiltonians via compressive sensing},
  author = {Shabani, A. and Mohseni, M. and Lloyd, S. and Kosut, R. L. and Rabitz, H.},
  journal = {Phys. Rev. A},
  volume = {84},
  issue = {1},
  pages = {012107},
  numpages = {8},
  year = {2011},
  month = {Jul},
  publisher = {American Physical Society},
  doi = {10.1103/PhysRevA.84.012107},
  url = {https://link.aps.org/doi/10.1103/PhysRevA.84.012107}
}

@article{sinha2025improvedhamiltonianlearningsparsity,
      title={Improved Hamiltonian learning and sparsity testing through Bell sampling}, 
      author={Savar D. Sinha and Yu Tong},
      year={2025},
      eprint={2509.07937},
      archivePrefix={arXiv},
      optprimaryClass={quant-ph},
      url={https://arxiv.org/abs/2509.07937},
      journal={arXiv}  
}

@article{Wiebe_2014,
   title={Quantum Hamiltonian learning using imperfect quantum resources},
   volume={89},
   DOI={10.1103/physreva.89.042314},
   pages=042314,
   journal={Phys. Rev. A},
   publisher={American Physical Society (APS)},
   author={Wiebe, Nathan and Granade, Christopher and Ferrie, Christopher and Cory, David},
   year={2014},
   month=Apr }

@article{Wiebe_2014_prl,
   title={Hamiltonian Learning and Certification Using Quantum Resources},
   volume={112},
   DOI={10.1103/physrevlett.112.190501},
   pages=190501,
   journal={Phys. Rev. Lett.},
   publisher={American Physical Society (APS)},
   author={Wiebe, Nathan and Granade, Christopher and Ferrie, Christopher and Cory, D. G.},
   year={2014},
   month=May }

@inproceedings{Zhao_2025, series={STOC ’25},
   title={Learning the Structure of Any Hamiltonian from Minimal Assumptions},
   url={http://dx.doi.org/10.1145/3717823.3718115},
   DOI={10.1145/3717823.3718115},
   booktitle={Proceedings of the 57th Annual ACM Symposium on Theory of Computing},
   publisher={ACM},
   author={Zhao, Andrew},
   year={2025},
   month={June}, 
   pages={1201–1211},
   collection={STOC ’25} }

@article{zubida2021optimalshorttimemeasurementshamiltonian,
      title={{Optimal short-time measurements for Hamiltonian learning}}, 
      author={Assaf Zubida and Elad Yitzhaki and Netanel H. Lindner and Eyal Bairey},
      year={2021},
      eprint={2108.08824},
      archivePrefix={arXiv},
      optprimaryClass={quant-ph},
      url={https://arxiv.org/abs/2108.08824},
      journal={arXiv}  
}

@article{shin2026heisenberglimitedhamiltonianlearningshorttime,
      title={Heisenberg-limited Hamiltonian learning without short-time control}, 
      author={Myeongjin Shin and Junseo Lee and Changhun Oh},
      year={2026},
      eprint={2604.27838},
      archivePrefix={arXiv},
      optprimaryClass={quant-ph},
      url={https://arxiv.org/abs/2604.27838},
      journal={arXiv}  
}

@article{Dutkiewicz_2024,
   title={The advantage of quantum control in many-body Hamiltonian learning},
   volume={8},
   ISSN={2521-327X},
   url={http://dx.doi.org/10.22331/q-2024-11-26-1537},
   DOI={10.22331/q-2024-11-26-1537},
   journal={Quantum},
   publisher={Verein zur Forderung des Open Access Publizierens in den Quantenwissenschaften},
   author={Dutkiewicz, Alicja and O'Brien, Thomas E. and Schuster, Thomas},
   year={2024},
   month=Nov, pages={1537} }

@article{abbas2025nearlyoptimalalgorithmslearn,
      title={Nearly optimal algorithms to learn sparse quantum Hamiltonians in physically motivated distances}, 
      author={Amira Abbas and Nunzia Cerrato and Francisco Escudero Gutiérrez and Dmitry Grinko and Francesco Anna Mele and Pulkit Sinha},
      year={2025},
      eprint={2509.09813},
      archivePrefix={arXiv},
      optprimaryClass={quant-ph},
      url={https://arxiv.org/abs/2509.09813},
      journal={arXiv}  
}

@article{chen2025quantumprobetomography,
      title={Quantum Probe Tomography}, 
      author={Sitan Chen and Jordan Cotler and Hsin-Yuan Huang},
      year={2025},
      eprint={2510.08499},
      archivePrefix={arXiv},
      optprimaryClass={quant-ph},
      url={https://arxiv.org/abs/2510.08499},
      journal={arXiv}  
}

@article{huang2023learningpredictarbitraryquantum,
      title={Learning to predict arbitrary quantum processes}, 
      author={Hsin-Yuan Huang and Sitan Chen and John Preskill},
      year={2023},
      eprint={2210.14894},
      archivePrefix={arXiv},
      optprimaryClass={quant-ph},
      url={https://arxiv.org/abs/2210.14894},
      journal={arXiv}  
}

@article{molteni2026identification,
      title={Quantum machine learning advantages beyond hardness of evaluation}, 
      author={Riccardo Molteni and Simon C. Marshall and Vedran Dunjko},
      year={2026},
      eprint={2504.15964},
      archivePrefix={arXiv},
      optprimaryClass={quant-ph},
      url={https://arxiv.org/abs/2504.15964},
      journal={arXiv}  
}

@article{daley2022practical,
  title={Practical quantum advantage in quantum simulation},
  author={Daley, Andrew J and Bloch, Immanuel and Kokail, Christian and Flannigan, Stuart and Pearson, Natalie and Troyer, Matthias and Zoller, Peter},
  journal={Nature},
  volume={607},
  number={7920},
  pages={667--676},
  year={2022},
  DOI={10.1038/s41586-022-04940-6},
  publisher={Nature Publishing Group UK London}
}

@article{Childs_2018,
   title={Toward the first quantum simulation with quantum speedup},
   volume={115},
   ISSN={1091-6490},
   url={http://dx.doi.org/10.1073/pnas.1801723115},
   DOI={10.1073/pnas.1801723115},
   number={38},
   journal={Proc. Natl. Ac. Sc.},
   publisher={National Academy of Sciences},
   author={Childs, Andrew M. and Maslov, Dmitri and Nam, Yunseong and Ross, Neil J. and Su, Yuan},
   year={2018},
   month={sept}, pages={9456–9461} }

@article{trivedi2024quantum,
  title={Quantum advantage and stability to errors in analogue quantum simulators},
  author={Trivedi, Rahul and Franco Rubio, Adrian and Cirac, J Ignacio},
  journal={Nature Comm.},
  volume={15},
  number={1},
  pages={6507},
  year={2024},
  DOI={10.1038/s41467-024-50750-x},
  publisher={Nature Publishing Group UK London}
}

@article{kashyap2025accuracy,
  title={Accuracy guarantees and quantum advantage in analog open quantum simulation with and without noise},
  author={Kashyap, Vikram and Styliaris, Georgios and Mouradian, Sara and Cirac, J Ignacio and Trivedi, Rahul},
  DOI={10.1103/PhysRevX.15.021017},
  journal={Phys. Rev. X},
  volume={15},
  number={2},
  pages={021017},
  year={2025},
  publisher={APS}
}

@article{arunachalam2021quantum,
  title={Quantum hardness of learning shallow classical circuits},
  author={Arunachalam, Srinivasan and Grilo, Alex Bredariol and Sundaram, Aarthi},
  journal={SIAM J. Comp.},
  volume={50},
  DOI={10.1137/20M1344202},
  pages={972--1013},
  year={2021},
  publisher={SIAM}
}

@inproceedings{lyubashevsky2010ideal,
  title={On ideal lattices and learning with errors over rings},
  author={Lyubashevsky, Vadim and Peikert, Chris and Regev, Oded},
  booktitle={Annual international conference on the theory and applications of cryptographic techniques},
  pages={1--23},
  year={2010},
  organization={Springer},
  doi={10.1145/2535925}
}

@inproceedings{bshouty1995learning,
  title={Learning DNF over the uniform distribution using a quantum example oracle},
  author={Bshouty, Nader H and Jackson, Jeffrey C},
  booktitle={Proceedings of the eighth annual conference on Computational learning theory},
  pages={118--127},
  year={1995},
  doi={10.1145/225298.225312}
}

@article{depradenne2026learninghamiltonianslongtimes,
      title={Learning Hamiltonians at Long Times}, 
      author={Constantin Cedillo Vayson de Pradenne and Jordan Cotler and Hsin-Yuan Huang},
      year={2026},
      eprint={2606.05690},
      archivePrefix={arXiv},
      optprimaryClass={quant-ph},
      url={https://arxiv.org/abs/2606.05690},
      journal={arXiv}  
}

@article{Rouz__2026,
   title={Efficient thermalization and universal quantum computing with quantum Gibbs samplers},
   ISSN={1745-2481},
   url={http://dx.doi.org/10.1038/s41567-026-03246-y},
   DOI={10.1038/s41567-026-03246-y},
   journal={Nature Phys.},
   publisher={Springer Science and Business Media LLC},
   author={Rouz\'e, Cambyse and Stilck Fran\c{c}a, Daniel and Alhambra, \'Alvaro M.},
   year={2026},
   month=Apr }

@article{Bravyi_2022,
   title={Quantum Hamiltonian complexity in thermal equilibrium},
   volume={18},
   ISSN={1745-2481},
   url={http://dx.doi.org/10.1038/s41567-022-01742-5},
   DOI={10.1038/s41567-022-01742-5},
   number={11},
   journal={Nature Phys.},
   publisher={Springer Science and Business Media LLC},
   author={Bravyi, Sergey and Chowdhury, Anirban and Gosset, David and Wocjan, Pawel},
   year={2022},
   month=Oct, pages={1367–1370} }

@article{bravyi2024quantumcomplexitykroneckercoefficients,
      title={Quantum complexity of the Kronecker coefficients}, 
      author={Sergey Bravyi and Anirban Chowdhury and David Gosset and Vojtech Havlicek and Guanyu Zhu},
      year={2024},
      eprint={2302.11454},
      archivePrefix={arXiv},
      primaryClass={quant-ph},
      url={https://arxiv.org/abs/2302.11454},
      journal={arXiv}  
}

@article{liu2025robustlindbladianestimationquantum,
      title={Robust Lindbladian Estimation for Quantum Dynamics}, 
      author={Yinchen Liu and James R. Seddon and Tamara Kohler and Emilio Onorati and Toby S. Cubitt},
      year={2025},
      eprint={2507.07912},
      archivePrefix={arXiv},
      primaryClass={quant-ph},
      url={https://arxiv.org/abs/2507.07912},
      journal={arXiv}  
}

@article{Onorati_2023,
   title={Fitting quantum noise models to tomography data},
   volume={7},
   ISSN={2521-327X},
   url={http://dx.doi.org/10.22331/q-2023-12-05-1197},
   DOI={10.22331/q-2023-12-05-1197},
   journal={Quantum},
   publisher={Verein zur Forderung des Open Access Publizierens in den Quantenwissenschaften},
   author={Onorati, Emilio and Kohler, Tamara and Cubitt, Toby S.},
   year={2023},
   month=Dec, pages={1197} }

@article{Bairey_2020,
   title={Learning the dynamics of open quantum systems from their steady states},
   volume={22},
   DOI={10.1088/1367-2630/ab73cd},
   number={3},
   journal={New J. Phys.},
   publisher={IOP Publishing},
   author={Bairey, Eyal and Guo, Chu and Poletti, Dario and Lindner, Netanel H and Arad, Itai},
   year={2020},
   month=Mar, pages={032001} }

@article{franca2022efficientrobustestimationmanyqubit,
      title={Efficient and robust estimation of many-qubit Hamiltonians}, 
      author={Daniel Stilck Fran\c{c}a and Liubov A. Markovich and V. V. Dobrovitski and Albert H. Werner and Johannes Borregaard},
      year={2022},
      eprint={2205.09567},
      archivePrefix={arXiv},
      optprimaryClass={quant-ph},
      url={https://arxiv.org/abs/2205.09567},
      journal={arXiv}  
}

@article{ivashkov2026ansatzfreelearninglindbladiandynamics,
      title={Ansatz-Free Learning of Lindbladian Dynamics In Situ}, 
      author={Petr Ivashkov and Nikita Romanov and Weiyuan Gong and Andi Gu and Hong-Ye Hu and Susanne F. Yelin},
      year={2026},
      eprint={2603.05492},
      archivePrefix={arXiv},
      optprimaryClass={quant-ph},
      url={https://arxiv.org/abs/2603.05492},
      journal={arXiv} 
}

@article{cai2026optimaldetectiondissipationlindbladian,
      title={Optimal detection of dissipation in Lindbladian dynamics}, 
      author={Yiyi Cai},
      year={2026},
      eprint={2603.17736},
      archivePrefix={arXiv},
      optprimaryClass={quant-ph},
      url={https://arxiv.org/abs/2603.17736},
      journal={arXiv}  
}

\newpage

\appendix

\section{Classical shadows protocol}
\label{app:classical-shadows-protocol}

In this appendix, we provide details on the formalism of classical shadows, since we use it in Algorithm \ref{alg:inference-via-hamiltonian-simulation}. In essence, one can think of the classical shadows protocol as consisting of a data acquisition phase, followed by a median-of-means procedure that estimates certain properties of interest. Below, we provide the subroutine used to construct the classical shadow of an unknown quantum state, given numerous copies thereof.

\begin{figure}[!ht]
\centering
\begin{minipage}{0.98\textwidth}
\footnotesize
\hrule
\vspace{0.75ex}
\subroutinecaption{Data acquisition phase of classical shadows protocol~\cite{Shadows}}
{alg:subroutine-data-acquisition-classical-shadows}
    \begin{algorithmic}[1]
        \Input 
        \Statex \hspace{\algorithmicindent} 1. $\tilde{N} \in \mathbb{N}$ copies of an $n$-qubit quantum state $\rho$.
        \Statex \hspace{\algorithmicindent} 2. Sample-access to the uniform distribution over the ensemble of unitaries $\mathcal{U} \coloneqq \left\{ \mathbbm{1}, H, HS^{\dagger} \right\}^{\otimes n}$.
        \Output A classical shadow given by a collection of unbiased snapshots, which we denote as $\bigl( \hat{\sigma}^{(\ell)} \bigr)_{\ell = 1}^{\tilde{N}}$.
        \For{$\ell = 1$ to $\tilde{N}$}
        \State Sample $U^{(\ell)} \coloneqq \bigotimes_{j = 1}^{n} U_{j}^{(\ell)} \overset{\mathrm{i.i.d.}}{\sim} \mathsf{Unif} \left( \mathcal{U} \right)$
        \State Prepare $U^{(\ell)} \rho ~U^{(\ell), \dagger}$
        \State Measure in the computational basis to obtain the pair $\bigl( U^{(\ell)}, b^{(\ell)} \bigr) \in \mathcal{U} \times \left\{ 0, 1 \right\}^{n}$    \Comment{Randomized measurement.}
        \State Construct the unbiased classical snapshot: \Comment{Inversion of randomized measurement channel.}
        \begin{align}
            \hat{\sigma}^{(\ell)} \leftarrow \bigotimes_{j = 1}^{n} \left( 3U_{j}^{(\ell), \dagger} \bigl\lvert b_{j}^{(\ell)} \bigr\rangle \bigl\langle b_{j}^{(\ell)} \bigr\rvert U_{j}^{(\ell)} - \mathbbm{1} \right)
            \end{align}
        \EndFor \Comment{Obtain a collection of classical snapshots: $\bigl( \hat{\sigma}^{(\ell)} \bigr)_{\ell = 1}^{\tilde{N}}$.}
        \State \Return $\bigl( \hat{\sigma}^{(\ell)} \bigr)_{\ell = 1}^{\tilde{N}}$ 
    \end{algorithmic}
\vspace{0.6ex}
\hrule

\end{minipage}
\end{figure}
Having outlined the method used to construct the classical shadow, let us take a look at how it can be used to estimate expectation values on some observables of interest.
\floatname{algorithm}{Subroutine}

\begin{figure}[!ht]
\centering
\begin{minipage}{0.98\textwidth}
\footnotesize
\hrule
\vspace{0.75ex}
\subroutinecaption{Median-of-means estimation of expectation values~\cite{Shadows}}
{alg:subroutine-median-of-means-classical-shadows}
    \begin{algorithmic}[1]
        \Input 
        \Statex \hspace{\algorithmicindent} 1. A classical shadow $\bigl( \hat{\sigma}^{(\ell)} \bigr)_{\ell = 1}^{\tilde{N}}$ of some $n$-qubit quantum state $\rho$.
        \vspace{1mm}
        \Statex \hspace{\algorithmicindent} 2. A set of observables $\mathcal{M} \coloneqq \left\{ O_{1}, O_{2}, \ldots, O_{P} \right\}$.
        \Statex \hspace{\algorithmicindent} 3. Number of segments $G \in \mathbb{N}$.
        \Output Estimates $\bigl( \hat{\langle O_{j} \rangle}_{\rho}~\bigr)_{j = 1}^{P}$ of the expectation values of the observables in $\mathcal{M}$.
        \State Set $m \leftarrow \lfloor \tilde{N} / G \rfloor$\Comment{Dividing classical snapshot collection into $G$ segments of size $m$.}
        \For{$j = 1$ to $P$} \Comment{Loop over all observables.}
            \For{$g = 1$ to $G$}    \Comment{Loop over elements of each equally sized segment of classical snapshot collection.}
                \State $S_{g}^{(j)} \leftarrow 0$   \Comment{Initialize sum variable to zero.}
                \For{$r = 1$ to $m$}
                    \State $\ell \leftarrow \left( g - 1 \right) m + r$
                    \State $S_{g}^{(j)} \leftarrow S_{g}^{(j)} + \operatorname{Tr} \left[ O_{j} \hat{\sigma}^{(\ell)} \right]$
                \EndFor
                \vspace{1mm}
                \State $\mu_{g}^{(j)} \leftarrow S_{g}^{(j)}/m$ \Comment{Sample mean for each segment.}
                \vspace{1mm}
            \EndFor
            \State $\hat{\langle O_{j} \rangle}_{\rho} \leftarrow \operatorname{median} \left( \mu_{1}^{(j)}, \mu_{2}^{(j)}, \ldots, \mu_{g}^{(j)} \right)$    \Comment{Median-of-means estimation.}
        \EndFor
        \State \Return $\left( \hat{\langle O_{j} \rangle}_{\rho} \right)_{j \in \left[ P \right]}$ \Comment{Predicted expectation values.}
    \end{algorithmic}
\vspace{0.6ex}
\hrule

\end{minipage}
\end{figure}
For theoretical details of the aforementioned subroutines, we refer the reader to the original work of Huang et al.~\cite{Shadows}.

\section{Quantum advantage}
\label{app:classical-hardness}

\subsection{A low-intersection BQP-complete Hamiltonian with constant-precision parameters}
\label{subsec:pauli_decomposition}

In this section, we construct a Hamiltonian family that satisfies the low-intersection property required by our learning algorithm while allowing us to simulate arbitrary BQP computations over long-time evolution. This Hamiltonian family also has the interesting property that the parameters $\lambda$ specifying its Pauli expansion only take values from a discrete set, up to some constant precision. This means that learning $\lambda$ to constant error (in the $\Vert \cdot \Vert_{\ell_{\infty}}$-norm) is sufficient to recover the Hamiltonian exactly, and therefore simulate it for arbitrarily long times without any error. We prove that our quantum algorithm solves the learning problem perfectly in the classically hard case, where, in the definition of the concept class in Def.~\ref{def:learning_problem_concept_class}, $U(t,\lambda)$ is given by the time evolution of $H_{\BQP}$ as discussed in Section~\ref{sec:quantum-advantage}. 

\begin{lemma}\label{lemma: quantum learnability const error app}
	 Let $\mathcal{L}$ be a promise problem in $(\mathsf{Promise})\BQP$. There exists a local, low-intersection Hamiltonian
\begin{equation}\label{eq: H_bqp appendix}
H_{\mathsf{BQP}} (\lambda)=\sum_{a \in \left[ M \right]} \lambda_a E_a,
\end{equation}

where the operators $E_a$ are local Pauli strings, such that:
\begin{itemize}
\item For evolution times linear in system size, the real-time evolution generated by $H_{\mathsf{BQP}}$ simulates a quantum circuit $U_{\mathsf{BQP}}$ that decides $\mathcal{L}$ with bounded error.

\item There exists a finite set $\Lambda\subset\mathbb{R}$, independent of the system size, such that
$
\lambda_a\in\Lambda
$
for every $a\in[M]$.

\end{itemize}
This Hamiltonian belongs to a Hamiltonian family $\{H_{\BQP}(\lambda)\}_\lambda$ that is capable of simulating any circuit $U_{\BQP}$ of the same depth, for a suitable choice of the parameters $\lambda$. 
 
\end{lemma}
\begin{proof}
Let $\calL$ be a $(\mathsf{Promise})\BQP$ problem and $U_{\BQP}$ the corresponding circuit that correctly decides it for input bit-strings in $\{0,1\}^m$. By the Feynman-Kitaev clock Hamiltonian construction~\cite{feynman1986quantum,nagaj2010fast}, there exists a local Hamiltonian $H_{\BQP}$ such that evolving under it for a linear time $T_{\BQP}$ implements $U_{\BQP}$ with high probability. We will now show that the coefficients of the Pauli strings in $H_{\BQP}$ can be enforced to only take values in a fixed discrete set so that constant-precision estimates are sufficient to determine $\lambda_{\BQP}$ exactly and thus reconstruct $H_{\BQP}$.

In particular, given the circuit $U_{\BQP}=U_K U_{K - 1} \dots U_2U_1$ pre-compiled in a brickwork architecture using the universal gateset $\{\mathrm{CNOT}, H, T\}$, the associated Feynman Hamiltonian is defined as
\begin{equation}\label{eq:Feynman-Hamiltonian}
    H_F=\sum_{j=1}^{K} H_j
\end{equation}
where the terms
\begin{equation}
    H_j=U_j\otimes \ket{j}\bra{j-1}+ U_j^\dagger\otimes\ket{j-1}\bra{j}
\end{equation}
act between a work register and a clock register. Later, in Subsection~\ref{subsec:time-evolution-feynman-kitaev-hamiltonian}, we show that, on average, the target concept implements $U_{\BQP}$ with high probability for a randomly sampled evolution time $t\in[0,T_{\BQP}]$,  where $T_{\BQP}$ is linear in the depth $K$ of the circuit $U_{\BQP}=U_K U_{K - 1} \cdots U_2U_1$. Specifically, in Lemma \ref{lem:avg-prob-mass-completed-subspace-clock-walk}, we show that this success probability of correctly implementing $U_{\BQP}$ can be brought to $1-\delta$, that is, arbitrarily close to $1$ by padding the circuit with $\mathcal{O}(K/\delta)$ identity gates and taking $T_{\BQP}=\widetilde{\mathcal{O}}(K/\delta^2)$.

Since the Feynman Hamiltonian is constructed with two-local gates $U_j\in\{\mathrm{CNOT}, H, T\}$, the standard unary encoding of the clock register makes $H_F$ $5$-local. However, in order to have quantum learnability, we must ensure that $H_{\BQP}$ satisfies the low-intersection property as well. For this, we leverage the construction in Ref.~\cite{oliveira2005complexity}. This involves transforming the circuit $U_{\BQP}$ into an equivalent circuit $U'_{\BQP}=U'_{K'}\cdots U'_2U'_1$ acting on $K'=mK$ qubits, such that each qubit participates in at most three gates (one $U_j$ and two $\mathrm{SWAP}$s). We refer to these additional qubits as the auxiliary work register. The Feynman Hamiltonian corresponding to this modified circuit is our final  $H_{\BQP}$. It takes the same form as in Eq.~(\ref{eq:Feynman-Hamiltonian}), but uses unitaries $U'_j\in\{\mathrm{CNOT},H,T,\mathrm{SWAP}\}$. Given that we choose a unary clock $\lvert j \rangle = \lvert 1^{j} 0^{K' - j} \rangle$, the $H_j$ terms can be rewritten as
\begin{equation}
    H_j=U'_j\otimes I\otimes\ket{1,1,0}\bra{1,0,0}_{j-1,j,j+1} \otimes I + (U_j')^\dagger \otimes I \otimes\ket{1,0,0}\bra{1,1,0}_{j-1,j,j+1}\otimes I
\end{equation}
These terms can be rewritten in the form of Eq.~(\ref{eq: H_bqp appendix}) by expanding each $U'_j$ and clock transition operators in the Pauli basis. In particular, using the relations
\[
X=\begin{pmatrix}0&1\\[2pt]1&0\end{pmatrix},\qquad
Y=\begin{pmatrix}0&-i\\[2pt]i&0\end{pmatrix},\qquad
Z=\begin{pmatrix}1&0\\[2pt]0&-1\end{pmatrix},
\]
and
\[
|1\rangle\langle 1|=\frac{I-Z}{2}, \qquad
|0\rangle\langle 0|=\frac{I+Z}{2},
\]
\[
|1\rangle\langle 0|=\frac{X-iY}{2}, \qquad
|0\rangle\langle 1|=\frac{X+iY}{2},
\]
we obtain
\[
|1,1,0\rangle\langle 1,0,0|
=
\left(\frac{I-Z}{2}\right)\otimes \left(\frac{X-iY}{2}\right)\otimes \left(\frac{I+Z}{2}\right),
\]

\[
|1,0,0\rangle\langle 1,1,0|
=
\left(\frac{I-Z}{2}\right)\otimes \left(\frac{X+iY}{2}\right)\otimes \left(\frac{I+Z}{2}\right),
\]

\[
\mathrm{SWAP}
=\frac{1}{2}\Big(I\otimes I+X\otimes X+Y\otimes Y+Z\otimes Z\Big),
\]

\[
\mathrm{CNOT}
=|0\rangle\langle 0|\otimes I \;+\; |1\rangle\langle 1|\otimes X
=\frac{I+Z}{2}\otimes I \;+\; \frac{I-Z}{2}\otimes X,
\]

\[
H=\frac{1}{\sqrt{2}}(X+Z),
\]

\[
T=\begin{pmatrix}1&0\\[2pt]0&e^{i\pi/4}\end{pmatrix}
=\frac{1+e^{i\pi/4}}{2}\,I+\frac{1-e^{i\pi/4}}{2}\,Z.
\]
Therefore, each gate $U'_j$, as well as each clock transition operator, admits a Pauli expansion with coefficients drawn from a fixed finite set $\Lambda_0$. Since any given Pauli string can appear in at most a constant number of terms $H'_j$, each coefficient in the final expansion $H_{\BQP}=\sum_{a \in \left[ M \right]} \lambda_a E_a$
is a sum of at most a constant number of elements of $\Lambda_0$. Consequently, the coefficients $\lambda_a$ belong to a finite set $\Lambda$ that is independent of the system size.
Our final Hamiltonian family $\{H_{\BQP}(\lambda)\}_\lambda$ is obtained by considering all possible Pauli decompositions obtainable by following the construction above for arbitrary choices of circuits $U_{\BQP}=U_K U_{K - 1} \dots U_2U_1$.
\end{proof}

\subsection{Time-evolution under the Feynman-Kitaev clock Hamiltonian}
\label{subsec:time-evolution-feynman-kitaev-hamiltonian}
Here, we provide a rigorous analysis of the dynamics generated by the Feynman-Kitaev Hamiltonian. In particular, we show that, on average over some polynomially large interval $\left[ 0, T \right]$ and upon padding with sufficiently many identity gates, long-time-evolution with respect to the clock Hamiltonian is able to correctly simulate $\mathsf{BQP}$-complete computations with success probability arbitrarily close to $1$.

\subsubsection{Set-up}
Our basic set-up will be as follows. Let $L, K \in \mathbb{N}$, with $L > K$. Consider $U_{\mathsf{BQP}} = U_{L} U_{L - 1} \ldots U_{K}U_{K - 1} \ldots U_{2} U_{1}$ with $\left\{ U_{L}, U_{L - 1} \ldots U_{K + 1} \right\}$ all identities and $L - K = R$. Let $U_{\mathsf{BQP}}$ decide membership of $z \in \left\{ 0, 1 \right\}^{m}$ in some language $\mathcal{L}$ belonging to the class $\mathsf{BQP}$, and let
\begin{align}
    \label{eq:feynman-kitaev-clock-hamiltonian}
    H_{F} = \sum_{j = 1}^{L} \left( U_{j} \otimes \lvert j \rangle \langle j - 1 \rvert + U_{j}^{\dagger} \otimes \lvert j - 1 \rangle \langle j \rvert \right)
\end{align}
be the corresponding Feynman-Kitaev clock Hamiltonian. For each $j \in \left\{ 1, \ldots, L \right\}$, define:
\begin{align}
    \lvert \psi_{j} \left( z \right) \rangle &\coloneqq U_{j} U_{j - 1} \ldots U_{2} U_{1} \left( \lvert z \rangle \otimes \lvert 0 \rangle^{\otimes a} \right) \in \mathcal{H}_{\text{work}}, \\
    \text{and}\qquad\lvert \eta_{j} \left( z \right) \rangle &\coloneqq \lvert \psi_{j} \left( z \right) \rangle \otimes \lvert j \rangle \in \mathcal{H}_{\text{work}} \otimes \mathcal{H}_{\text{clock}},
\end{align}
with $\lvert \psi_{0} \left( z \right) \rangle \coloneqq \lvert z \rangle \otimes \lvert 0 \rangle^{\otimes a}$ for $j = 0$ and $\lvert \eta_{0} \left( z \right) \rangle \coloneqq \left( \lvert z \rangle \otimes \lvert 0 \rangle^{\otimes a} \right) \otimes \lvert 0 \rangle$, to be the work and history states respectively. We will also consider the following ``path" Hamiltonian, which acts on the clock subspace $\mathcal{H}_{\text{clock}}$:
\begin{align}
    \label{eq:path-hamiltonian-clock-subspace}
    H_{\text{path}} \coloneqq \sum_{j = 1}^{L} \left( \lvert j \rangle \langle j - 1 \rvert + \lvert j - 1 \rangle \langle j \rvert \right).
\end{align}

\subsubsection{Time-evolution of history state}
\begin{proposition}[Unitary relation between $H_{F}$ and $H_{\text{path}}$]
    Let $H_{F}$ be the Feynman-Kitaev clock Hamiltonian as defined in Eq.~(\ref{eq:feynman-kitaev-clock-hamiltonian}) and $H_{\textnormal{path}}$ be the path Hamiltonian defined in Eq.~(\ref{eq:path-hamiltonian-clock-subspace}). Then, there exists a block-diagonal unitary
    \begin{align}
        W \coloneqq \sum_{j = 0}^{L} V_{j} \otimes \lvert j \rangle \langle j \rvert,
    \end{align}
    such that $H_{F} = W \left( I_{\textnormal{work}} \otimes H_{\textnormal{path}}  \right) W^{\dagger}$.
\end{proposition}

\begin{proof}
Consider the following:
\begin{align}
    &W \left( I_{\text{work}} \otimes H_{\text{path}} \right) W^{\dagger} \\
    = &W \left( I_{\text{work}} \otimes \sum_{j = 1}^{L} \lvert j \rangle \langle j - 1 \rvert + \lvert j - 1 \rangle \langle j \rvert \right) W^{\dagger} \\
    = &\sum_{j = 1}^{L} W \left( I_{\text{work}} \otimes \left( \lvert j \rangle \langle j - 1 \rvert + \lvert j - 1 \rangle \langle j \rvert \right) \right) W^{\dagger} \\
    = &\sum_{j = 1}^{L} \left[ \left( \sum_{k = 0}^{L} V_{k} \otimes \lvert k \rangle \langle k \rvert \right) \left( I_{\text{work}} \otimes \left( \lvert j \rangle \langle j - 1 \rvert + \lvert j - 1 \rangle \langle j \rvert \right) \right) \left( \sum_{\ell = 0}^{L} V_{\ell} \otimes \lvert \ell \rangle \langle \ell \rvert \right)^{\dagger} \right] \\
    = &\sum_{j = 1}^{L} \left[ \sum_{k = 0}^{L} \sum_{\ell = 0}^{L} V_{k} V_{\ell}^{\dagger} \otimes \left( \lvert k \rangle \langle k \rvert j \rangle \langle j - 1 \rvert \ell \rangle \langle \ell \rvert + \lvert k \rangle \langle k \rvert j - 1 \rangle \langle j \rvert \ell \rangle \langle \ell \rvert \right) \right] \\
    = &\sum_{j = 1}^{L} \left[ V_{j} V_{j - 1}^{\dagger} \otimes \lvert j \rangle \langle j - 1 \rvert + V_{j - 1} V_{j}^{\dagger} \otimes \lvert j - 1 \rangle \langle j \rvert  \right] \\
    = &\sum_{j = 1}^{L} \left[ U_{j} U_{j - 1} \ldots U_{1} U_{1}^{\dagger} \ldots U_{j-1}^{\dagger} \otimes \lvert j \rangle \langle j - 1 \rvert + U_{j - 1} \ldots U_{1} U_{1}^{\dagger} \ldots U_{j-1}^{\dagger} U_{j}^{\dagger} \otimes \lvert j - 1 \rangle \langle j \rvert \right] \\
    = &\sum_{j = 1}^{L} \left( U_{j} \otimes \lvert j \rangle \langle j - 1 \rvert + U_{j}^{\dagger} \otimes \lvert j - 1 \rangle \langle j \rvert \right) = H_{F}.
\end{align}
\end{proof}

\begin{lemma}[Time-evolved history state]
    \label{lem:time-evolved-history-state}
    Let $\lvert \eta_{0} \left( z \right) \rangle \coloneqq \lvert \psi_{0} \left( z \right) \rangle \otimes \lvert 0 \rangle$ be the initial history state, and let $\lvert \Psi_{z} \left( t \right) \rangle \coloneqq e^{-iH_{F}t} \lvert \eta_{0} \left( z \right) \rangle$ be the history state after time-evolution under the Feynman-Kitaev clock Hamiltonian $H_{F}$ for some time $t > 0$. The state $\lvert \Psi_{z} \left( t \right) \rangle$ is then given as follows:
    \begin{align}
        \lvert \Psi_{z} \left( t \right) \rangle = \sum_{j = 0}^{L} \langle j \rvert e^{-iH_{\textnormal{path}}t} \lvert 0 \rangle \lvert \eta_{j} \left( z \right) \rangle.
    \end{align}
\end{lemma}

\begin{proof}
We compute $\lvert \Psi_{z} \left( t \right) \rangle$ as follows:
\begin{align}
    \lvert \Psi_{z} \left( t \right) \rangle &\coloneqq e^{-iH_{F}t} \lvert \eta_{0} \left( z \right) \rangle\\ 
    &= e^{-iH_{F}t} \left( \lvert \psi_{0} \left( z \right) \rangle \otimes \lvert 0 \rangle \right) \\
    &= e^{-i W \left( I_{\text{work}} \otimes H_{\text{path}} \right) W^{\dagger} t} \left( \lvert \psi_{0} \left( z \right) \otimes \lvert 0 \rangle \right) \\
    &= W \left( I_{\text{work}} \otimes e^{-iH_{\text{path}}t} \right) W^{\dagger} \left( \lvert \psi_{0} \left( z \right) \rangle \otimes \lvert 0 \rangle \right) \\
    &= W \left( I_{\text{work}} \otimes e^{-iH_{\text{path}}t} \right) \left( \sum_{j = 0}^{L} V_{j}^{\dagger} \otimes \lvert j \rangle \langle j \rvert \right) \left( \lvert \psi_{0} \left( z \right) \rangle \otimes \lvert 0 \rangle \right) \\
    &= W \left( I_{\text{work}} \otimes e^{-iH_{\text{path}}t} \right) \left( V_{0}^{\dagger} \lvert \psi_{0} \left( z \right) \rangle \otimes \lvert 0 \rangle \langle 0 \rvert 0 \rangle \right) \\
    &= W \left( I_{\text{work}} \otimes e^{-iH_{\text{path}}t} \right) \left( \lvert \psi_{0} \left( z \right) \rangle \otimes \lvert 0 \rangle \right) \\
    &= W \left( \lvert \psi_{0} \left( z \right) \rangle \otimes e^{-iH_{\text{path}}t} \lvert 0 \rangle \right) \\
    &= \left( \sum_{j = 0}^{L} V_{j} \otimes \lvert j \rangle \langle j \rvert \right) \left( \lvert \psi_{0} \left( z \right) \rangle \otimes e^{-iH_{\text{path}}t} \lvert 0 \rangle \right) \\
    &= \sum_{j = 0}^{L} V_{j} \lvert \psi_{0} \left( z \right) \rangle \otimes \langle j \rvert e^{-iH_{\text{path}}t} \lvert 0 \rangle \lvert j \rangle \\
    &= \sum_{j = 0}^{L} \langle j \rvert e^{-iH_{\text{path}}t} \lvert 0 \rangle \lvert \psi_{j} \left( z \right) \rangle \lvert j \rangle = \sum_{j = 0}^{L} \langle j \rvert e^{-iH_{\text{path}}t} \lvert 0 \rangle \lvert \eta_{j} \left( z \right) \rangle.
\end{align}
This concludes the proof.
\end{proof}

\begin{lemma}[Average probability mass in the completed subspace]
    \label{lem:avg-prob-mass-completed-subspace-clock-walk}
    Define $\alpha_{j} \left( t \right) \coloneqq \langle j \rvert e^{-iH_{\textnormal{path}}t} \lvert 0 \rangle$ to be the $j^{\text{th}}$ probability amplitude of the time-evolved history state $\lvert \Psi_{z} \left( t \right) \rangle$. Let
    \begin{align}
        \overline{p}_{T} \coloneqq \frac{1}{T} \int_{0}^{T} \sum_{j = K}^{L} \left\vert \alpha_{j} \left( t \right) \right\vert^{2} dt = \frac{1}{T} \int_{0}^{T} \langle \Psi_{z} \left( t \right) \rvert I_{\textnormal{work}} \otimes \Pi_{\textnormal{done}} \rvert \Psi_{z} \left( t \right) \rangle dt,~~\text{where}~~\Pi_{\textnormal{done}} \coloneqq \sum_{j = K}^{L} \lvert j \rangle \langle j \rvert,
    \end{align}
    denote the time-averaged probability mass in the completed subspace, namely, the average probability that the clock register lies in the completed portion of the clock, $\left\{ K, K + 1, \ldots, L \right\}$. Then, for some $\delta > 0$, in order for~$\overline{p}_{T} \geqslant 1 - \delta$ to hold, it is sufficient to take a number of padded gates $R$ and a time-averaging window $[0,T]$ such that
    \begin{align}
        R &\in \mathcal{O} \left( \frac{K}{\delta} \right) \\
        \text{and}\quad T &\in \mathcal{O} \left( \frac{K}{\delta^{2}}  \log \left( \frac{K}{\delta} \right) \right)
    \end{align}
    respectively.
\end{lemma}
\begin{proof}
From Equations (A2) and (A3) of Ref.~\cite{nagaj2010fast}, we know that the eigenvalues and eigenvectors of $H_{\text{path}}$ are given as follows:
\begin{align}
    \lvert v_{r} \rangle &\coloneqq \sqrt{\frac{2}{L + 2}} \sum_{\ell = 0}^{L} \sin \left( \frac{\pi r \left( \ell + 1 \right)}{L + 2} \right) \lvert \ell \rangle, \\
    E_{r} &= 2 \cos \left( \frac{\pi r}{L + 2} \right),
\end{align}
for $r = 1, \ldots, L + 1$. Let us now compute
\begin{align}
    \sum_{j = K}^{L} \left\vert \alpha_{j} \left( t \right) \right\vert^{2} &= \sum_{j = K}^{L} \left\vert \langle j \rvert e^{-iH_{\text{path}}t} \lvert 0 \rangle \right\vert^{2} \\
    &= \sum_{j = K}^{L} \langle 0 \rvert e^{iH_{\text{path}}t} \lvert j \rangle \langle j \rvert e^{-iH_{\text{path}}t} \lvert 0 \rangle \\
    &= \sum_{j = K}^{L} \langle 0 \rvert \left( \sum_{s = 1}^{L + 1} e^{iE_{s}t} \lvert v_{s} \rangle \langle v_{s} \rvert \right) \rvert j \rangle \langle j \rvert \left( \sum_{r = 1}^{L + 1} e^{-iE_{r}t} \lvert v_{r} \rangle \langle v_{r} \rvert \right) \lvert 0 \rangle \\
    &= \sum_{j = K}^{L} \sum_{r = 1}^{L + 1} \sum_{s = 1}^{L + 1} e^{-i \left( E_{r} - E_{s} \right) t} \langle 0 \rvert v_{s} \rangle \langle v_{s} \rvert j \rangle \langle j \rvert v_{r} \rangle \langle v_{r} \rvert 0 \rangle \\
    &= \frac{4}{\left( L + 2 \right)^{2}} \sum_{r, s = 1}^{L + 1} \left\{ \sin \left( \frac{\pi r}{L + 2} \right) \sin \left( \frac{\pi s}{L + 2} \right) \sum_{j = K}^{L} \sin \left( \frac{\pi r \left( j + 1 \right)}{L + 2} \right) \sin \left( \frac{\pi s \left( j + 1 \right)}{L + 2} \right) e^{-i \left( E_{r} - E_{s} \right) t} \right\}.
\end{align}
Let us define
\begin{align}
    A_{r, s} \coloneqq \frac{4}{\left( L + 2 \right)^{2}} \sin \left( \frac{\pi r}{L + 2} \right) \sin \left( \frac{\pi s}{L + 2} \right) \sum_{j = K}^{L} \sin \left( \frac{\pi r \left( j + 1 \right)}{L + 2} \right) \sin \left( \frac{\pi s \left( j + 1 \right)}{L + 2} \right)
\end{align}
for notational convenience. We may rewrite the time-averaged probability mass as the following:
\begin{align}
     \overline{p}_{T} &\coloneqq \frac{1}{T}\int_{0}^{T} \sum_{j = K}^{L} \left\vert \alpha_{j} \left( t \right) \right\vert^{2} dt \\ 
     &= \frac{1}{T} \int_{0}^{T} \sum_{r, s = 1}^{L + 1} A_{r, s} e^{-i \left( E_{r} - E_{s} \right)t} dt \\
     \label{eq:avg-prob-diag-off-diag-split}
     &= \underset{\eqqcolon \overline{p}_{\text{diag}}}{\underbrace{\sum_{r = 1}^{L + 1} A_{r, r}}} + \underset{\eqqcolon \overline{p}_{\text{off}}}{\underbrace{\sum_{\substack{r, s = 1, \\ r \neq s}}^{L + 1} \left( A_{r, s} \cdot \frac{1}{T} \int_{0}^{T} e^{- i \left( E_{r} - E_{s} \right)t} dt \right)}},
\end{align}
where we have split the time-average into diagonal and off-diagonal parts. Let us focus on the diagonal part first. We have the following:
\begin{align}
    &\overline{p}_{\text{diag}} = \sum_{r = 1}^{L + 1} A_{r, r} \\
    = &\frac{4}{\left( L + 2 \right)^{2}} \sum_{r = 1}^{L + 1} \sin^{2} \left( \frac{\pi r}{L + 2} \right) \sum_{j = K}^{L} \sin^{2} \left( \frac{\pi r \left( j + 1 \right)}{L + 2} \right) \\
    = &\frac{4}{\left( L + 2 \right)^{2}} \sum_{j = K}^{L} \sum_{r = 1}^{L + 1} \sin^{2} \left( \frac{\pi r}{L + 2} \right) \sin^{2} \left( \frac{\pi r \left( j + 1 \right)}{L + 2} \right) \\
    = &\frac{4}{\left( L + 2 \right)^{2}} \sum_{j = K}^{L} \sum_{r = 1}^{L + 1} \frac{1}{2} \left\{ 1 - \cos \left( \frac{2 \pi r}{L + 2} \right) \right\} \cdot \frac{1}{2} \left\{ 1 - \cos \left( \frac{2 \pi r \left( j + 1 \right)}{L + 2} \right) \right\} \\
    = &\frac{1}{\left( L + 2 \right)^{2}} \sum_{j = K}^{L} \sum_{r = 1}^{L + 1} \left\{ 1 - \cos \left( \frac{2 \pi r}{L + 2} \right) - \cos \left( \frac{2 \pi r \left( j + 1 \right)}{L + 2} \right) + \cos \left( \frac{2 \pi r}{L + 2} \right) \cos \left( \frac{2 \pi r \left( j + 1 \right)}{L + 2} \right) \right\} \\
    = &\frac{1}{\left( L + 2 \right)^{2}} \sum_{j = K}^{L} \sum_{r = 1}^{L + 1} \biggl\{ 1 - \cos \left( \frac{2 \pi r}{L + 2} \right) - \cos \left( \frac{2 \pi r \left( j + 1 \right)}{L + 2} \right) + \frac{1}{2} \biggl( \cos \left( \frac{2 \pi r \left( j + 1 \right)}{L + 2} - \frac{2 \pi r}{L + 2} \right) \\ 
    &\hspace{10cm} + \cos \left( \frac{2 \pi r \left( j + 1 \right)}{L + 2} + \frac{2 \pi r}{L + 2} \right) \biggr) \biggr\} \\
    \label{eq:trig-diagonal-avg-prob-1}
    = &\frac{1}{\left( L + 2 \right)^{2}} \sum_{j = K}^{L} \biggl\{ \sum_{r = 1}^{L + 1} 1 - \sum_{r = 1}^{L + 1} \cos \left( \frac{2 \pi r}{L + 2} \right) - \sum_{r = 1}^{L + 1} \cos \left( \frac{2 \pi r \left( j + 1 \right)}{L + 2} \right) + \frac{1}{2} \sum_{r = 1}^{L + 1} \cos \left( \frac{2 \pi r j}{L + 2} \right) \\ 
    \label{eq:trig-diagonal-avg-prob-2a}
    &\hspace{10cm} + \frac{1}{2} \sum_{r = 1}^{L + 1} \cos \left( \frac{2 \pi r \left( j + 2 \right)}{L + 2} \right) \biggr\}.
\end{align}
At this point, we will make use of the elementary roots-of-unity identity, which is given as the following, for some $q \in \mathbb{N}$:
\begin{align}
    \sum_{r = 0}^{L + 1} \exp \left( i \frac{2\pi r q}{L + 2} \right) = \begin{cases}
        L + 2, &q \equiv 0~\left( \mathrm{mod}~L + 2 \right), \\
        0, &q \not\equiv 0~\left( \mathrm{mod}~L + 2 \right).
    \end{cases}
\end{align}
Taking the real part of the above:
\begin{align}
    \sum_{r = 1}^{L + 1} \cos \left( \frac{2\pi r q}{L + 2} \right) &= \sum_{r = 0}^{L + 1} \cos \left( \frac{2 \pi r q}{L + 2} \right) - 1 \\
    &= \mathfrak{Re} \left[ \sum_{r = 0}^{L + 1} \exp \left( i \frac{2\pi r q}{L + 2} \right) \right] - 1 = \begin{cases}
        L + 1, &q \equiv 0~\left( \mathrm{mod}~L + 2 \right), \\
        -1, &q \not\equiv 0~\left( \mathrm{mod}~L + 2 \right).
    \end{cases}
\end{align}
As we can observe from Eqs.~(\ref{eq:trig-diagonal-avg-prob-1}) and~(\ref{eq:trig-diagonal-avg-prob-2a}), in our case, $q \in \left\{ 1, j, j + 1, j + 2 \right\}$. Also, in our case, $j \in \left\{ K, K+1, \ldots, L \right\}$. Note that, for $K \leqslant j < L$, we have that for all $q \in \left\{ 1, j, j + 1, j + 2 \right\}$, $q \not\equiv 0 \left( \mathrm{mod}~L + 2 \right)$. Thus, picking up again from Eq.~(\ref{eq:trig-diagonal-avg-prob-2a}):
\begin{align}
    \sum_{r = 1}^{L + 1} A_{r, r} = &\frac{1}{\left( L + 2 \right)^{2}} \sum_{j = K}^{L} \biggl\{ \sum_{r = 1}^{L + 1} 1 - \sum_{r = 1}^{L + 1} \cos \left( \frac{2 \pi r}{L + 2} \right) - \sum_{r = 1}^{L + 1} \cos \left( \frac{2 \pi r \left( j + 1 \right)}{L + 2} \right) + \frac{1}{2} \sum_{r = 1}^{L + 1} \cos \left( \frac{2 \pi r j}{L + 2} \right) \\ 
    \label{eq:trig-diagonal-avg-prob-2}
    &\hspace{10cm} + \frac{1}{2} \sum_{r = 1}^{L + 1} \cos \left( \frac{2 \pi r \left( j + 2 \right)}{L + 2} \right) \biggr\} \\
    &= \frac{1}{\left( L + 2 \right)^{2}} \left( \sum_{j = K}^{L - 1} \left\{ \left( L + 1 \right) - \left( -1 \right) - \left( -1 \right) + \frac{1}{2} \left( -1 \right) + \frac{1}{2} \left( -1 \right) \right\} \right) + \frac{3}{2 \left( L + 2 \right)} \\
    &= \frac{1}{\left( L + 2 \right)^{2}} \sum_{j = K}^{L - 1} \left( L + 2 \right) + \frac{3}{2 \left( L + 2 \right)} = \frac{R}{L + 2} + \frac{3}{2 \left( L + 2 \right)} = \frac{2R + 3}{2 \left( L + 2 \right)}.
\end{align}
In the above, the term $\frac{3}{2 \left( L + 2 \right)}$ is the result of computing the final summand. The derivation proceeds as follows:
\begin{align}
    &\frac{4}{\left( L + 2 \right)^{2}} \sum_{r = 1}^{L + 1} \sin^{2} \left( \frac{\pi r}{L + 2} \right) \sin^{2} \left( \frac{\pi r \left( L + 1 \right)}{L + 2} \right) \\
    = &\frac{4}{\left( L + 2 \right)^{2}} \sum_{r = 1}^{L + 1} \sin^{4} \left( \frac{\pi r}{L + 2} \right) \\
    = &\frac{4}{\left( L + 2 \right)^{2}} \sum_{r = 1}^{L + 1} \left( \frac{3}{8} - \frac{1}{2} \cos \left( \frac{2 \pi r}{L + 2} \right) + \frac{1}{8} \cos \left( \frac{4 \pi r}{L + 2} \right) \right) \\
    = &\frac{4}{\left( L + 2 \right)^{2}} \cdot \left\{ \sum_{r = 1}^{L + 1} \frac{3}{8} - \frac{1}{2} \sum_{r = 1}^{L + 1} \cos \left( \frac{2 \pi r}{L + 2} \right) + \frac{1}{8} \sum_{r = 1}^{L + 1} \cos \left( \frac{4 \pi r}{L + 2} \right) \right\} \\
    = &\frac{4}{\left( L + 2 \right)^{2}} \left\{ \frac{3}{8} \left( L + 1 \right) - \frac{1}{2} \left( -1 \right) + \frac{1}{8} \left( -1 \right) \right\} \\
    = &\frac{4}{\left( L + 2 \right)^{2}} \left\{ \frac{3}{8} \left( L + 1 \right) + \frac{3}{8} \right\} = \frac{3}{2 \left( L + 2 \right)}, 
\end{align}
where, in the first line, we have used the fact that:
\begin{align}
    \sin^{2} \left( \frac{\pi r}{L + 2} \right) &= \sin^{2} \left( \frac{\pi r \left( \left( L + 2 \right) - \left( L + 1 \right) \right)}{L + 2} \right) \\
    &= \sin^{2} \left( \pi r - \frac{\pi r \left( L + 1 \right)}{L + 2} \right) \\
    &= \left\{ \underset{= 0}{\underbrace{\sin \left( \pi r \right)}} \cos \left( \frac{\pi r \left( L + 1 \right)}{L + 2} \right) - \sin \left( \frac{\pi r \left( L + 1 \right)}{L + 2} \right) \underset{=(-1)^{r}}{\underbrace{\cos \left( \pi r \right)}} \right\}^{2} \\
    &= \sin^{2} \left( \frac{\pi r \left( L + 1 \right)}{L + 2} \right).
\end{align}
Let us now turn our attention to the off-diagonal part of Eq.~(\ref{eq:avg-prob-diag-off-diag-split}). Recall that this is given as the following:
\begin{align}
    \overline{p}_{\text{off}} = \sum_{\substack{r, s = 1, \\ r \neq s}}^{L + 1} A_{r, s} \cdot \underset{\eqqcolon I_{r, s} \left( T \right)}{\underbrace{\frac{1}{T} \int_{0}^{T} e^{-i \left( E_{r} - E_{s} \right) t} dt}} = \sum_{\substack{r, s = 1, \\ r \neq s}}^{L + 1} A_{r, s} I_{r, s} \left( T \right).
\end{align}
Consider now that
\begin{align}
    I_{r, s} \left( T \right) = \frac{1}{T} \int_{0}^{T} e^{-i \left( E_{r} - E_{s} \right)t} dt = \frac{1}{T} \int_{0}^{T} e^{i \left( E_{s} - E_{r} \right)t} dt = I_{s, r}^{*} \left( T \right).
\end{align} 
The above will allow us to rewrite $\overline{p}_{\text{off}}$ in the following manner, grouping summands corresponding to the indices $(r, s)$ and $(s, r)$ together:
\begin{align}
    \overline{p}_{\text{off}} &= \sum_{\substack{r, s = 1, \\ r \neq s}}^{L + 1} A_{r, s} I_{r, s} \left( T \right) \\
    &= \sum_{1 \leqslant  r < s \leqslant L + 1} A_{r, s} I_{r, s} \left( T \right) + A_{s, r} I_{s, r} \left( T \right) \\
    &= \sum_{1 \leqslant r < s \leqslant L + 1} A_{r, s} \left( I_{r, s} \left( T \right) + I_{r, s}^{*} \left( T \right) \right) \\ 
    &= \sum_{1 \leqslant r < s \leqslant L + 1} A_{r, s} \cdot 2 \mathfrak{Re} \left[ I_{r, s} \left( T \right) \right] \\
    &= \sum_{1 \leqslant r < s \leqslant L + 1} A_{r, s} \cdot 2\mathfrak{Re} \left[ \frac{1}{T} \int_{0}^{T} \left( e^{-i\left( E_{r} - E_{s} \right)t} \right) dt \right] \\
    &= \sum_{1 \leqslant r < s \leqslant L + 1} 2 A_{r, s} \cdot \frac{1}{T} \int_{0}^{T} \mathfrak{Re} \left[ e^{-i \left( E_{r} - E_{s} \right) t} \right] dt \\
    &= \sum_{1 \leqslant r < s \leqslant L + 1} 2 A_{r, s} \cdot \frac{1}{T} \int_{0}^{T} \cos \left( \left( E_{r} - E_{s} \right) t \right) dt \\
    &= \sum_{1 \leqslant r < s \leqslant L + 1} 2 A_{r, s} \frac{\sin \left( \left( E_{r} - E_{s} \right) T \right)}{\left( E_{r} - E_{s} \right) T}.
\end{align}
We therefore have
\begin{align}
    \overline{p}_{T} = \frac{2R + 3}{2 \left( L + 2 \right)} + 2 \hspace{-4mm}\sum_{1 \leqslant r < s \leqslant L + 1}\hspace{-3mm}A_{r, s} \frac{\sin \left( \left( E_{r} - E_{s} \right) T \right)}{\left( E_{r} - E_{s} \right) T}.
\end{align}
We now proceed to compute an upper-bound on $\left\vert \overline{p}_{\text{off}} \right\vert$.
\begin{align}
    \left\vert \overline{p}_{\text{off}} \right\vert &= \left\vert 2 \hspace{-2mm}\sum_{1 \leqslant r < s \leqslant L + 1}\hspace{-2mm}A_{r, s} \frac{\sin \left( \left( E_{r} - E_{s} \right) T \right)}{\left( E_{r} - E_{s} \right) T} \right\vert \leqslant \frac{2}{T} \sum_{1 \leqslant r < s \leqslant L + 1} \frac{\left\vert A_{r, s} \right\vert}{\left\vert E_{r} - E_{s} \right\vert}.
\end{align}
Now, let us introduce 
\begin{align}
    \theta_{r} \coloneqq \frac{\pi r}{L + 2},~\theta_{s} \coloneqq \frac{\pi s}{L + 2}.
\end{align}
for notational simplicity. Note that we have the following:
\begin{align}
    \left\vert E_{r} - E_{s} \right\vert &= \left\vert 2 \left( \cos \theta_{r} - \cos \theta_{s} \right) \right\vert \\
    &= 4 \left\vert \sin \left( \frac{\theta_{r} + \theta_{s}}{2} \right) \sin \left( \frac{\theta_{s} - \theta_{r}}{2} \right) \right\vert,
\end{align}
as well as
\begin{align}
    \left\vert A_{r, s} \right\vert &= \left\vert \frac{4}{\left( L + 2 \right)^{2}} \sin \theta_{r} \sin \theta_{s} \sum_{j = K}^{L} \sin \left( \theta_{r} \left( j + 1 \right) \right) \sin \left( \theta_{s} \left( j + 1 \right) \right) \right\vert \\
    &\leqslant \frac{4}{\left( L + 2 \right)^{2}} \left\vert \sin \theta_{r} \sin \theta_{s} \right\vert \sum_{j = K}^{L} \underset{\leqslant 1}{\underbrace{\left\vert \sin \left( \theta_{r} \left( j + 1 \right) \right) \sin \left( \theta_{s} \left( j + 1 \right) \right) \right\vert}} \\
    &\leqslant \frac{4 \left( R + 1 \right)}{\left( L + 2 \right)^{2}} \left\vert \sin \theta_{r} \sin \theta_{s} \right\vert.
\end{align}
Therefore, we have:
\begin{align}
    \frac{\left\vert A_{r, s} \right\vert}{\left\vert E_{r} - E_{s} \right\vert} &\leqslant \frac{ \frac{4 \left( R + 1 \right)}{\left( L + 2 \right)^{2}} \left\vert \sin \theta_{r} \sin \theta_{s} \right\vert }{4 \left\vert \sin \left( \frac{\theta_{r} + \theta_{s}}{2} \right) \sin \left( \frac{\theta_{s} - \theta_{r}}{2} \right) \right\vert} \\
    &= \frac{\left( R + 1 \right)}{\left( L + 2 \right)^{2}} \cdot \frac{\left\vert \sin \left( \left( \frac{\theta_{r} + \theta_{s}}{2} \right) - \left( \frac{\theta_{s} - \theta_{r}}{2} \right) \right) \sin \left( \left( \frac{\theta_{r} + \theta_{s}}{2} \right) + \left( \frac{\theta_{s} - \theta_{r}}{2} \right) \right) \right\vert}{\left\vert \sin \left( \frac{\theta_{r} + \theta_{s}}{2} \right) \sin \left( \frac{\theta_{s} - \theta_{r}}{2} \right) \right\vert} \\
    &= \frac{\left( R + 1 \right)}{\left( L + 2 \right)^{2}} \cdot \frac{\left\vert \sin^{2} \left( \frac{\theta_{r} + \theta_{s}}{2} \right) - \sin^{2} \left( \frac{\theta_{s} - \theta_{r}}{2} \right) \right\vert}{\left\vert \sin \left( \frac{\theta_{r} + \theta_{s}}{2} \right) \sin \left( \frac{\theta_{s} - \theta_{r}}{2} \right) \right\vert} \\
    &\leqslant \frac{\left( R + 1 \right)}{\left( L + 2 \right)^{2}} \cdot \frac{\left\vert \sin^{2} \left( \frac{\theta_{r} + \theta_{s}}{2} \right) \right\vert}{\left\vert \sin \left( \frac{\theta_{r} + \theta_{s}}{2} \right) \sin \left( \frac{\theta_{s} - \theta_{r}}{2} \right) \right\vert} \\
    &= \frac{\left( R + 1 \right)}{\left( L + 2 \right)^{2}} \cdot \frac{\left\vert \sin \left( \frac{\theta_{r} + \theta_{s}}{2} \right) \right\vert}{ \left\vert \sin \left( \frac{\theta_{s} - \theta_{r}}{2} \right) \right\vert} \leqslant \frac{\left( R + 1 \right)}{\left( L + 2 \right)^{2}} \cdot \frac{1}{\left\vert \sin \left( \frac{\theta_{s} - \theta_{r}}{2} \right) \right\vert} = \frac{\left( R + 1 \right)}{\left( L + 2 \right)^{2}} \cdot \frac{1}{\left\vert \sin \left( \frac{\pi \left( s - r \right)}{2 \left( L + 2 \right)} \right) \right\vert}.
\end{align}
Since $1 \leqslant s - r \leqslant L$, we have that:
\begin{align}
    0 < \frac{\pi}{2 \left( L + 2 \right)} \leqslant \frac{\pi \left( s - r \right)}{2 \left( L + 2 \right)} \leqslant \frac{L}{\left( L + 2 \right)} \cdot \frac{\pi}{2} < \frac{\pi}{2},
\end{align}
allowing us to use Jordan's inequality, as follows:
\begin{align}
    \frac{\left\vert A_{r, s} \right\vert}{\left\vert E_{r} - E_{s} \right\vert} \leqslant \frac{\left( R + 1 \right)}{\left( L + 2 \right)^{2}} \cdot \frac{1}{\left\vert \sin \left( \frac{\pi \left( s - r \right)}{2 \left( L + 2 \right)} \right) \right\vert} \leqslant \frac{R + 1}{\left( L + 2 \right) \left( s - r \right)}.
\end{align}
Going back to $\overline{p}_{\text{off}}$, we have:
\begin{align}
    \left\vert \overline{p}_{\text{off}} \right\vert &\leqslant \frac{2}{T} \sum_{1 \leqslant r < s \leqslant L + 1} \frac{R + 1}{\left( L + 2 \right) \left( s - r \right)} = \frac{2 \left( R + 1 \right)}{T \left( L + 2 \right)} \sum_{1 \leqslant r < s \leqslant L + 1} \frac{1}{\underset{\eqqcolon k}{\underbrace{s - r}}} \\
    &= \frac{2 \left( R + 1 \right)}{T \left( L + 2 \right)} \sum_{k = 1}^{L} \frac{L + 1 - k}{k} = \frac{2 \left( R + 1 \right)}{T \left( L + 2 \right)} \left( \sum_{k = 1}^{L} \frac{L + 1}{k} - \sum_{k = 1}^{L} 1 \right) \\
    &= \frac{2 \left( R + 1 \right)}{T \left( L + 2 \right)} \left( \left( L + 1 \right) H_{L} - L \right) \leqslant \frac{2 \left( R + 1 \right)}{T \left( L + 2 \right)} \left( \left( L + 1 \right) \left( 1 + \log L \right) - L \right) = \frac{2 \left( R + 1 \right)}{T \left( L + 2 \right)} \left( 1 + \left( L + 1 \right) \log L \right),
\end{align}
where, in the fourth equality, $H_{L}$ denotes the $L^{\text{th}}$ harmonic number, and in the inequality that follows, we have used the elementary bound on harmonic numbers: $H_{L} \leqslant 1 + \log L$. Now, recall that $\overline{p}_{T} \geqslant \overline{p}_{\text{diag}} - \left\vert \overline{p}_{\text{off}} \right\vert$, from which we get the following:
\begin{align}
     \overline{p}_{T} &= \overline{p}_{\text{diag}} + \overline{p}_{\text{off}} \\
     &\geqslant \overline{p}_{\text{diag}} - \left\vert \overline{p}_{\text{off}} \right\vert \\ 
     &\geqslant \frac{2R + 3}{2 \left( L + 2 \right)} - \frac{2 \left( R + 1 \right)}{T \left( L + 2 \right)} \left( 1 + \left( L + 1 \right) \log L \right) \\
     &= 1 - \frac{2 \left( L + 2 \right) - 2R - 3}{2 \left( L + 2 \right)} - \frac{2 \left( R + 1 \right)}{T \left( L + 2 \right)} \left( 1 + \left( L + 1 \right) \log L \right) \\
     &= 1 - \frac{2K + 1}{2 \left( L + 2 \right)} - \frac{2 \left( R + 1 \right)}{T \left( L + 2 \right)} \left( 1 + \left( L + 1 \right) \log L \right).
\end{align}
We would like $\overline{p}_{T} \geqslant 1 - \delta$, for some failure probability $\delta > 0$. Let us now require:
\begin{align}
    \frac{2K + 1}{2 \left( L + 2 \right)} \leqslant \frac{\delta}{2} \implies L + 2 \geqslant \frac{2K + 1}{\delta} \implies R \geqslant \left\lceil\frac{2K + 1}{\delta} - K - 2\right\rceil.
\end{align}
Finally,
\begin{align}
    \frac{2 \left( R + 1 \right)}{T \left( L + 2 \right)} \left( 1 + \left( L + 1 \right) \log L \right) \leqslant \frac{\delta}{2} \implies T \geqslant \frac{4 \left( R + 1 \right)}{\delta \left( L + 2 \right)} \left( 1 + \left( L + 1 \right) \log L \right).
\end{align}
That is, in order to have $\overline{p}_{T} \geqslant 1 - \delta$, it is sufficient for the number of padded identity gates and length of time-averaging to scale as:
\begin{align}
    R &\in \mathcal{O} \left( \frac{K}{\delta} \right), \\
    T &\in \mathcal{O} \left( \frac{K}{\delta^{2}} \log \left( \frac{K}{\delta} \right) \right).
\end{align}
This completes the proof.
\end{proof}
Throughout this subsection, we have written the ``done" projector in the following form:
\begin{align}
    \Pi_{\text{done}} = \sum_{j = K}^{L} \lvert j \rangle \langle j \rvert.
\end{align}
However, note that we use a unary encoding to represent the clock states; in particular, we write $\lvert j \rangle_{\text{clock}} = \lvert 1^{j} 0^{L -j} \rangle$, where $j = 0, 1, \ldots, L$. This leads to the $K^{\text{th}}$ qubit in the clock region to be equal to $1$ precisely when $j \geqslant K$. Hence, with this convention in mind, the done projector may be written as the following:
\begin{align}
    \Pi_{\text{done}} = \sum_{j = K}^{L} \lvert j \rangle \langle j \rvert \equiv \lvert 1 \rangle \langle 1 \rvert_{K},
\end{align}
and $I - \Pi_{\text{done}} \equiv \lvert 0 \rangle \langle 0 \rvert_{K}$. Intuitively, we may thus think of the completed-clock information as being accessible by the following single-qubit Pauli $Z$ operator acting on the $K^{\text{th}}$ clock qubit:
\begin{align}
    -Z_{K}^{\text{clock}} = - \lvert 0 \rangle \langle 0 \rvert_{K} + \lvert 1 \rangle \langle 1 \rvert_{K},
\end{align}
such that the $+1$-eigenstates of the above operator are given by clock states in the region corresponding to $j \geqslant K$, and the $-1$-eigenstates by those lying in the region where $j < K$. Our choice of observable for which we will show classical hardness will thus be given by the following two-body operator:
\begin{align}
    \label{eq:hard-observable-q}
    Q \coloneqq Z_{1}^{\text{work}} \otimes \left( - Z_{K}^{\text{clock}} \right) &= Z_{1}^{\text{work}} \otimes \left( -\lvert 0 \rangle \langle 0 \rvert_{K} \right) + Z_{1}^{\text{work}} \otimes \lvert 1 \rangle \langle 1 \rvert_{K} \\
    &= Z_{1}^{\text{work}} \otimes \Pi_{\text{done}} - Z_{1}^{\text{done}} \otimes \left( I - \Pi_{\text{done}} \right),
\end{align}
where the operator $Z_{1}^{\text{work}} = Z \otimes I \otimes \ldots \otimes I$ is acting on the first qubit of the fully simulated circuit (i.e., right before padding with identities).  When convenient, we will drop the superscripts denoting which subspace the tensored operators belong to. 

\vspace{3mm}
\noindent\textbf{Remark.} Note that the observable $Q$ above is indeed in the set of operators $\left(Q_a\right)_{a \in \left[ M \right]}$ used for Hamiltonian learning (see Subsection \ref{subsubsec:background-haah-kothari-tang}). Indeed, for this to be the case, we need $Q$ to be of the form $Q_a = i[P_a,E_a] = 2iP_aE_a$ for $P_a$ a single-qubit Pauli operator that anti-commutes with a Pauli $E_a$ appearing in the decomposition of our Hamiltonian. For that, we make use of a Pauli from the transition term 
\begin{align}
    \ket{1,1,0}\bra{1,0,0}_{K', K'+1, K'+2}^{\text{clock}} \otimes {U'}^{\text{ work}}_{K'+1},
\end{align}
which corresponds to the first transition in the idling phase. ${U'}^{\text{work}}_{K'+1}$ is assumed to act on the first qubit of the fully simulated circuit and to be freely defined in our Hamiltonian family (but effectively set to identity). This allows us to use $E_a = IXI_{mK, mK+1, mK+2}^{\text{clock}} \otimes Z_1^{\text{work}}$ and simply choose the associated single-qubit $P_a = IYI_{mD, mK+1, mK+2}^{\text{clock}} \otimes I^{\text{work}}$ such that $O = Q_a = 2iP_aE_a = 2\ IZI_{mK, mK+1, mK+2}^{\text{clock}} \otimes Z_1^{\text{work}}$. The factor $-2$ is irrelevant.

\subsection{Classical hardness of learning}
\label{subsec:classical-hardness-learning-appendix}

\begin{theorem}[Classical hardness of learning, formal]
\label{thm: classical hardness app}
    For any $(\mathsf{Promise})\BQP$-complete language, there exists a Hamiltonian family $\mathsf{H}_{\BQP}$ specifying a concept class $\mathcal{C}_{U_{\text{enc}}, T}^{\mathsf{H}_{\BQP}}$ as per Def.~\ref{def:learning_problem_concept_class} and a family of input distributions following the specification of Def.~\ref{def: input distribution} such that no randomized polynomial-time classical algorithm $\mathcal{A}_c$ satisfies the PAC-learning condition of Def.~\ref{def: learning condition} for this concept class, even when restricted to a constant learning error $\varepsilon < 1/48$, unless $\BQP \subseteq \mathsf{P/poly}$.
\end{theorem}

\begin{proof}
We consider $Q$ as defined in Eq.~(\ref{eq:hard-observable-q}), and begin by computing the expectation value of this observable evaluated on the time-evolved history state $\lvert \Psi_{z} \left( t \right) \rangle$.
\begin{align}
    q \left( z, t \right) &\coloneqq \operatorname{Tr} \left[ Q \lvert \Psi_{z} \left( t \right) \rangle \langle \Psi_{z} \left( t \right) \rvert \right] \\
    &= \operatorname{Tr} \left[ \left( Z_{1} \otimes \left( -Z_{K} \right) \right) \lvert \Psi_{z} \left( t \right) \rangle \langle \Psi_{z} \left( t \right) \rvert \right] \\
    &= \operatorname{Tr} \left[ \left( Z_{1} \otimes \left( -Z_{K} \right) \right) \sum_{j = 0}^{L} \sum_{k = 0}^{L} \alpha_{j} \left( t \right) \alpha_{k}^{*} \left( t \right) \lvert \eta_{j} \left( z \right) \rangle \langle \eta_{k} \left( z \right) \rvert  \right] \\
    &= \sum_{j = 0}^{L} \sum_{k = 0}^{L} \alpha_{j} \left( t \right) \alpha_{k}^{*} \left( t \right) \operatorname{Tr} \left[ \left( Z_{1} \otimes \left( -Z_{K} \right) \right) \left( \lvert \psi_{j} \left( z \right) \rangle \langle \psi_{k} \left( z \right) \rvert \otimes \lvert j \rangle \langle k \rvert \right) \right] \\
    &= \sum_{j = 0}^{L} \sum_{k = 0}^{L} \alpha_{j} \left( t \right) \alpha_{k}^{*} \left( t \right) \operatorname{Tr} \left[ Z_{1} \lvert \psi_{j} \left( z \right) \rangle \langle \psi_{k} \left( z \right) \rvert \right] \operatorname{Tr} \left[ \left( -Z_{K} \right) \lvert j \rangle \langle k \rvert \right] \\
    &= \sum_{j = 0}^{L} \sum_{k = 0}^{L} \alpha_{j} \left( t \right) \alpha_{k}^{*} \left( t \right) \langle \psi_{k} \left( z \right) \rvert Z_{1} \lvert \psi_{j} \left( z \right) \rangle \underset{s_{j} \delta_{j, k}}{\underbrace{\langle k \rvert \left( -Z_{K} \right) \lvert j \rangle}},
\end{align}
where we set
\begin{align}
    s_{j} \coloneqq \begin{cases}
        +1,~\text{if}~j \geqslant K, \\
        -1,~\text{if}~j < K.
    \end{cases}
\end{align}
We resume our computation of $q \left( z, t \right)$:
\begin{align}
    q \left( z, t \right) &= \sum_{j = 0}^{L} s_{j} \left\vert \alpha_{j} \left( t \right) \right\vert^{2} \langle \psi_{j} \left( z \right) \rvert Z_{1} \lvert \psi_{j} \left( z \right) \rangle \\
    &= \sum_{j = K}^{L} \left\vert \alpha_{j} \left( t \right) \right\vert^{2} \langle \psi_{j} \left( z \right) \rvert Z_{1} \lvert \psi_{j} \left( z \right) \rangle - \sum_{j = 0}^{K - 1} \lvert \alpha_{j} \left( t \right) \rvert^{2} \langle \psi_{j} \left( z \right) \rvert Z_{1} \lvert \psi_{j} \left( z \right) \rangle \\
    &= \underset{\eqqcolon m \left( z \right)}{\underbrace{\langle \psi_{K} \left( z \right) \rvert Z_{1} \lvert \psi_{K} \left( z \right) \rangle}} \underset{\eqqcolon p_{\text{done}} \left( t \right)}{\underbrace{\sum_{j = K}^{L} \lvert \alpha_{j} \left( t \right) \rvert^{2}}} - \underset{\eqqcolon r \left( z, t \right)}{\underbrace{\sum_{j = 0}^{K - 1} \lvert \alpha_{j} \left( t \right) \rvert^{2} \langle \psi_{j} \left( z \right) \rvert Z_{1} \lvert \psi_{j} \left( z \right) \rangle}} \\
    &= m \left( z \right) p_{\text{done}} \left( t \right) - r \left( z, t \right),
\end{align}
where we denote by $r \left( z, t \right)$ a ``garbage" term, corresponding to stages of the computation where $j < K$ and therefore not necessarily encoding the answer of the $\mathsf{BQP}$ computation. We note that:
\begin{align}
    \lvert r \left( z, t \right) \rvert &= \left\vert \sum_{j = 0}^{K - 1} \lvert \alpha_{j} \left( t \right) \rvert^{2} \langle \psi_{j} \left( z \right) \rvert Z_{1} \lvert \psi_{j} \left( z \right) \rangle \right\vert \\
    &\leqslant \sum_{j = 0}^{K - 1} \lvert \alpha_{j} \left( t \right) \rvert^{2} \left\vert \langle \psi_{j} \left( z \right) \rvert Z_{1} \lvert \psi_{j} \left( z \right) \rangle \right\vert \\ 
    &\leqslant \sum_{j = 0}^{K - 1} \left\vert \alpha_{j} \left( t \right) \right\vert^{2} = 1 - p_{\text{done}} \left( t \right).
\end{align}
Now, consider:
\begin{align}
    \overline{q} \left( z \right) &\coloneqq \frac{1}{T} \int_{0}^{T} q \left( z, t \right) dt = \frac{1}{T} \int_{0}^{T} \left( m \left( z \right) p_{\text{done}} \left( t \right) - r \left( z, t \right) \right) dt \\
    &= m \left( z \right) \cdot \underset{\overline{p}_{T}}{\underbrace{\frac{1}{T} \int_{0}^{T}  p_{\text{done}} \left( t \right) dt}} - \underset{\eqqcolon \overline{r} \left( z \right)}{\underbrace{\frac{1}{T} \int_{0}^{T} r \left( z, t \right) dt}} = m \left( z \right) \overline{p}_{T} - \overline{r} \left( z \right). 
\end{align}
Let us note that:
\begin{align}
    \lvert \overline{r} \left( z \right)\rvert = \left\vert \frac{1}{T} \int_{0}^{T} r \left( z, t \right) dt \right\vert \leqslant \frac{1}{T} \int_{0}^{T} \lvert r \left( z, t \right) \rvert dt \leqslant \frac{1}{T} \int_{0}^{T} \left( 1 - p_{\text{done}} \left( t \right) \right) dt = 1 - \overline{p}_{T}.
\end{align}
Thus, for some $\delta' > 0$, if by Lemma~\ref{lem:avg-prob-mass-completed-subspace-clock-walk} we have that $\overline{p}_{T} \geqslant 1 - \delta'$, we also have that $\lvert \overline{r} \left( z \right) \rvert \leqslant \delta'$. Now, consider the following:
\begin{align}
    \left\vert \overline{q} \left( z \right) - m \left( z \right) \right\vert &= \left\vert m \left( z \right) \overline{p}_{T} - \overline{r} \left( z \right) - m \left( z \right) \right\vert \\
    &= \left\vert m \left( z \right) \left( \overline{p}_{T} - 1 \right) - \overline{r} \left( z \right) \right\vert \\
    &\leqslant \left\vert m \left( z \right) \right\vert \left\vert 1 - \overline{p}_{T} \right\vert + \left\vert \overline{r} \left( z \right) \right\vert \leqslant \delta' + \delta' = 2 \delta'.
\end{align}
Now, for $\gamma > 0$, assume that the underlying circuit $U_{\mathsf{BQP}}$ decides a language $\mathcal{L}$ such that, for $z \in \left\{ 0, 1 \right\}^{m}$, we have
\begin{align}
    z \in \mathcal{L} &\implies m \left( z \right) \leqslant -\gamma \implies \overline{q} \left( z \right) \leqslant -\left( \gamma - 2 \delta' \right), \\
    \text{and}\quad z \not\in \mathcal{L} &\implies m \left( z \right) \geqslant \gamma \implies \overline{q} \left( z \right) \geqslant \gamma - 2 \delta'. 
\end{align}
We now show that an efficient learner for the concept class $\mathcal{C}_{U_{\text{enc}}, T}^{\mathsf{H}_{\mathsf{BQP}}}$ would imply an efficient learner for $q$ under an \emph{arbitrary} distribution on $z$. Fix any distribution \(\mu\) over \(\{0,1\}^{m}\). By the definition of the family of input distributions \(\{\mathcal{D}_i\}_i\) (Def.~\ref{def: input distribution}), we can consider a distribution \(\mathcal{D}_\mu\) over $x=(x_1, z, x_{[m+2,p(n)]})\in\{0,1\}^{p(n)}$ and $t\in[0,T]$ such that:
\begin{itemize}
    \item $t$ is uniformly distributed over $[0,T]$;
    \item \(x_1\) is uniform in \(\{0,1\}\);
    \item conditioned on \(x_1=1\), the bits \( z = x_2 \cdots x_{m+1} \) are distributed according to \(\mu\), and the remaining bits $x_{m+2}, \ldots, x_{p(n)}$ are deterministically set to $0$;
    \item conditioned on \(x_1=0\), the remaining bits are distributed uniformly in $\{0,1\}^{p(n)-1}$.
\end{itemize}
Let $c_{\lambda_{\mathsf{BQP}}} \left( x, t \right)$ be the target concept and, in particular,
\begin{align}
    \left( c_{\lambda_{\mathsf{BQP}}} (\left( \bm 1, z, \bm 0 \right), t) \right)_{i^{\star}} = q \left( z, t \right)
\end{align}
be the distinguished component whose output equals the expectation value of the time-evolved history state on the observable $Q$. If our PAC-learning condition is met, there exists a learner that succeeds for the target concept \(c_{\lambda_{\mathsf{BQP}}}\)  under any such distribution \(\mathcal{D}_\mu\). Assume, towards contradiction, that $\mathcal{C}_{U_{\text{enc}}, T}^{\mathsf{H}_{\mathsf{BQP}}}$ is efficiently learnable for an accuracy parameter \(\varepsilon\) to be defined later, and an arbitrary confidence parameter \(\delta>0\). Then, there is a randomized polynomial-time learner which, given training data generated from \(c_{\lambda_{\mathsf{BQP}}}\) under \(\mathcal{D}_\mu\), outputs, with probability at least \(1-\delta\), a hypothesis \(h(\bm x,t)\) satisfying
\begin{equation}
\underset{(\bm x,t)\sim \mathcal{D}_\mu}{\mathbb{E}}
\left[
\|c_{\lambda_{\mathsf{BQP}}}(x, t) - h(x, t)\|_\infty
\right]
\leqslant \varepsilon.
\end{equation}
Now, consider that:
\begin{align}
    \underset{\substack{\left( x, t \right) \sim \mathcal{D}_{\mu}}}{\mathbb{E}} \left[ \left\vert \left( h \left( x, t \right) \right)_{i^{\star}} - \left( c_{\lambda_{\mathsf{BQP}}} \left( x, t \right) \right)_{i^{\star}} \right\vert \right] \leqslant \underset{(\bm x,t)\sim \mathcal{D}_\mu}{\mathbb{E}} \left[ \|c_{\lambda_{\mathsf{BQP}}}(x, t) - h(x, t)\|_\infty \right] \leqslant \varepsilon.
\end{align}
Restricting the above expectation to the event $x_{1} = 1$, we have the following:
\begin{align}
    &\underset{\substack{z \sim \mu, \\ t \sim \mathsf{Unif} \left( \left[ 0, T \right] \right)}}{\mathbb{E}} \left[ \left\vert \left( h \left( \left( \boldsymbol{1}, z, \boldsymbol{0} \right), t \right) \right)_{i^{\star}} - q \left( z, t \right) \right\vert \right] \cdot \Pr \left[ x_{1} = 1 \right] \leqslant \varepsilon \\
    \implies &\underset{\substack{z \sim \mu, \\ t \sim \mathsf{Unif} \left( \left[ 0, T \right] \right)}}{\mathbb{E}} \left[ \left\vert \left( h \left( \left( \boldsymbol{1}, z, \boldsymbol{0} \right), t \right) \right)_{i^{\star}} - q \left( z, t \right) \right\vert \right] \leqslant 2\varepsilon. 
\end{align}
Let us define the time-averaged hypothesis as
\begin{align}
    \overline{h} \left( z \right) \coloneqq \underset{t \sim \mathsf{Unif} \left( \left[ 0, T \right] \right)}{\mathbb{E}} \left[ \left( h \left( \left( \boldsymbol{1}, z, \boldsymbol{0} \right), t \right) \right)_{i^{\star}} \right] = \frac{1}{T} \int_{0}^{T} \left( h \left( \left( \boldsymbol{1}, z, \boldsymbol{0} \right), t \right) \right)_{i^{\star}} dt.
\end{align}
Recall also that we had defined
\begin{align}
    \overline{q} \left( z \right) &= \frac{1}{T} \int_{0}^{T} q \left( z, t \right) dt = \underset{t \sim \mathsf{Unif} \left( \left[ 0, T \right] \right)}{\mathbb{E}} \left[ q \left( z, t \right) \right].
\end{align}
Now, consider the following:
\begin{align}
    \underset{z \sim \mu}{\mathbb{E}} \left[ \left\vert \overline{h} \left( z \right) - \overline{q} \left( z \right) \right\vert \right] &\leqslant \underset{z \sim \mu}{\mathbb{E}} \left[ \left\vert \frac{1}{T} \int_{0}^{T} \left[ \left( h \left( \left( \boldsymbol{1}, z, \boldsymbol{0} \right), t \right) \right)_{i^{\star}} - q \left( z, t \right) \right] dt \right\vert \right] \\
    &\leqslant \underset{z \sim \mu}{\mathbb{E}} \left[ \frac{1}{T} \int_{0}^{T} \left\vert \left( h \left( \left( \boldsymbol{1}, z, \boldsymbol{0} \right), t \right) \right)_{i^{\star}} - q \left( z, t \right) \right\vert dt \right] \\
    &= \underset{\substack{z \sim \mu, \\ t \sim \mathsf{Unif} \left( \left[ 0, T \right] \right)}}{\mathbb{E}} \left[ \left\vert \left( h \left( \left( \boldsymbol{1}, z, \boldsymbol{0} \right), t \right) \right)_{i^{\star}} - q \left( z, t \right) \right\vert \right] \leqslant 2 \varepsilon.
\end{align}
Let us choose $\delta' = \gamma/4$, so that we have $\gamma - 2\delta' = \gamma/2$. Further, let $\varepsilon < \gamma/16$. This leads to:
\begin{align}
    \underset{z \sim \mu}{\mathbb{E}} \left[ \left\vert \overline{h} \left( z \right) - \overline{q} \left( z \right) \right\vert \right] \leqslant 2\varepsilon \leqslant \frac{\gamma}{8}.
\end{align}
Now, for the particular choice of $\delta' = \gamma/4$, we recall that the time-averaged distinguished component $\overline{q} \left( z \right)$ has a margin of at least $\gamma/2$ around zero. That is,
\begin{align}
    z \in \mathcal{L} &\implies \overline{q} \left( z \right) \leqslant - \gamma/2, \\
    \text{and}\quad z \notin \mathcal{L} &\implies \overline{q} \left( z \right) \geqslant \gamma/2.
\end{align}
Therefore, if the sign of $\overline{h} \left( z \right)$ differs from the sign of $\overline{q} \left( z \right)$, then we have
\begin{align}
    \left\vert \overline{h} \left( z \right) - \overline{q} \left( z \right) \right\vert \leqslant \frac{\gamma}{2}.
\end{align}
From the above, we may conclude the following:
\begin{align}
    \Pr_{z \sim \mu} \left[ \operatorname{sign} \left(\overline{h} \left( z \right) \right) \neq \operatorname{sign} \left( \overline{q} \left( z \right)  \right) \right] &\leqslant \Pr_{z \sim \mu} \left[ \left\vert \overline{h} \left( z \right) - \overline{q} \left( z \right) \right\vert \geqslant \frac{\gamma}{2} \right] \\
    &\leqslant \frac{\underset{z \sim \mu}{\mathbb{E}} \left[ \left\vert \overline{h} \left( z \right) - \overline{q} \left( z \right) \right\vert \right]}{\gamma/2} \\
    &\leqslant \frac{\gamma/8}{\gamma/2} = \frac{1}{4} < \frac{1}{2},\label{eq:boosting-requirement}
\end{align}
where the second inequality follows from an application of Markov's inequality. Since the sign of $\overline{q} \left( z \right)$ encodes membership of $z$ in $\mathcal{L}$, thresholding $\overline{h} \left( z \right)$ at zero gives a Boolean classifier with error at most $1/4$ under the arbitrary distribution $\mu$. The assumed learner thus yields a distribution-free weak learner for $\mathcal{L}$ with constant error $1/4$. Taking the usual $\mathsf{BQP}$ gap of $\gamma = 1/3$ and choosing $\delta' = \gamma/4 = 1/12$, we finally get $\varepsilon < 1/48$, yielding a classification error of at most $1/4$. At this point, we invoke the standard consequence of Schapire's theorem: if a Boolean concept class is weakly learnable under arbitrary distributions, then it is strongly learnable~\cite{schapire1990strength}. On the other hand, if it is strongly learnable in randomized polynomial time, then for each input length there exists a polynomial-size circuit computing the target function exactly, i.e., the concept class lies in $\mathsf{P}\slash\mathsf{poly}$. Applying this to the language $\mathcal{L}$, we conclude that
\begin{equation}
\mathcal{L}\in \mathsf{P}\slash\mathsf{poly}.
\end{equation}
Since $\mathcal{L}$ has been chosen to be ($\mathsf{Promise}$) $\mathsf{BQP}$-complete, this implies
\begin{equation}
\mathsf{BQP}\subseteq \mathsf{P}\slash\mathsf{poly},
\end{equation}
contradicting the assumption. Therefore, no randomized polynomial-time classical algorithm can satisfy the PAC-learning condition for the concept
class $\mathcal{C}_{U_{\text{enc}}, T}^{\mathsf{H}_{\mathsf{BQP}}}$, unless $\mathsf{BQP}\subseteq \mathsf{P}\slash \mathsf{poly}$.
\end{proof}
\noindent \textbf{Remark.} To put this result in perspective, we note that the trivial learning performance in this task is obtained for a hypothesis $h(x,t) = (0)_{i\in[M']},\ \forall x,t$, which achieves $\varepsilon = 1$. The proof above shows hardness for $\varepsilon < 1/48$. Pushing the constants to their limits, namely, by taking the $\BQP$ acceptance threshold to be $\gamma = 1 - 1/\mathsf{poly}(n)$, the failure probability of Hamiltonian simulation to be $\delta'=1 -1/\mathsf{poly}(n)$, and by considering a non-uniform distribution over $x_1$ (we only need to have probability weight $1/\mathsf{poly}(n)$ over $x_1=0$ and over short times $t \in [0,t^*]$ in order to get quantum learnability) as to get $\Pr[x_1=1] = 1 - 1/\mathsf{poly}(n)$, we can easily bring the hardness threshold to $\varepsilon = 1/2 - 1/\mathsf{poly}(n)$. The only unavoidable factor of $1/2$ originates from the advantage requirement for boosting in Eq.~(\ref{eq:boosting-requirement}).

\section{Quantum learnability}
\label{app:quantum-learnability}

In this appendix, we supply technical details related to the proof of quantum learnability of the concept class defined in Def.~\ref{def:learning_problem_concept_class}. First, in Subsection \ref{subsec:additional-definitions-app}, we provide a few technical definitions relevant to the analysis of quantum learnability. Then, following the main ideas explicated in Section \ref{sec:quantum-learnability}, we show a tail bound on $\mathcal{O} \left( t^{2} \right)$ terms of $\mathcal{F}_{a, t}$, which may be found in Lemma \ref{lem:tail-bound-higher-order} of Subsection \ref{subsec:tail-bound-f-a-t}. In Lemma \ref{lem:error-prop-ham-coeff-exp-values} of Subsection \ref{subsec:hl-precision}, we then determine how errors on Hamiltonian coefficients propagate to expectation values of time-evolved states, so as to obtain the necessary training precision for $\varepsilon$-accurate inference accuracy. In Subsection \ref{subsec:sample-time-complexity-analysis}, we provide bounds on the ideal estimator $\tilde{\mathcal{F}}_{a, t^{*}}$ (Lemma \ref{lem:bounds-time-averaged-integral}), followed by a thorough analysis of the sample complexity of PAC-learning our main concept class (Lemmas \ref{lem:short-time-sample-complexity-analysis}, \ref{lem:tot-sample-complexity-analysis}, and \ref{lem:sample-complexity-inference-classical-shadows}). We conclude by formally restating our quantum learnability result in Theorem \ref{thm:quantum-learnability-formal}, where we also provide the time complexity of our quantum learning algorithm.

\subsection{Additional definitions}
\label{subsec:additional-definitions-app}
We begin this subsection with the definition of a commutator and the Hadamard formula.
\begin{definition}[Commutator]
    \label{def:commutator}
    Given operators $A, B \in \mathbb{C}^{D \times D}$, the commutator of $A$ and $B$ is defined as $\left[ A, B \right] \coloneqq AB - BA$. The $k$-fold nested commutator of $A$ and $B$ is given by $\left[ A, B \right]_{k}$, which is defined recursively as 
    \begin{align}
        \left[ A, B \right]_{k} \coloneqq \begin{cases}
            B, &k = 0, \\
            A\left[ A, B \right]_{k - 1} - \left[ A, B \right]_{k - 1}A, &k\geqslant 1.
        \end{cases}
    \end{align}
    That is, $\left[ A, B \right]_{k} = \underset{\text{$k$ times}}{\underbrace{[ A, [ A, \ldots, [ A}}, B ] \ldots ] ]$.
\end{definition}
\begin{lemma}[Hadamard formula]
    \label{lem:hadamard-formula}
    For $A, B \in \mathbb{C}^{D \times D}$, it holds that
    \begin{align}
        \label{eq:hadamard-formula}
        e^{-iA} B e^{iA} = \sum_{k = 0}^{\infty} \frac{\left( -i \right)^{k}}{k!} \left[ A, B \right]_{k}.
    \end{align}
\end{lemma}
The usefulness of Def.~\ref{def:commutator} and Lemma~\ref{lem:hadamard-formula} will become readily apparent, since our analysis of quantum learnability will involve manipulating and expanding time-evolved operators of the form shown on the L.H.S.\ of Eq.\ (\ref{eq:hadamard-formula}). However, our calculations will involve nested commutators that are more general than $\left[ A, B \right]_{k}$, such as $\left[ M_{1}, \left[ M_{2}, \left[ \ldots, \left[ M_{\ell - 1}, M_{\ell} \right] \ldots \right] \right] \right]$, which are nested commutators constructed using more than two operators. In particular, we will be interested in scenarios when such quantities are non-zero. This leads us to the definition of a \textit{cluster}.  
\begin{definition}[Cluster, Definition 2.17,~\cite{bakshi2023learningquantumhamiltonianstemperature}]
    \label{def:cluster}
    Let $M_{1}, M_{2}, \ldots, M_{\ell} \in \mathbb{C}^{D \times D}$. We say that the ordered $\ell$-tuple $\left( M_{1}, M_{2}, \ldots, M_{\ell} \right)$ forms a cluster if, for all $a \in \left[ \ell - 1 \right]$, $\operatorname{supp} \left( M_{a + 1} \right) \cap \left( \operatorname{supp} \left( M_{1} \right) \cup \operatorname{supp} \left( M_{2} \right) \cup \ldots \cup \operatorname{supp} \left( M_{a} \right) \right) \neq \emptyset$.
\end{definition}
To see the motivation behind Def.~\ref{def:cluster}, note that the quantity $\left[ M_{1}, \left[ M_{2}, \left[ \ldots , \left[ M_{\ell - 1}, M_{\ell} \right] \ldots \right] \right] \right]$ is non-zero only when  $\operatorname{supp} \left( M_{\ell - 1} \right) \cap \operatorname{supp} \left( M_{\ell} \right) \neq \emptyset$, $\operatorname{supp} \left( M_{\ell - 2} \right)$ overlaps with $\operatorname{supp} \left( M_{\ell - 1} \right) \cup \operatorname{supp} \left( M_{\ell} \right)$, and so on. This is exactly the condition imposed by Def.~\ref{def:cluster}. The final technical ingredient we will require is a cluster-counting lemma that has  originally been shown in Ref.~\cite{bakshi2023learningquantumhamiltonianstemperature}, which we restate below for completeness:
\begin{lemma}[Counting clusters, Item (b) of Lemma 2.18, ~\texorpdfstring{\cite{bakshi2023learningquantumhamiltonianstemperature}}{Bakshi et al.}]
    \label{lemma:cluster-counting-blmt}
    Let $\mathcal{E} \subset \mathcal{P}$ be a set of Pauli operators, where $\mathcal{P}$ is the set of $n$-qubit Pauli strings of the form $P_{1} \otimes \ldots \otimes P_{n}$, where each $P_{i} \in \left\{ I, X, Y, Z \right\}$. Let every $P \in \mathcal{E}$ satisfy $\left\vert \operatorname{supp} \left( P \right) \right\vert \leqslant \mathfrak{K}$ and let the dual-interaction graph associated with $\mathcal{E}$ have maximum degree $\mathfrak{d}$. Let $H_{1}, H_{2}, \ldots, H_{\ell}$ be linear combinations of elements in $\mathcal{E}$, such that $H_{i} \coloneqq \sum_{P \in \mathcal{E}} \lambda_{i, P} P$, for all $i \in \left[ \ell \right]$. Let also $A \in \mathcal{E}$. Then, we may write 
    \begin{align}
        \label{eq:nested-comm-cluster-expansion}
        \left[ H_{1}, \left[ H_{2}, \left[ \ldots, \left[ H_{\ell}, A \right] \ldots \right] \right] \right] = 2^{\ell} \hspace{-8mm} \sum_{\substack{P_{1}, P_{2}, \ldots, P_{\ell} \in \mathcal{E}, \\ \text{$\left( A, P_{1}, \ldots, P_{\ell} \right)$ is a cluster.}}} \hspace{-7mm} c_{P_{1}, P_{2}, \ldots, P_{\ell}} Q_{P_{1}, P_{2}, \ldots, P_{\ell}} \prod_{j = 1}^{\ell} \lambda_{j, P_{j}},
    \end{align}
    where $c_{P_{1}, P_{2}, \ldots, P_{\ell}} \in \left\{ 0, \pm 1, \pm i \right\}$ and $Q_{P_{1}, P_{2}, \ldots, P_{\ell}}$ are certain Pauli operators. Also, the number of terms constituting the sum in Eq.\ (\ref{eq:nested-comm-cluster-expansion}) is given by $\ell! \left( \mathfrak{d} + 1 \right)^{\ell}$. 
\end{lemma}

\subsection{Tail bound}
\label{subsec:tail-bound-f-a-t}
\begin{lemma}[Tail bound for higher-order terms of \texorpdfstring{$\mathcal{F}_{a, t}$}{F_{a, t}}]
    \label{lem:tail-bound-higher-order}
    Let $H = \sum_{a \in \left[ M \right]} \lambda_{a} E_{a}$ be an $n$-qubit Hamiltonian as defined in Def.~\ref{def:hamiltonian}. Let $U \coloneqq e^{-iHt}$, for $t > 0$. For each $a \in \left[ M \right]$, let $P_{a}$ be any single-qubit Pauli that anti-commutes with $E_{a}$ and $Q_{a} = i \left[ P_{a}, E_{a} \right] = 2 i P_{a} E_{a}$. Then, define 
    \begin{align}
        \label{eq:mathcal_f_a_t}
        \mathcal{F}_{a, t} \coloneqq \frac{1}{D} \operatorname{Tr} \left[ Q_{a} U P_{a} U^{\dagger} \right].
    \end{align}
    Then, it holds true that
    \begin{align}
        \left\vert \mathcal{F}_{a, t} - 4 t \lambda_{a} \right\vert \leqslant \frac{8t^{2} \left( \mathfrak{d} + 1 \right)^{2}}{1 - 2t \left( \mathfrak{d} + 1 \right)},
    \end{align}
    where $\mathfrak{d}$ is the degree of $\mathfrak{G}$, the dual-interaction graph of $H$, provided that $2t \left( \mathfrak{d} + 1 \right) < 1$.
\end{lemma}

\begin{proof}
    Consider 
    \begin{align}
        U P_{a} U^{\dagger} &= e^{-iHt} P_{a} e^{iHt} \\
        \nonumber
        &= \sum_{k = 0}^{\infty} \frac{\left( -it \right)^{k}}{k!} \left[ H, P_{a} \right]_{k} \\
         \nonumber
        &= \sum_{k = 0}^{\infty} \frac{\left( -it \right)^{k}}{k!} \hspace{-1cm} \sum_{\substack{a_{1}, a_{2}, \ldots, a_{k} \in \left[ M \right] \colon \\ 
         \nonumber\forall j \in \left[ k \right], \operatorname{supp} \left( E_{a_{j}} \right) \cap \left( \bigcup_{i = j + 1}^{k} \operatorname{supp} \left( E_{a_{i}} \right) \cup \operatorname{supp} \left( A \right) \right) \neq \emptyset}} \hspace{-1cm} \underset{\eqqcolon C_{a_{1}, a_{2}, \ldots, a_{k}}}{\underbrace{\left[ E_{a_{1}}, \left[ E_{a_{2}}, \left[ \ldots , \left[ E_{a_{k}}, P_{a} \right] \ldots \right] \right] \right]}} \lambda_{a_{1}} \lambda_{a_{2}} \ldots \lambda_{a_{k}} \\
          \nonumber
        &= \sum_{k = 0}^{\infty} t^{k} q_{k} \left( \lambda_{1}, \lambda_{2}, \ldots, \lambda_{M} \right), 
         \nonumber
    \end{align}
    where, in the final equality, according to the notation of Ref.~\cite{Haah_2024}, $q_{k} \left( \lambda_{1}, \lambda_{2}, \ldots, \lambda_{M} \right)$ is a degree-$k$, homogeneous, matrix-valued polynomial in the Hamiltonian coefficients $\bigl( \lambda_{a} \bigr)_{a \in \left[ M \right]}$. Also, note that, in the penultimate equality, the interior sum is being taken over clusters of the form $\left( E_{a_{1}}, E_{a_{2}}, \ldots, E_{a_{k}}, P_{a} \right)$, which is to ensure that the summands are non-zero. This condition is exactly the same as requiring that the $(k+1)$-tuple $\left( E_{a_{1}}, E_{a_{2}}, \ldots, E_{a_{k}}, P_{a} \right)$ is a cluster as in Def.~\ref{def:cluster}. Now, consider 
    \begin{align}
        \left\Vert q_{k} \left( \lambda_{1}, \lambda_{2}, \ldots, \lambda_{M} \right) \right\Vert_{\infty} &= \left\Vert \frac{\left( -i \right)^{k}}{k!} \hspace{-1cm} \sum_{\substack{a_{1}, a_{2}, \ldots, a_{k} \in \left[ M \right] \colon \\ 
        \nonumber\forall j \in \left[ k \right], \operatorname{supp} \left( E_{a_{j}} \right) \cap \left( \bigcup_{i = j + 1}^{k} \operatorname{supp} \left( E_{a_{i}} \right) \cup \operatorname{supp} \left( P_{a} \right) \right) \neq \emptyset}} \hspace{-2cm} C_{a_{1}, a_{2}, \ldots, a_{k}} \lambda_{a_{1}} \lambda_{a_{2}} \ldots \lambda_{a_{k}} \right\Vert_{\infty} \\
        \nonumber
        &\leqslant \frac{1}{k!} \hspace{-1cm} \sum_{\substack{a_{1}, a_{2}, \ldots, a_{k} \in \left[ M \right] \colon \\ 
        \nonumber\forall j \in \left[ k \right], \operatorname{supp} \left( E_{a_{j}} \right) \cap \left( \bigcup_{i = j + 1}^{k} \operatorname{supp} \left( E_{a_{i}} \right) \cup \operatorname{supp} \left( P_{a} \right) \right) \neq \emptyset}}  \hspace{-0.75cm} \underset{\leqslant 2^{k}}{\underbrace{\left\Vert C_{a_{1}, a_{2}, \ldots, a_{k}} \right\Vert_{\infty}}} \underset{\leqslant 1}{\underbrace{\left\vert \lambda_{a_{1}} \lambda_{a_{2}} \ldots \lambda_{a_{k}} \right\vert}} \\
        \nonumber
        &\leqslant \frac{2^{k}}{k!} k! \left( \mathfrak{d} + 1 \right)^{k} = \left( 2 \left( \mathfrak{d} + 1 \right) \right)^{k},
    \end{align}
    where the bound on $C_{a_{1}, a_{2}, \ldots, a_{k}}$ follows from iterative use of the trivial commutator bound
    \begin{align}
        \left\Vert C_{a_{1}, a_{2}, \ldots, a_{k}} \right\Vert_{\infty} &= \left\Vert \left[ E_{a_{1}}, \left[ E_{a_{2}}, \left[ \ldots , \left[ E_{a_{k}}, P_{a} \right] \ldots \right] \right] \right] \right\Vert_{\infty} \\
        &\leqslant 2^{k} \biggl(~\underset{\leqslant 1}{\underbrace{\prod_{i = 1}^{k} \left\Vert E_{a_{i}} \right\Vert_{\infty}}}~\biggr) \underset{\leqslant 1}{\underbrace{\left\Vert P_{a} \right\Vert_{\infty}}} \leqslant 2^{k},
        \nonumber
    \end{align}
    and the inequality in the third line follows from Lemma \ref{lemma:cluster-counting-blmt}, which counts the number of monomials (or, equivalently, the number of clusters over which the sum is taken) contained in the summation form of $\left[ H, P \right]_{k}$. Now, consider
    \begin{align}
        \mathcal{F}_{a, t} &= \frac{1}{D} \operatorname{Tr} \left[ Q_{a} U P_{a} U^{\dagger} \right] \\
        \nonumber
        &= \frac{1}{D} \operatorname{Tr} \left[ Q_{a} \sum_{k = 0}^{\infty} \frac{\left( -it \right)^{k}}{k!} \left[ H, P_{a} \right]_{k} \right] \\
        \nonumber
        &= \frac{1}{D} \sum_{k = 0}^{\infty} \frac{\left( -it \right)^{k}}{k!} \operatorname{Tr} \left[ Q_{a} \left[ H, P_{a} \right]_{k} \right] \\
        \nonumber
        &= \frac{1}{D} \biggl( \underset{=0}{\underbrace{\operatorname{Tr} \left[ Q_{a} P_{a} \right]}} - it \operatorname{Tr} \left[ Q_{a} \left[ H, P_{a} \right] \right] \biggr) + \underset{\mathcal{O} \left( t^{2} \right)}{\underbrace{\frac{1}{D} \operatorname{Tr} \left[ Q_{a} \sum_{k = 2}^{\infty} \frac{\left( -it \right)^{k}}{k!} \left[ H, P_{a} \right]_{k} \right]}} \\
        \nonumber
        &= \frac{1}{D} \left( -it \operatorname{Tr} \left[ Q_{a} \left[ \sum_{b \in \left[ M \right]} \lambda_{b} E_{b}, P_{a} \right] \right] \right) + \mathcal{O} \left( t^{2} \right) \\
        \nonumber
        &= \frac{-it}{D} \sum_{b \in \left[ M \right]} \lambda_{b} \operatorname{Tr} \left[ Q_{a} \left[ E_{b}, P_{a} \right] \right] + \mathcal{O} \left( t^{2} \right) \\
        \nonumber
        &= \frac{-it}{D} \sum_{b \in \left[ M \right]} \lambda_{b} \operatorname{Tr} \left[ 2iP_{a}E_{a} \left( E_{b} P_{a} - P_{a} E_{b} \right) \right] + \mathcal{O} \left( t^{2} \right) \\
        \nonumber
        &= \frac{2t}{D} \sum_{b \in \left[ M \right]} \lambda_{b} \operatorname{Tr} \left[ P_{a} E_{a} \left( E_{b} P_{a} - P_{a} E_{b} \right) \right] + \mathcal{O} \left( t^{2} \right) \\
        \nonumber
        &= \frac{2t}{D} \sum_{b \in \left[ M \right]} \lambda_{b} \left( \operatorname{Tr} \left[ P_{a} E_{a} E_{b} P_{a} \right] - \operatorname{Tr} \left[ P_{a} E_{a} P_{a} E_{b} \right] \right) + \mathcal{O} \left( t^{2} \right) \\
        \nonumber
        &= \frac{2t}{D} \sum_{b \in \left[ M \right]} \lambda_{b} \left( \operatorname{Tr} \left[ P_{a}^{2} E_{a} E_{b} \right] + \operatorname{Tr} \left[ E_{a} P_{a}^{2} E_{b} \right] \right) + \mathcal{O} \left( t^{2} \right) \\
        \nonumber
        &= \frac{4t}{D} \sum_{b \in \left[ M \right]} \lambda_{b} \operatorname{Tr} \left[ E_{a} E_{b} \right] + \mathcal{O} \left( t^{2} \right) = \frac{4t}{D} \cdot \lambda_{a} \operatorname{Tr} \left[ E_{a}^{2} \right] + \mathcal{O} \left( t^{2} \right) = 4t\lambda_{a} + \mathcal{O} \left( t^{2} \right),
        \nonumber
    \end{align}
    where, in the fifth equality, we have simply plugged in the definition of $H$, in the seventh equality, we have used the definition of $Q_{a}$, and the tenth equality follows from the cyclicity of the trace and the fact that $P_{a} E_{a} = -E_{a} P_{a}$. From the above, we have 
    \begin{align}
        \left\vert \mathcal{F}_{a} - 4 t \lambda_{a} \right\vert &\leqslant \left\vert \frac{1}{D} \operatorname{Tr} \left[ Q_{a} \sum_{k = 2}^{\infty} \frac{\left( -it \right)^{k}}{k!} \left[ H, P_{a} \right]_{k} \right] \right\vert \\
        \nonumber
        &\leqslant \frac{1}{D} \sum_{k = 2}^{\infty} t^{k} \left\vert \operatorname{Tr} \left[ Q_{a} \underset{q_{k} \left( \lambda_{1}, \lambda_{2}, \ldots, \lambda_{M} \right)}{\underbrace{\frac{\left( -i \right)^{k}}{k!} \left[ H, P_{a} \right]_{k}}} \right] \right\vert \\
        \nonumber
        &= \frac{1}{D} \sum_{k = 2}^{\infty} t^{k} \left\vert \operatorname{Tr} \left[ Q_{a} q_{k} \left( \lambda_{1}, \lambda_{2}, \ldots, \lambda_{M} \right) \right] \right\vert \\
        \nonumber
        &\leqslant \frac{1}{D} \sum_{k = 2}^{\infty} t^{k} \underset{=2}{\underbrace{\left\Vert Q_{a} \right\Vert_{\infty}}} \left\Vert q_{k} \left( \lambda_{1}, \lambda_{2}, \ldots, \lambda_{M} \right) \right\Vert_{1} \\
        \nonumber
        &\leqslant \frac{2}{D} \sum_{k = 2}^{\infty} t^{k} \cdot D \left\Vert q_{k} \left( \lambda_{1}, \lambda_{2}, \ldots, \lambda_{M} \right) \right\Vert_{\infty} \\
        \nonumber
        &\leqslant 2\sum_{k = 2}^{\infty} \left( 2t \left( \mathfrak{d} + 1 \right) \right)^{k} \\
        \nonumber
        &= 2 \left( \frac{1}{1 - 2t \left( \mathfrak{d} + 1 \right)} - 1 - 2t \left( \mathfrak{d} + 1 \right) \right) \\
        \nonumber
        &= 2 \left( \frac{1 - \left( 1 - 2t \left( \mathfrak{d} + 1 \right) \right)\left( 1 + 2t \left( \mathfrak{d} + 1 \right) \right)}{1 - 2t \left( \mathfrak{d} + 1 \right)} \right) \\
        \nonumber
        &= 2 \left( \frac{1 - \left( 1 - \left( 2t \left( \mathfrak{d} + 1 \right) \right)^{2} \right)}{1 - 2t \left( \mathfrak{d} + 1 \right)} \right) = \frac{8t^{2} \left( \mathfrak{d} + 1 \right)^{2}}{1 - 2t \left( \mathfrak{d} + 1 \right)}.
        \nonumber
    \end{align}
    The second equality holds when $2t \left( \mathfrak{d} + 1 \right) < 1$. 
\end{proof}

\subsection{Hamiltonian learning precision required for the continuous-parameter case}
\label{subsec:hl-precision}
\begin{lemma}[Error propagation from Hamiltonian coefficients to time-evolved expectation values]
    \label{lem:error-prop-ham-coeff-exp-values}
    Let $H \left( \lambda \right) \coloneqq \sum_{a \in \left[ M \right]} \lambda_{a} E_{a}$ and $H \left( \lambda' \right) \coloneqq \sum_{a \in \left[ M \right]} \lambda_{a}' E_{a}$ be Hamiltonians defined on $n$ qubits, where $\lambda \coloneqq \left( \lambda_{a} \right)_{a \in \left[ M \right]}$ and $\lambda' \coloneqq \left( \lambda_{a}' \right)_{a \in \left[ M \right]}$. Also, let $\left\Vert \lambda - \lambda' \right\Vert_{\ell_{\infty}}\coloneqq \max_{a \in \left[ M \right]} \left\vert \lambda_{a} - \lambda_{a}' \right\vert \leqslant \varepsilon$, for some $\varepsilon > 0$. Define $\rho \left( \lambda \right) \coloneqq e^{-iH \left( \lambda \right) t} \rho_{0} e^{iH \left( \lambda \right) t}$ and $\rho \left( \lambda' \right) \coloneqq e^{-i H \left( \lambda' \right) t} \rho_{0} e^{i H \left( \lambda' \right) t}$, where $\rho_{0} \in \mathbb{C}^{2^{n} \times 2^{n}}$ is an arbitrary initial state. Then:
    \begin{align}
        \left\vert \langle O \rangle_{\rho \left( \lambda \right)} - \langle O \rangle_{\rho \left( \lambda' \right)} \right\vert &\leqslant 2 \left\Vert O \right\Vert_{\infty} M\left\vert t \right\vert \varepsilon, 
    \end{align}
    where $O \in \mathrm{Herm} \left( 2^{n} \right)$ is an observable.
\end{lemma}
\begin{proof}
    For notational convenience, we will denote $H \left( \lambda \right)$ and $H \left( \lambda' \right)$ by $H$ and $H'$ respectively, and $\rho \left( \lambda \right)$ and $\rho \left( \lambda' \right)$ by $\rho$ and $\rho'$ respectively. Recall now that, for some parametrized operator $A \left( t \right)$, by Duhamel's formula for derivatives of matrix exponentials~\cite{feldman2007}, we have
    \begin{align}
        \partial_{t} e^{ A \left( t \right)} = \int_{0}^{1} e^{s A \left( t \right)} \partial_{t} A \left( t \right) e^{\left( 1 - s \right) A \left( t \right)} ds.
    \end{align}
    Now, for $s \in \left[ 0, 1 \right]$, 
    define
    \begin{align}
        H \left( s \right) \coloneqq H' + s \left( H - H' \right),
    \end{align}
    so that $H \left( 0 \right) = H'$ and $H \left( 1 \right) = H$. Now, consider 
    \begin{align}
        e^{-iHt} - e^{-iH't} &= \int_{0}^{1} \partial_{s} e^{-iH \left( s \right) t} ds \\
        \nonumber
        &= -it \int_{0}^{1} \int_{0}^{1} e^{-i s' H \left( s \right) t} \partial_{s} H \left( s \right) e^{\left( 1 - s' \right) \left( -i H \left( s \right) t \right)} ds ds' \\
        \nonumber
        &= -it \int_{0}^{1} \int_{0}^{1} e^{-i s' H \left( s \right) t} \left( H - H' \right) e^{\left( 1 - s' \right) \left( -i H \left( s \right) t \right)} ds ds',
        \nonumber
    \end{align}
    where the first equality follows from the fundamental theorem of calculus, the second equality follows from an application of Duhamel's formula, and, in the third equality, we have made use of the fact that $\partial_{s} H \left( s \right) = H - H'$. Taking the norm of both sides, we get:
    \begin{align}
        \left\Vert e^{-iHt} - e^{-iH't} \right\Vert_{\infty} &\leqslant \left\vert t \right\vert \int_{0}^{1} \int_{0}^{1} \underset{=1}{\underbrace{\left\Vert e^{-i s' H \left( s \right) t} \right\Vert_{\infty}}} \left\Vert H - H' \right\Vert_{\infty} \underset{=1}{\underbrace{\left\Vert e^{\left( 1 - s' \right) \left( -i H \left( s \right) t \right)} \right\Vert_{\infty}}} ds ds' \\
        &\leqslant \left\vert t \right\vert \left\Vert H - H' \right\Vert_{\infty}.
    \end{align}
    Let $U \coloneqq e^{-iHt}$ and $U' \coloneqq e^{-iH't}$. We bound the trace norm between $\rho$ and $\rho'$ as 
    \begin{align}
        \left\Vert \rho - \rho' \right\Vert_{1} &= \left\Vert U \rho_{0} U^{\dagger} - U'\rho_{0} U'^{\dagger}  \right\Vert_{1} \\
        \nonumber
        &= \left\Vert U \rho_{0} U^{\dagger} + U'\rho_{0}U^{\dagger} - U' \rho_{0} U^{\dagger} - U'\rho_{0} U'^{\dagger}  \right\Vert_{1} \\
        \nonumber
        &\leqslant \left\Vert \left( U - U' \right) \rho_{0} U^{\dagger} \right\Vert_{1} + \left\Vert U' \rho_{0} \left( U^{\dagger} - U'^{\dagger} \right) \right\Vert_{1} \\
        \nonumber
        &\leqslant \left\Vert U - U' \right\Vert_{\infty} + \left\Vert U^{\dagger} - U'^{\dagger} \right\Vert_{\infty} \\
        \nonumber
        &\leqslant 2 \left\vert t \right\vert \left\Vert H - H' \right\Vert_{\infty} \leqslant 2 t \varepsilon \sum_{a \in \left[ M  \right]} \left\Vert E_{a} \right\Vert_{\infty} = 2 t M \varepsilon.
        \nonumber
    \end{align}
    Finally, for some observable $O$, we find
    \begin{align}
        \left\vert \langle O \rangle_{\rho} - \langle O \rangle_{\rho'} \right\vert \leqslant \left\Vert O \right\Vert_{\infty} \left\Vert \rho - \rho' \right\Vert_{1} \leqslant 2 \left\Vert O \right\Vert_{\infty} t M \varepsilon, 
    \end{align}
    which is exactly the advertised bound.
\end{proof} 

\subsection{Sample and time complexity analysis for quantum learning algorithm}
\label{subsec:sample-time-complexity-analysis}
\begin{lemma}[Bounds on time-averaged integral]
    \label{lem:bounds-time-averaged-integral}
    Let $\mathcal{F}_{a, t}$ be as defined in Eq.\  (\ref{eq:mathcal_f_a_t}), Lemma \ref{lem:tail-bound-higher-order}. For $t^{*} > 0$, define the 
    time-averaged integral
    \begin{align}
        \tilde{\mathcal{F}}_{a, t^{*}} \coloneqq \frac{1}{t^{*}} \int_{0}^{t^{*}} \mathcal{F}_{a, t} dt.
    \end{align}
    Then, for $\varepsilon > 0$, we find
    \begin{align}
        \left\vert \frac{\tilde{\mathcal{F}}_{a, t^{*}}}{2t^{*}} - \lambda_{a} \right\vert \leqslant \frac{2t^{*} \left( \mathfrak{d} + 1 \right)^{2}}{1 - 2 \left( \mathfrak{d} + 1 \right) t^{*}},
    \end{align}
    where $\mathfrak{d}$ is the maximum degree of the dual-interaction graph $\mathfrak{G}$ of $H$.
\end{lemma}

\begin{proof}
    Consider that
    \begin{align}
        \tilde{\mathcal{F}}_{a, t^{*}} &\coloneqq \frac{1}{t^{*}} \int_{0}^{t^{*}} \mathcal{F}_{a, t}~dt \\
        \nonumber 
        &\leqslant \frac{1}{t^{*}} \int_{0}^{t^{*}} 4 \lambda_{a} t + \frac{8t^{2} \left( \mathfrak{d} + 1 \right)^{2}}{1 - 2t \left( \mathfrak{d} + 1 \right)} dt \\
         \nonumber 
        &= 2 \lambda_{a} t^{*} + \frac{2}{t^{*}} \int_{0}^{t^{*}} \frac{1 - \left( 1 - 4t^{2} \left( \mathfrak{d} + 1 \right)^{2} \right)}{1 - 2t \left( \mathfrak{d} + 1 \right)} dt \\
         \nonumber 
        &= 2\lambda_{a}t^{*} + \frac{2}{t^{*}} \int_{0}^{t^{*}} \frac{1 - \left( 1 - 2t \left( \mathfrak{d} + 1 \right) \right) \left( 1 + 2t \left( \mathfrak{d} + 1 \right) \right)}{1 - 2t \left( \mathfrak{d} + 1 \right)} dt \\
         \nonumber 
        &= 2 \lambda_{a} t^{*} + \frac{2}{t^{*}} \int_{0}^{t^{*}} \frac{1}{1 - 2t \left( \mathfrak{d} + 1 \right)} dt - \frac{2}{t^{*}} \int_{0}^{t^{*}} 1 + 2t \left( \mathfrak{d} + 1 \right) dt \\
         \nonumber 
        &= 2 \lambda_{a} t^{*} - \frac{2}{t^{*}} \cdot \frac{1}{2 \left( \mathfrak{d} + 1 \right)} \int_{1}^{1 - 2t^{*} \left( \mathfrak{d} + 1 \right)} \frac{1}{u}~du - \frac{2}{t^{*}} \left( t^{*} + 2 \left( \mathfrak{d} + 1 \right) \cdot \frac{{t^{*}}^{2}}{2} \right) \\
         \nonumber 
        &= 2 \lambda_{a} t^{*} - \frac{\ln \left\vert 1 - 2t^{*} \left( \mathfrak{d} + 1 \right) \right\vert}{t^{*} \left( \mathfrak{d} + 1 \right)} - 2 - 2\left( \mathfrak{d} + 1 \right) t^{*}.
         \nonumber 
\end{align}
The above implies
\begin{align}
    \left\vert \frac{\tilde{\mathcal{F}}_{a, t}}{2t^{*}} - \lambda_{a} \right\vert \leqslant \frac{1}{2t^{*}} \left( - \frac{\ln \left\vert 1 - 2t^{*} \left( \mathfrak{d} + 1 \right) \right\vert}{t^{*} \left( \mathfrak{d} + 1 \right)} - 2 - 2 \left( \mathfrak{d} + 1 \right) t^{*} \right).
\end{align}
For ease of symbolic manipulation, let us define $a \leftarrow \left( \mathfrak{d} + 1 \right)$ and $x \leftarrow 2at^{*}$. Now, consider that
\begin{align}
    \ln \left( 1 - x \right) = -\sum_{n = 1}^{\infty} \frac{x^{n}}{n} \implies -\frac{\ln \left( 1 - x \right)}{x} - 1 - \frac{x}{2} = \sum_{n = 2}^{\infty} \frac{x^{n}}{n + 1} \leqslant \frac{1}{2} \sum_{n = 2}^{\infty} x^{n} = \frac{x^{2}}{2 \left( 1 - x \right)}.  
\end{align}
From the above, we have
\begin{align}
    \frac{2a}{x} \left( -\frac{\ln \left( 1 - x \right)}{x} - 1 - \frac{x}{2} \right) \leqslant \frac{ax}{\left( 1 - x \right)} = \frac{\left( \mathfrak{d} + 1 \right) \cdot 2 a t^{*}}{1 - 2 a t^{*}} = \frac{2t^{*} \left( \mathfrak{d} + 1 \right)^{2}}{1 - 2 \left( \mathfrak{d} + 1 \right) t^{*}}.
\end{align}
That is,
\begin{align}
    \label{eq:diff-mc-f-coeff}
    \left\vert \frac{\tilde{\mathcal{F}}_{a, t^{*}}}{2t^{*}} - \lambda_{a} \right\vert \leqslant \frac{2t^{*} \left( \mathfrak{d} + 1 \right)^{2}}{1 - 2 \left( \mathfrak{d} + 1 \right) t^{*}}.
\end{align}
This completes the proof.
\end{proof}

\begin{theorem}[Hoeffding's inequality]
    \label{thm:hoeffding-inequality}
    Let $Y_{1}, Y_{2}, \ldots, Y_{N}$ be $N$ independent samples of some random variable $Y$ that is bounded in the interval $\left[ a, b \right]$ and with mean $\mu \coloneqq \mathbb{E} \left[ Y \right]$. Let $\overline{Y}_{N}$ be the sample mean
    \begin{align}
        \overline{Y}_{N} \coloneqq \frac{1}{N} \sum_{i = 1}^{N} Y_{i}.
    \end{align}
    Then, for all precision and confidence parameters $\varepsilon, \delta > 0$, 
    \begin{align}
        \Pr \left[ \left\vert \overline{Y}_{N} - \mu \right\vert \leqslant \varepsilon \right] \geqslant 1 - \delta
    \end{align}
    holds true,
    provided that the number of samples $T \geqslant \frac{M^{2}}{2 \varepsilon^{2}} \ln \left( \frac{2}{\delta} \right)$, where $M \coloneqq b - a$.
\end{theorem}

\begin{lemma}[Short-time sample complexity analysis]
    \label{lem:short-time-sample-complexity-analysis}
    Let
    \begin{align}
        B_{a}^{(j)} = \operatorname{Tr} \left[ Q_{a} e^{-iHt^{(j)}} \left( \bigotimes_{i = 1}^{n} \bigl\lvert s_{i}^{(j)} \bigr\rangle \bigl\langle s_{i}^{(j)} \bigr\rvert \right) e^{iHt^{(j)}} \right]
    \end{align}
    be the $a^{\mathrm{th}}$ component of the $j^{\mathrm{th}}$ label in the training dataset. Furthermore, for each $a \in \left[ M \right]$, let 
    \begin{align}
        C_{a}^{(j)} \coloneqq \mathbbm{1} \left[ P_{a} \bigl\lvert s_{i}^{(j)} \bigr\rangle = + \bigl\lvert s_{i}^{(j)} \bigr\rangle \right] B_{a}^{(j)}
    \end{align}
    and let $\hat{\mathcal{F}}_{a, t^{*}} \coloneqq \frac{1}{N_{t^{*}}} \sum_{j = 1}^{N_{t^{*}}} 6C_{a}^{(j)}$ be an empirical average over training samples whose times belong to the short-time interval $[0, t^{*}]$, where we use $N_{t^{*}}$ to denote the number of such training samples. Let $\hat{\lambda} \coloneqq \bigl( \hat{\lambda}_{a} \bigr)_{a \in \left[ M \right]} \coloneqq \left( \hat{\mathcal{F}}_{a, t^{*}}/2t^{*} \right)_{a \in \left[ m \right]}$ be the vector of coefficients that are obtained as the output of Algorithm \ref{alg:training-via-ham-learning-short-times}. Then, for a desired precision $\varepsilon > 0$, some confidence parameter $\delta > 0$, and a given observable $O$, we have $\left\vert \langle O \rangle_{\rho \left( \lambda, t \right)} - \langle O \rangle_{\rho \bigl( \hat{\lambda}, t \bigr)} \right\vert \leqslant \varepsilon$ provided that
    \begin{align}
        t^{*} &\in \mathcal{O} \left(  \frac{\varepsilon}{M T \left\Vert O \right\Vert_{\infty} \left( \mathfrak{d} + 1 \right)^{2}} \right), \\
        \text{and}\quad N_{t^{*}} &\in \mathcal{O} \left( \frac{M^{4} T^{4} \left\Vert O \right\Vert_{\infty}^{2} \left( \mathfrak{d} + 1 \right)^{4}}{\varepsilon^{4}} \log \left( \frac{M}{\delta} \right) \right),
    \end{align}
    with probability at least $1 - \delta$.
\end{lemma}
\begin{proof}
    Consider the error decomposition
    \begin{align}
        \left\vert \hat{\lambda}_{a} - \lambda_{a} \right\vert &= \left\vert \frac{\hat{\mathcal{F}}_{a, t^{*}}}{2t^{*}} - \lambda_{a} \right\vert = \left\vert \frac{\hat{\mathcal{F}}_{a, t^{*}}}{2t^{*}} - \frac{\tilde{\mathcal{F}}_{a, t^{*}}}{2t^{*}} + \frac{\tilde{\mathcal{F}}_{a, t^{*}}}{2t^{*}} - \lambda_{a}  \right\vert \leqslant \left\vert \frac{\hat{\mathcal{F}}_{a, t^{*}}}{2t^{*}} - \frac{\tilde{\mathcal{F}}_{a, t^{*}}}{2t^{*}} \right\vert + \left\vert \frac{\tilde{\mathcal{F}}_{a, t^{*}}}{2t^{*}} - \lambda_{a} \right\vert.
    \end{align}
    Note that we have already bounded the second term by means of Lemma \ref{lem:bounds-time-averaged-integral}. We bound the first term as 
    \begin{align}
        \left\vert \frac{\hat{\mathcal{F}}_{a, t^{*}}}{2t^{*}} - \frac{\tilde{\mathcal{F}}_{a, t^{*}}}{2t^{*}} \right\vert &\leqslant \frac{1}{2t^{*}} \left\vert \frac{1}{N_{t^{*}}} \sum_{j = 1}^{N_{t^{*}}} 6C_{a}^{(j)} - \frac{1}{t^{*}}\int_{0}^{t^{*}} \mathcal{F}_{a, t} dt \right\vert \\
        \nonumber
        &= \frac{1}{2t^{*}} \left\vert \frac{1}{N_{t^{*}}} \sum_{j = 1}^{N_{t^{*}}} 6C_{a}^{(j)} - \underset{s, t}{\mathbb{E}} \left[ 6C_{a} \left( s, t \right) \bigl\vert t \in \left[ 0, t^{*} \right] \right] \right\vert,
         \nonumber
    \end{align}
    where $6C_{a}^{(j)}$ are i.i.d.\ samples of the random variable $C_{a} \left( s, t \right)$, which is defined as 
    \begin{align}
        C_{a} \left( s, t \right) \coloneqq \begin{cases}
        B_{a} \left( s, t \right),\quad&\text{if $P_{a} \lvert s_{i} \rangle = +\lvert s_{i} \rangle$, $\operatorname{supp} \left( P_{a} \right) = \left\{ i \right\}$,} \\
        0, \quad&\text{otherwise.}
        \end{cases} \\
        \text{where} \qquad B_{a} \left( s, t \right) = \operatorname{Tr} \left[ Q_{a} e^{-iHt} \left( \bigotimes_{i = 1}^{n} \bigl\lvert s_{i} \bigr\rangle \bigl\langle s_{i} \bigr\rvert \right) e^{iHt} \right],
    \end{align}
    with $\left( s, t \right) \overset{\mathrm{i.i.d.}}{\sim} \mathsf{Unif} \left( \left\{ 0, 1, +, -, i+, i- \right\}^{n} \times \left[ 0, T \right] \right)$, where $T \in \mathcal{O} \left( \mathsf{poly} \left( n \right) \right)$. We are able to replace the time-averaged integral by the expectation value of $6C_{a} \left( s, t \right)$ conditioned on $t$ belonging to the short-time window $\left[ 0, t^{*} \right]$, where the expectation is taken over $s$ and $t$, because of the following: 
    \begin{align}
        &\underset{s, t}{\mathbb{E}} \left[ 6C_{a} \left( s, t \right)~\bigl\vert~t \in \left[ 0, t^{*} \right]  \right] \\ 
        &= 
        \underset{t}{\mathbb{E}} \left[ \underset{s}{\mathbb{E}} \left[ 6 C_{a} \left( s, t \right) \bigl\vert t \right] \bigl\vert~t \in \left[ 0, t^{*} \right] \right] \\
        \nonumber
        &= \underset{t}{\mathbb{E}} \left[ 6~\underset{s}{\mathbb{E}} \left[ \mathbbm{1}\left[ P_{a} \lvert s_{i} \rangle = + \lvert s_{i} \rangle \right] B_{a} \left( s, t \right) \lvert t \right]~\lvert~t \in\left[ 0, t^{*} \right] \right] \\
        \nonumber
        &= \underset{t}{\mathbb{E}} \left[ 6 \cdot \Pr \left[ P_{a} \lvert s_{i} \rangle = + \lvert s_{i} \rangle \bigl\vert t \right] \underset{s}{\mathbb{E}} \left[ B_{a} \left( s, t \right) \lvert P_{a} \lvert s_{i} \rangle = + \lvert s_{i} \rangle, t \right] \lvert~t \in \left[ 0, t^{*} \right] \right] \\
        \nonumber
        &= \underset{t}{\mathbb{E}} \left[ 6 \cdot \frac{1}{6} \cdot \underset{s}{\mathbb{E}} \left[ \operatorname{Tr} \left[ Q_{a} e^{-iHt} \left( \bigotimes_{i = 1}^{n} \lvert s_{i} \rangle \langle s_{i} \rvert \right) e^{iHt} \right] \biggl\vert P_{a} \lvert s_{i} \rangle = + \lvert s_{i} \rangle, t \right] \biggl\lvert t \in \left[ 0, t^{*} \right] \right] \\
        \nonumber
        &= \underset{t}{\mathbb{E}} \left[ \operatorname{Tr} \left[ Q_{a} e^{-iHt}~\underset{s}{\mathbb{E}} \left[ \bigotimes_{i = 1}^{n} \lvert s_{i} \rangle \langle s_{i} \rvert~\biggl\vert~P_{a} \lvert s_{i} \rangle = + \lvert s_{i} \rangle, t \right] e^{iHt} \right] \biggl\vert~t \in \left[ 0, t^{*} \right] \right] \\
        \nonumber
        &= \underset{t}{\mathbb{E}} \left[ \operatorname{Tr} \left[ Q_{a} e^{-iHt} \biggl\{ \frac{I}{2} \otimes \ldots \otimes \underset{\text{$i^{\text{th}}$ qubit}}{\underbrace{\left( \frac{I + P_{a}}{2} \right)}} \otimes \ldots \otimes \frac{I}{2} \biggr\} e^{iHt} \right] \biggl\vert~t \in \left[ 0, t^{*} \right]  \right] \\
        \nonumber
        &= \underset{t}{\mathbb{E}} \left[ \frac{1}{2^{n}}~\underset{= 0}{\underbrace{\operatorname{Tr} \left[ Q_{a} e^{-iHt} e^{iHt}  \right]}} +  \operatorname{Tr} \left[ Q_{a} e^{-iHt} \left\{ \frac{I}{2} \otimes \frac{I}{2} \otimes \ldots \otimes \frac{P_{a}}{2} \otimes \ldots \otimes \frac{I}{2} \right\} e^{iHt} \right] \biggl\vert~t \in \left[ 0, t^{*} \right] \right] \\
        \nonumber
        &= \underset{t}{\mathbb{E}} \left[ \frac{1}{2^{n}} \operatorname{Tr} \left[ Q_{a} e^{-iHt} P_{a} e^{iHt} \right] \bigl\vert~t \in \left[ 0, t^{*} \right] \right] = \underset{t}{\mathbb{E}} \left[ \mathcal{F}_{a, t} \bigl\vert t \in \left[ 0, t^{*} \right] \right] = \frac{1}{t^{*}} \int_{0}^{t^{*}} \mathcal{F}_{a, t} dt,
        \nonumber
    \end{align}
    where the third equality follows from the elementary conditional probability identity $\mathbb{E} \left[ \mathbbm{1}_{A} X \right] = \Pr \left[ A \right] \mathbb{E} \left[ X \vert A \right]$, where $\mathbbm{1}_{A}$ is equal to $1$ if the event $A$ occurs, and $0$ otherwise. Going from the sixth to the seventh equality, we have used the traceless property of $Q_{a}$, and, per the convention of Ref.~\cite{Haah_2024}, replaced the $n$-qubit operator $I \otimes I \otimes \ldots \otimes P_{a} \otimes \ldots I$ with simply $P_{a}$ to reduce notational clutter.
    By Theorem \ref{thm:hoeffding-inequality}, for some failure probability $\eta > 0$, we have that
    \begin{align}
        \Pr \left[ \frac{1}{2t^{*}} \left\vert \hat{\mathcal{F}}_{a, t^{*}} - \tilde{\mathcal{F}}_{a, t^{*}} \right\vert \leqslant \frac{1}{2t^{*}} \sqrt{\frac{288}{N_{t^{*}}}\log \left( \frac{2}{\eta} \right)}~\right] \geqslant 1 - \eta.
    \end{align}
    Now, setting $\eta \leftarrow \delta/2M$ for some desired confidence parameter $\delta > 0$ and applying the union bound over all $a \in \left[ M \right]$, we find that 
    \begin{align}
        \left\vert \frac{\hat{\mathcal{F}}_{a, t^{*}}}{2t^{*}} - \frac{\tilde{\mathcal{F}}_{a, t^{*}}}{2t^{*}} \right\vert + \left\vert \frac{\tilde{\mathcal{F}}_{a, t^{*}}}{2t^{*}} - \lambda_{a} \right\vert \leqslant \frac{1}{2t^{*}} \sqrt{\frac{288}{N_{t^{*}}}\log \left( \frac{4M}{\delta} \right)} + \frac{2t^{*} \left( \mathfrak{d} + 1 \right)^{2}}{1 - 2 \left( \mathfrak{d} + 1 \right) t^{*}}
    \end{align}
    holds for all $a \in \left[ M \right]$,
    with probability at least $1 - \delta/2$, where we have allocated the failure probability budget $\delta/2$ to the training stage. We note here in advance that the other half of the final/inference failure probability $\delta$ will be allocated to the inference step of our algorithm, where we estimate expectation values of numerous observables via the classical shadows formalism. This will ensure that the final prediction outputs will be $\varepsilon$-accurate with probability $1 - \delta$. In the above, we have also plugged in the bound from Eq.\  (\ref{eq:diff-mc-f-coeff}). Now, we allocate the  error budgets
    \begin{align}
        &\frac{1}{2t^{*}} \sqrt{\frac{288}{N_{t^{*}}}\log \left( \frac{4M}{\delta} \right)} \leqslant \frac{\varepsilon'}{2}, \\
        \text{and}\qquad\qquad &\frac{2t^{*} \left( \mathfrak{d} + 1 \right)^{2}}{1 - 2 \left( \mathfrak{d} + 1 \right) t^{*}} \leqslant \frac{\varepsilon'}{2}.
    \end{align}
    Now, in order to ensure that $\left\vert \langle O \rangle_{\rho \left( \lambda, t \right)} - \langle O \rangle_{\rho \left( \lambda', t \right)} \right\vert \leqslant \varepsilon$, in accordance with Lemma \ref{lem:error-prop-ham-coeff-exp-values}, we will set $\varepsilon' \leftarrow \varepsilon/2MT \left\Vert O \right\Vert_{\infty}$. Then, we have the solution 
    \begin{align}
        &\frac{2t^{*} \left( \mathfrak{d} + 1 \right)^{2}}{1 - 2 \left( \mathfrak{d} + 1 \right) t^{*}} \leqslant \frac{\varepsilon'}{2} \implies 4 \left( \mathfrak{d} + 1 \right)^{2} t^{*} \leqslant \varepsilon' - 2 \left( \mathfrak{d} + 1 \right) \varepsilon' t^{*} \implies \left\{ 4 \left( \mathfrak{d} + 1 \right)^{2} + 2 \left( \mathfrak{d} + 1 \right) \varepsilon' \right\} t^{*} \leqslant \varepsilon' \\
        \nonumber
        \implies &t^{*} \leqslant \frac{\varepsilon'}{4 \left( \mathfrak{d} + 1 \right)^{2} + 2 \left( \mathfrak{d} + 1 \right)\varepsilon'} \leqslant \frac{\varepsilon/2 M T \left\Vert O \right\Vert_{\infty}}{4 \left( \mathfrak{d} + 1 \right)^{2} + 2 \left( \mathfrak{d} + 1 \right) \left( \varepsilon / 2 M T \left\Vert O \right\Vert_{\infty} \right)} = \frac{\varepsilon}{8 M T \left\Vert O \right\Vert_{\infty} \left( \mathfrak{d} + 1 \right)^{2} + 2 \left( \mathfrak{d} + 1 \right) \varepsilon} \\
         \nonumber
        \implies &t^{*} \in \mathcal{O} \left( \frac{\varepsilon}{M T \left\Vert O \right\Vert_{\infty} \left( \mathfrak{d} + 1 \right)^{2}} \right)
    \end{align}
    for $t^{*}$.
    Now, we will solve for the sample complexity $N_{t^{*}}$, to get
    \begin{align}
        \frac{1}{2t^{*}} \sqrt{\frac{288}{N_{t^{*}}} \log \left( \frac{4M}{\delta} \right)} \leqslant \frac{\varepsilon'}{2} \implies \frac{1}{4{t^{*}}^{2}} \cdot \frac{288}{N_{t^{*}}} \log \left( \frac{4M}{\delta} \right) \leqslant \frac{{\varepsilon'}^{2}}{4} \implies N_{t^{*}} \geqslant \frac{288}{{t^{*}}^{2} {\varepsilon'}^{2}} \log \left( \frac{4M}{\delta} \right).
    \end{align}
    Plugging in the solution obtained for $t^{*}$, as well as our chosen bound on $\varepsilon'$, we get
    \begin{align}
        N_{t^{*}} &\geqslant \frac{288}{\left\{ \frac{\varepsilon}{8MT \left\Vert O \right\Vert_{\infty} \left( \mathfrak{d} + 1 \right)^{2} + 2 \left( \mathfrak{d} + 1 \right) \varepsilon} \right\}^{2} \cdot \left( \frac{\varepsilon}{2 M T \left\Vert O \right\Vert_{\infty}} \right)^{2}} \log \left( \frac{4M}{\delta} \right) \\
         \nonumber
        &= \frac{288 \cdot 8 M^{2} T^{2} \left\Vert O \right\Vert_{\infty}^{2} \cdot \left\{ 16 M^{2} T^{2} \left\Vert O \right\Vert_{\infty}^{2} \left( \mathfrak{d} + 1 \right)^{4} + 16 MT \left\Vert O \right\Vert_{\infty} \left( \mathfrak{d} + 1 \right)^{3} \varepsilon + 4 \left( \mathfrak{d} + 1 \right)^{2} \varepsilon^{2} \right\}}{\varepsilon^{4}} \log \left( \frac{4M}{\delta} \right) \\
         \nonumber
        &= \frac{9216 M^{2} T^{2} \left\{ 2MT \left\Vert O \right\Vert_{\infty} \left( \mathfrak{d} + 1 \right)^{2} + \left( \mathfrak{d} + 1 \right) \varepsilon \right\}^{2}}{\varepsilon^{4}} \log \left( \frac{4M}{\delta} \right) \in \mathcal{O} \left( \frac{M^{4} T^{4} \left\Vert O \right\Vert_{\infty}^{2} \left( \mathfrak{d} + 1 \right)^{4}}{\varepsilon^{4}} \log \left( \frac{M}{\delta} \right) \right).
         \nonumber
    \end{align}
\end{proof}

\begin{lemma}[Total sample complexity analysis]
    \label{lem:tot-sample-complexity-analysis}
    Let $N$ be the total sample complexity of the training dataset:
    \begin{align}
        \mathcal{T} \coloneqq \left\{ \left( s^{(j)}, t^{(j)}, \left( B_{a}^{(j)} \right)_{a \in \left[ m \right]} \right) \right\}_{j = 1}^{N},
    \end{align}
    which is given as input to Algorithm \ref{alg:training-via-ham-learning-short-times}. Also, let $N_{t^{*}}$ be the short-time sample complexity, as in Lemma \ref{lem:short-time-sample-complexity-analysis}. Then, for Lemma \ref{lem:short-time-sample-complexity-analysis} to hold, $N$ must scale as 
    \begin{align}
        N \in \mathcal{O} \left( \frac{M^{5} T^{6} \left\Vert O \right\Vert_{\infty}^{3} \left( \mathfrak{d} + 1 \right)^{6}}{\varepsilon^{5}} \log \left( \frac{M}{\delta} \right) \right).
    \end{align}
\end{lemma}

\begin{proof}
    For each $j \in \left[ N \right]$, we define
    \begin{align}
        X^{(j)} \coloneqq \mathbbm{1} \left[ t^{(j)} \in \left[ 0, t^{*} \right] \right] \mathbbm{1} \left[ x_{1}^{(j)} = 0 \right] = \begin{cases}
            1,\quad\text{if $t^{(j)} \in \left[ 0, t^{*} \right]$ and $x_{1}^{(j)} = 0$}, \\
            0,\quad\text{otherwise,}
        \end{cases}
    \end{align}
    so that $X^{(j)} \in \left\{ 0, 1 \right\}$ is a Bernoulli random variable, since $\Pr \left[ X^{(j)} = 1 \right] = t^{*}/2T$ and $\Pr \left[ X^{(j)} = 0 \right] = 1 - t^{*}/2T$. The number of 
    short-time samples is then given as
    \begin{align}
        &N_{t^{*}} \coloneqq \sum_{j = 1}^{N} X^{(j)}, \\
        \text{so that
        }\qquad \mu &\coloneqq \mathbb{E} \left[ N_{t^{*}} \right] = \sum_{j = 1}^{N}  \mathbb{E} \left[ X^{(j)} \right] = \frac{Nt^{*}}{2T}. 
    \end{align}
    so that $N_{t^{*}} \sim B \left( N, t^{*}/2T \right)$ is itself a binomial random variable, thereby allowing us to use the Chernoff bound. For $0 \leqslant \eta \leqslant 1$, it holds that
    \begin{align}
        \label{eq:total-sample-complexity-chernoff-step}
        \Pr \left[ N_{t^{*}} \leqslant \left( 1 - \eta \right) \mu \right] \leqslant \exp \left( - \frac{\eta^{2} \mu}{2} \right) \implies \Pr \left[ N_{t^{*}} \leqslant \frac{\left( 1 - \eta \right) Nt^{*}}{2T} \right] \leqslant \exp \left( - \frac{\eta^{2} N t^{*}}{4T} \right).
    \end{align}
    As a notational shorthand, we will set 
    \begin{align}
        N_{\text{short}} \coloneqq \frac{2304 M^{2} T^{2} \left\{ 2MT \left\Vert O \right\Vert_{\infty} \left( \mathfrak{d} + 1 \right)^{2} + \left( \mathfrak{d} + 1 \right) \varepsilon \right\}^{2}}{\varepsilon^{4}} \log \left( \frac{4M}{\delta} \right)
    \end{align}
    to be the required number of short-time samples calculated in Lemma \ref{lem:short-time-sample-complexity-analysis}.
    Additionally, we will also require that
    \begin{align}
        \frac{\left( 1 - \eta \right) N t^{*}}{2T} = N_{\text{short}} \implies \eta = 1 - \frac{2N_{\text{short}}T}{Nt^{*}}.
    \end{align}
    Next, in the implication of Eq.\ (\ref{eq:total-sample-complexity-chernoff-step}), we will set the R.H.S. to be less than or equal to some $\delta_{\text{short}}$, and then solve for $N$. We proceed as follows.
    \begin{align}
        \exp \left[ - \frac{N t^{*}}{4T} \left( 1 - \frac{2N_{\text{short}} T}{N t^{*}} \right)^{2} \right] \leqslant \delta_{\text{short}} \implies &\exp \left[ -\frac{1}{2} \left\{ \frac{Nt^{*}}{2T} - 2N_{\text{short}} + \frac{2N_{\text{short}}^{2} T}{N t^{*}} \right\} \right] \leqslant \delta_{\text{short}} \\\implies &\frac{Nt^{*}}{2T} - 2N_{\text{short}} + \frac{2N_{\text{short}}^{2} T}{N t^{*}} \geqslant 2\log \left( \frac{1}{\delta_{\text{short}}} \right).
        \nonumber
    \end{align}
    The above may be re-written as the quadratic inequality in $Nt^{*}/2T$
    \begin{align}
        \left( \frac{Nt^{*}}{2T} \right)^{2} - \left( 2N_{\text{short}} + 2 \log \left( \frac{1}{\delta_{\text{short}}} \right) \right) \frac{Nt^{*}}{2T} + N_{\text{short}}^{2} \geqslant 0. 
    \end{align}
    After application of the quadratic formula and rearrangement, we will find that 
    \begin{align}
        \label{eq:tot-sample-complexity-quadratic-soln}
        N &\geqslant \frac{2T}{t^{*}} \left[ N_{\text{short}} + \log \left( \frac{1}{\delta_{\text{short}}} \right) + \sqrt{ \log \left( \frac{1}{\delta_{\text{short}}} \right) \left( 2 N_{\text{short}} + \log \left( \frac{1}{\delta_{\text{short}}} \right)  \right) } ~\right] \\
        &\in \mathcal{O} \left( \frac{T}{t^{*}} N_{\text{short}} \right) \subseteq \mathcal{O} \left( \frac{M^{5} T^{6} \left\Vert O \right\Vert_{\infty}^{3} \left( \mathfrak{d} + 1 \right)^{6}}{\varepsilon^{5}} \log \left( \frac{M}{\delta} \right) \right),
        \nonumber
    \end{align}
    holds true,
    where, in the final step, we have plugged in the previously obtained scalings for $N_{\text{short}}$ and $t^{*}$ in Lemma \ref{lem:short-time-sample-complexity-analysis}.
\end{proof}

\begin{lemma}[Sample complexity of inference via classical shadows]
    \label{lem:sample-complexity-inference-classical-shadows}
    Let $\bigl( \hat{\lambda}_{a} \bigr)_{a \in \left[ M \right]}$ be the learned Hamiltonian coefficients obtained as outputs of Algorithm \ref{alg:training-via-ham-learning-short-times} such that we have, for all $a \in \left[ M \right]$ and all $\left( Q_{a} \right)_{a \in \left[ M \right]}$, $\bigl\vert \hat{\lambda}_{a} - \lambda_{a} \bigr\vert \leqslant \varepsilon/2 M T~\underset{j \in \left[ P \right]}{\max} \left\Vert O_{j} \right\Vert_{\infty}$ with probability $1 - \delta/2$, where $\varepsilon, \delta > 0$ are the final precision and confidence parameters associated with the inference stage. Then, for a newly provided input $\bigl( s, t \bigr) \sim \mathsf{Unif} \left( \left\{ 0, 1, +, -, i+, i- \right\}^{n} \times \left[ 0, T \right] \right)$, in order to ensure $\bigl\vert \langle \hat{Q}_{a} \rangle_{\tilde{\rho} \bigl( s, t \bigr)} - \langle Q_{a} \rangle_{\rho \left( s, t \right)} \bigr\vert \leqslant \varepsilon$, where
    \begin{align}
        \quad \hat{\rho} \bigl( s, t \bigr) \coloneqq e^{-i \hat{H} t} \left( \bigotimes_{i = 1}^{n} \lvert s_{i} \rangle \langle s_{i} \rvert \right) e^{i \hat{H} t},~\hat{H} &\coloneqq \sum_{a \in \left[ M \right]} \hat{\lambda}_{a} E_{a}, \\
        \quad \rho \left( s, t \right) \coloneqq e^{-iHt} \left( \bigotimes_{i = 1}^{n} \lvert s_{i} \rangle \langle s_{i} \rvert \right) e^{iHt},~H &= \sum_{a \in \left[ M \right]} \lambda_{a} E_{a}, \\
        \text{and}\hspace{3cm} \left\Vert \tilde{\rho} \left( s, t \right) - \hat{\rho} \left( s, t \right) \right\Vert_{1} \leqslant \varepsilon_{\text{sim}},
    \end{align}
    with $\varepsilon_{\text{sim}} > 0$ being some state preparation error, we require $\tilde{N} \in \mathcal{O} \left( \frac{1}{\varepsilon^{2}} \log \left( \frac{M}{\delta} \right) \right)$ copies of $\hat{\rho} \bigl( s, t \bigr)$ to be prepared in Algorithm \ref{alg:inference-via-hamiltonian-simulation}.
\end{lemma}
\begin{proof}
    The proof follows straightforwardly from the classical shadows formalism. Consider the  error decomposition
    \begin{align}
        \left\vert \langle \hat{Q}_{a} \rangle_{\tilde{\rho} \bigl( s, t \bigr)} - \langle Q_{a} \rangle_{\rho \left( s, t \right)} \right\vert &= \left\vert \left( \langle \hat{Q}_{a} \rangle_{\tilde{\rho} \left( s, t \right)} - \langle Q_{a} \rangle_{\tilde{\rho} \left( s, t \right)} \right) + \left( \langle Q_{a} \rangle_{\tilde{\rho} \bigl( s, t \bigr)} - \langle Q_{a} \rangle_{\hat{\rho} \bigl( s, t \bigr)} \right) + \left( \langle Q_{a} \rangle_{\hat{\rho} \bigl( s, t \bigr)} - \langle Q_{a} \rangle_{\rho \left( s, t \right)} \right) \right\vert \\
        &\leqslant \underset{\leqslant \varepsilon_{\text{shadow}}}{\underbrace{\left\vert \langle \hat{Q}_{a} \rangle_{\tilde{\rho} \left( s, t \right)} - \langle Q_{a} \rangle_{\tilde{\rho} \left( s, t \right)} \right\vert}} + \underset{\leqslant \left\Vert Q_{a} \right\Vert_{\infty} \left\Vert \tilde{\rho} \left( s, t \right) - \hat{\rho} \left( s, t \right) \right\Vert_{1} \leqslant 2\varepsilon_{\text{sim}}}{\underbrace{\left\vert \langle Q_{a} \rangle_{\tilde{\rho} \bigl( s, t \bigr)} - \langle Q_{a} \rangle_{\hat{\rho} \bigl( s, t \bigr)} \right\vert}} + \underset{\leqslant \varepsilon_{\text{exp}}}{\underbrace{\left\vert \langle Q_{a} \rangle_{\hat{\rho} \bigl( s, t \bigr)} - \langle Q_{a} \rangle_{\rho \left( s, t \right)} \right\vert}} \leqslant \varepsilon,
        \nonumber
    \end{align}
    so that we have $\varepsilon_{\text{shadow}} + \varepsilon_{\text{sim}} + \varepsilon_{\text{exp}} \leqslant \varepsilon$. We set $\varepsilon_{\text{shadow}} \leftarrow \varepsilon/3, \varepsilon_{\text{sim}} \leftarrow \varepsilon/6$, and $\varepsilon_{\text{exp}} \leftarrow \varepsilon/3$, where the first and third each holds true with probability $1 - \delta/2$. Recall again that, here, we are allocating the leftover failure probability of $\delta/2$ from the training stage. Then, using Theorem 1 of Ref.~\cite{Shadows}, we require 
    \begin{align}
        m &= 2 \log \left( \frac{4M}{\delta} \right), \\
        \text{and}\qquad G &= \frac{136}{\varepsilon^{2}} \underset{a \in \left[ M \right]}{\max} \hspace{-0.75cm}\underset{\hspace{1cm}\leqslant 4 \cdot 3^{\left\vert\text{supp}\left( E_{a} \right)  \right\vert}~\text{(for Pauli measurements)}}{\underbrace{\left\Vert Q_{a} - \frac{\operatorname{Tr} \left[ Q_{a} \right]}{2^{n}} \mathbbm{1} \right\Vert_{\text{shadow}}^{2}}}, \\
        \text{so that}\qquad\tilde{N} &\leqslant Gm \in \mathcal{O} \left( \frac{1}{\varepsilon^{2}} \log \left( \frac{M}{\delta} \right) \right),
    \end{align}
    where $G$ and $m$ are the parameters of the median-of-means estimation procedure, as in Subroutine \ref{alg:subroutine-median-of-means-classical-shadows}.
\end{proof}

\begin{theorem}[Quantum learnability, formal]
\label{thm:quantum-learnability-formal}
    The concept class of Def.~\ref{def:learning_problem_concept_class} is efficiently quantumly PAC-learnable; that is, Algorithms \ref{alg:training-via-ham-learning-short-times} and \ref{alg:inference-via-hamiltonian-simulation} require a total sample complexity
    \begin{align}
        \label{eq:final-quantum-learnability-thm-sample-complexity}
        N \in \mathcal{O} \left( \frac{M^{5} T^{6} \left\Vert O \right\Vert_{\infty}^{3} \left( \mathfrak{d} + 1 \right)^{6}}{\varepsilon^{5}} \log \left( \frac{M}{\delta} \right) \right),
    \end{align}
    and a runtime of $\mathcal{O} \left( MN \right)$ and $\mathcal{O} \bigl( \tilde{N} \bigl( M + t_{\text{sim}} \bigr) \bigr)$, where $t_{\text{sim}} \in \mathcal{O} \left( \mathsf{poly} \left( n, 1/\varepsilon \right) \right)$, in order to satisfy the learning condition of Def.~\ref{def: learning condition}.
\end{theorem}

\begin{proof}
    The sample complexity shown in Eq.\ (\ref{eq:final-quantum-learnability-thm-sample-complexity}) follows from Lemma \ref{lem:tot-sample-complexity-analysis}. Let us now thoroughly analyze the time complexities of Algorithms \ref{alg:training-via-ham-learning-short-times} and \ref{alg:inference-via-hamiltonian-simulation}. For each $a \in \left[ M \right]$ and $j \in \left[ N \right]$, Algorithm \ref{alg:training-via-ham-learning-short-times} involves checking whether $t^{(j)} \in \left[ 0, t^{*} \right]$, $x_{1}^{(j)} = 0$, and $P_{a} \lvert s_{i}^{(j)} \rangle = + \lvert s_{i}^{(j)} \rangle$. Each of these checks is $\mathcal{O} \left( 1 \right)$; the Pauli eigenstate check being constant-time follows from the fact that $P_{a}$ is a fixed single qubit Pauli and $s^{(j)}_{i}$ can only be one of six symbols. This is followed by a $\mathcal{O} \left( M \right)$ cost of performing elementary arithmetic operations to obtain the estimates $\bigl( \hat{\lambda}_{a} \bigr)_{a \in \left[ M \right]}$. Therefore, the time complexity of training is dominated by $\mathcal{O} \left( MN \right)$. We move on to Algorithm \ref{alg:inference-via-hamiltonian-simulation}, the first step of which requires us to prepare, for each $\ell \in \bigl[ \tilde{N} \bigr]$, the state $\rho_{\text{init}}^{(\ell)} \left( s \right)$, which can be done in $\mathcal{O} \left( n \right)$ time and $\mathcal{O} \left( 1 \right)$ depth. Then, we perform a time-evolution step, which, abstractly and combined with the previous initial state preparation step, incurs a cost of $t_{\text{sim}} \in \mathcal{O} \left( \mathsf{poly} \left( n, t, 1/\varepsilon_{\text{sim}} \right) \right) \subseteq \mathcal{O} \left( \mathsf{poly} \left( n, 1/\varepsilon \right) \right)$ per copy. Finally, the runtime of the classical post-processing is given by $\mathcal{O} \bigl( \tilde{N} M \bigr)$. The total time complexity of inference is then given as $\mathcal{O} \bigl( \tilde{N}M + \tilde{N} t_{\text{sim}} \bigr)$.
\end{proof}

\subsection{Quantum learnability of the classically hard instance}
\begin{lemma}[Sample- and time-complexity of learning $\mathcal{C}_{U_{\text{enc}}, T}^{\mathsf{H}_{\mathsf{BQP}}}$ using Algorithm \ref{alg:training-via-ham-learning-short-times}]
    \label{lem:samp-time-complexity-of-quantumly-learning-classically-hard-instance}
    Let $\mathsf{H}_{\mathsf{BQP}}$ be the Hamiltonian family as in Theorem \ref{thm: classical hardness}, and let $\mathcal{C}_{U_{\text{enc}}, T}^{\mathsf{H}_{\mathsf{BQP}}}$ be the corresponding concept class defined per Def.~\ref{def:learning_problem_concept_class}. Then, with probability at least $1 - \delta$, $\mathcal{C}_{U_{\text{enc}}, T}^{\mathsf{H}_{\mathsf{BQP}}}$ is PAC-learnable in the sense of Def.~\ref{def: learning condition} via Algorithm \ref{alg:training-via-ham-learning-short-times} using a number of samples scaling as 
    \begin{align}
        N_{\mathsf{BQP}} \in \mathcal{O} \left( T \cdot \frac{\left( \mathfrak{d} + 1 \right)^{6}}{\kappa^{5}} \log \left( \frac{M}{\delta} \right) \right) \subseteq \mathcal{O} \left( T \log \left( \frac{M}{\delta} \right) \right),
    \end{align}
    where, as described in Lemma \ref{lemma: quantum learnability const error} and for $\varepsilon_{P} \approx 0.66$, $\kappa \coloneqq \varepsilon_{P}'/2 = \varepsilon_{P}/64$ is the minimum estimation precision required to learn $H_{\mathsf{BQP}} \left( \lambda \right)$.
\end{lemma}

\begin{proof}
Recall, from Lemma \ref{lem:bounds-time-averaged-integral}, that we have 
\begin{align}
    \left\vert \frac{\tilde{\mathcal{F}}_{a, t}}{2t^{*}} - \lambda_{a} \right\vert \leqslant \frac{2 \left( \mathfrak{d} + 1 \right)^{2} t^{*}}{1 - 2 \left( \mathfrak{d} + 1 \right) t^{*}}.
\end{align}
Let us denote by $t^{*}_{\mathsf{BQP}}$ the short-time cutoff for the classically hard case. We then require that 
\begin{align}
    \frac{2 \left( \mathfrak{d} + 1 \right)^{2} t^{*}_{\mathsf{BQP}}}{1 - 2 \left( \mathfrak{d} + 1 \right) t^{*}} \leqslant \frac{\kappa}{2} &\implies 4 \left( \mathfrak{d} + 1 \right)^{2} t^{*}_{\mathsf{BQP}} \leqslant \kappa - 2 \left( \mathfrak{d} + 1 \right) \kappa t^{*}_{\mathsf{BQP}} 
    \\
    \nonumber
    &\implies \left\{ 4 \left( \mathfrak{d} + 1 \right)^{2} + 2 \left( \mathfrak{d} + 1 \right) \kappa \right\} t^{*}_{\mathsf{BQP}} \leqslant \kappa \\
     \nonumber
    &\implies t^{*}_{\mathsf{BQP}} \leqslant \frac{\kappa}{4 \left( \mathfrak{d} + 1 \right)^{2} + 2 \left( \mathfrak{d} + 1 \right) \kappa} \\
     \nonumber
    &\implies t^{*}_{\mathsf{BQP}} \leqslant \frac{\varepsilon_{P}/64}{4 \left( \mathfrak{d} + 1 \right)^{2} + 2 \left( \mathfrak{d} + 1 \right) \cdot \left( \varepsilon_{P}/64 \right)} \\
     \nonumber
    &\implies t^{*}_{\mathsf{BQP}} \leqslant \frac{66}{25600 \left( \mathfrak{d} + 1 \right)^{2} + 132 \left( \mathfrak{d} + 1 \right)} \in \mathcal{O} \left( \frac{1}{\left( \mathfrak{d} + 1 \right)^{2}} \right)  \in \mathcal{O} \left( 1 \right)
     \nonumber
\end{align}
holds true, with probability at least $1 - \delta/2$.
On the other hand, we also require 
\begin{align}
    \frac{1}{2t^{*}_{\mathsf{BQP}}} \sqrt{\frac{288}{N_{t^{*}_{\mathsf{BQP}}}} \log \left( \frac{4M}{\delta} \right)} \leqslant \frac{\kappa}{2} &\implies \frac{1}{{4t^{*2}_{\mathsf{BQP}}}} \cdot \frac{288}{N_{t^{*}_{\mathsf{BQP}}}} \log \left( \frac{4M}{\delta} \right) \leqslant \frac{\kappa^{2}}{4} \\
    &\implies N_{t^{*}_{\mathsf{BQP}}} \geqslant \frac{288}{t^{*2}_{\mathsf{BQP}} \kappa^{2}} \log \left( \frac{4 M}{\delta} \right),
    \nonumber
\end{align}    
where $N_{t^{*}_{\mathsf{BQP}}}$ is the number of short-time samples. Now, let us compute the total number of samples, $N_{\mathsf{BQP}}$, following the logic of Lemma \ref{lem:tot-sample-complexity-analysis}. By analogy with Eq.\ (\ref{eq:tot-sample-complexity-quadratic-soln}), we will find that $N_{\mathsf{BQP}}$ satisfies
\begin{align}
    N_{\mathsf{BQP}} \geqslant \frac{2T}{t^{*}_{\mathsf{BQP}}} \left[ N_{t^{*}_{\mathsf{BQP}}} + \log \left( \frac{1}{\delta_{\mathsf{BQP}}} \right) + \sqrt{\log \left( \frac{1}{\delta_{\mathsf{BQP}}} \right) \left( 2N_{t^{*}_{\mathsf{BQP}}} + \log \left( \frac{1}{\delta_{\mathsf{BQP}}} \right) \right)}~\right] \in \mathcal{O} \left( \frac{T}{t^{*}} N_{t^{*}_{\mathsf{BQP}}} \right).
\end{align}
Plugging in the previously acquired expressions for $t^{*}$ and $N_{t^{*}_{\mathsf{BQP}}}$, we finally obtain
\begin{align}
    N_{\mathsf{BQP}} \in \mathcal{O} \left( T \cdot \frac{\left( \mathfrak{d} + 1 \right)^{6}}{\kappa^{5}} \log \left( \frac{M}{\delta} \right) \right) \in \mathcal{O} \left( T \log \left( \frac{M}{\delta} \right) \right),
\end{align}
where we have used the facts that $\mathfrak{d} \in \mathcal{O} \left( 1 \right)$ and $\kappa = \varepsilon_{P}/64 \in \mathcal{O} \left( 1 \right)$. The time-complexity of reconstructing $H_{\mathsf{BQP}}$ is then simply given as $\mathcal{O} \left( MN_{\mathsf{BQP}} \right)$.
\end{proof}

\end{document}